\documentclass[10pt,a4paper]{article}

\usepackage{hyperref}  

\usepackage{cite}
\usepackage{enumitem}

\usepackage[T1]{fontenc}    
\usepackage{lmodern}
\usepackage{times}

\usepackage{pgf}
\usepackage{tikz}

\usepackage{doi}
\usepackage{cite}

\usepackage[fleqn]{amsmath}
\usepackage{amsfonts,amsthm,amssymb}

\numberwithin{equation}{section}

\newtheorem{theorem}{Theorem}[section]
\newtheorem{proposition}{Proposition}[section]

\newtheorem{identity}{Identity}
\newtheorem{prop}{Proposition}[section]

{\theoremstyle{remark}
\newtheorem{rem}{\sl Remark}

\makeatletter
\newcommand{\vast}{\bBigg@{3}}
\newcommand{\Vast}{\bBigg@{4}}
\makeatother

\newcommand{\bra}[1]{\langle\,#1\,|}
\newcommand{\ket}[1]{|\,#1\,\rangle}

\newcommand{\moy}[1]{\langle\,#1\,\rangle}

\def\tr{\operatorname{tr}}
\newcommand{\End}{\operatorname{End}}

\DeclareMathOperator*{\Motimes}{\text{\raisebox{0.25ex}{\scalebox{0.8}{$\bigotimes$}}}}

\DeclareMathSymbol{\Alpha}{\mathalpha}{operators}{"41}
\DeclareMathSymbol{\Beta}{\mathalpha}{operators}{"42}
\DeclareMathSymbol{\Epsilon}{\mathalpha}{operators}{"45}
\DeclareMathSymbol{\Zeta}{\mathalpha}{operators}{"5A}
\DeclareMathSymbol{\Eta}{\mathalpha}{operators}{"48}
\DeclareMathSymbol{\Iota}{\mathalpha}{operators}{"49}
\DeclareMathSymbol{\Kappa}{\mathalpha}{operators}{"4B}
\DeclareMathSymbol{\Mu}{\mathalpha}{operators}{"4D}
\DeclareMathSymbol{\Nu}{\mathalpha}{operators}{"4E}
\DeclareMathSymbol{\Omicron}{\mathalpha}{operators}{"4F}
\DeclareMathSymbol{\Rho}{\mathalpha}{operators}{"50}
\DeclareMathSymbol{\Tau}{\mathalpha}{operators}{"54}
\DeclareMathSymbol{\Chi}{\mathalpha}{operators}{"58}
\DeclareMathSymbol{\omicron}{\mathord}{letters}{"6F}

\DeclareMathOperator{\thd}{\vartheta}
\DeclareMathOperator{\ths}{\theta}

\newcommand{\mathsc}[1]{{\normalfont\textsc{#1}}}

\def\barE{\underline{\mathbf E}}

\newcommand\restrict[1]{\raisebox{-.5ex}{$|$}_{#1}}

\begin{document}

\begin{flushright}
LPENSL-TH-07/25
\end{flushright}

\bigskip

\bigskip

\begin{center}



\textbf{\Large On correlation functions of the open XYZ spin 1/2 chain with boundary fields related by a constraint} \vspace{45pt}
\end{center}

\begin{center}
{\large \textbf{G. Niccoli}\footnote{{Univ Lyon, Ens de Lyon, Univ
Claude Bernard, CNRS, Laboratoire de Physique, F-69342 Lyon, France;
giuliano.niccoli@ens-lyon.fr}} }, {\large \textbf{V. Terras}\footnote{{Université Paris-Saclay, CNRS,  LPTMS, 91405, Orsay, France; veronique.terras@universite-paris-saclay.fr}} }
\end{center}

\begin{center}
\vspace{45pt}

\today
\vspace{45pt}
\end{center}

\begin{abstract}

In this paper, we consider the quantum XYZ open spin-1/2 chain with boundary fields. We focus on the particular case in which the six boundary parameters are related by a single constraint enabling us to describe part of the spectrum by standard Bethe equations. We derive for this model exact representations for a set of elementary blocks of correlation functions, hence generalising to XYZ  the results obtained in the XXZ open case in \cite{NicT23}.
Our approach is also similar to the approach proposed in the XXZ case \cite{NicT23}:
we solve the model by Sklyanin's version of the quantum Separation of Variables, using Baxter's Vertex-IRF transformation;
in this framework, we identify a basis of local operators with a relatively simple action on the transfer matrix eigenstates; we then use the solution of the quantum inverse problem and our recent formulae on scalar products of separate states \cite{NicT24} to compute some of the corresponding matrix elements, which can therefore be considered as elementary building blocks for the correlation functions.
The latter are expressed in terms of multiple sums in the finite chain, and as multiple integrals in the thermodynamic limit.
Our results evidence that, once the basis of local operators is properly chosen, the corresponding building blocks for correlation functions have a similar structure in the XXX/XXZ/XYZ open chains, and do not require any insertion of non-local “tail operators”.


\end{abstract}

\parbox{12cm}{\small }\newpage

\tableofcontents
\newpage

\section{Introduction}

In this paper, we extend our approach \cite{NicPT20,Nic21,NicT22,NicT23} to compute correlation functions in the framework of the quantum Separation of Variables (SoV) \cite{Skl85,Skl85a,Skl90,Skl92,Skl95,Skl96,BabBS96,Smi98a,Smi01,DerKM01,DerKM03,DerKM03b,BytT06,vonGIPS06,FraSW08,AmiFOW10,NicT10,Nic10a,Nic11,FraGSW11,GroMN12,GroN12,Nic12,Nic13,Nic13a,Nic13b,GroMN14,FalN14,FalKN14,KitMN14,NicT15,LevNT16,NicT16,KitMNT16,JiaKKS16,KitMNT17,MaiNP17,KitMNT18}, from the XXX/XXZ spin chains to the case of the XYZ open spin chain with the most general integrable boundary conditions up to one boundary constraint, as described in \cite{NicT24}. 

Let us recall that the quantum XYZ spin 1/2 chain has historically represented a more involved integrable quantum model, in particular with respect to the non-completely anisotropic XXX and XXZ Heisenberg chains. In fact, the Bethe Ansatz approach \cite{Bet31} for the spectrum cannot be directly applied to this model due to the non-conservation of the spin. 
Forty years after Bethe's work, a breakthrough was realised by Baxter 
\cite{Bax71a,Bax71b,Bax72,Bax73a,Bax76,Bax77,Bax82L,Bax02,Bax04}, who managed to describe the eigenvalues and eigenstates of the eight-vertex model, and of the associated XYZ spin chain,  in terms of those of the Interaction-Round-a-Face (IRF) elliptic solid-on-solid (SOS) model, by introducing the so-called Vertex-IRF transformation. Indeed, for the latter model, the Bethe Ansatz spectrum analysis can be done both in its Coordinate \cite{Bax73a,AlcBBBQ87,Bax02} and Algebraic (ABA) \cite{FadT79,FelV96a,FelV96b,Fel95} version. Nevertheless, Baxter's solution produces representations for the XYZ eigenstates which are much more complicated than in the XXX/XXZ cases: these eigenstates are written as integrals of non-local transformations of elliptic SOS Bethe eigenstates, which reduce to finite sums only in the cyclic cases \cite{FadT79,FelV96a,FelV96b,Fel95}.

Remarkable progresses on the exact computation of correlation functions of quantum integrable models, and in particular of quantum integrable spin chains, have been achieved in the last decades. The first results were obtained by the Kyoto group \cite{JimM95L,JimMMN92,JimM96}, who computed correlation functions for the XXX/XXZ spin 1/2 quantum chains directly in infinite volume, at zero temperature and at zero magnetic field, in the framework of the $q$-vertex operator approach. A few years later, the Lyon group developed, instead, a genuine finite lattice approach to compute correlation functions from the study of the model in finite volume in the framework of the Quantum Inverse Scattering Method (QISM) \cite{KitMT99,MaiT00,KitMT00,KitMST02a,KitMST05a,KitMST05b,KitKMST07}. The  main ingredients of the latter approach are the simple monomial form of the Bethe states in the ABA approach \cite{FadS78,FadST79}, the resolution of the quantum inverse problem \cite{KitMT99,MaiT00,GohK00} (which allows to compute the action of local operators on Bethe states by Yang-Baxter commutation relations), and simple determinant representations for the scalar products of Bethe states \cite{Sla89,KitMT99}. 
This method has been developed for the periodic and some specific open boundary conditions for the XXX/XXZ spin-1/2 Heisenberg chain  \cite{KitMT99,MaiT00,KitMT00,KitMST02a,KitMST05a,KitMST05b,KitKMST07,KitKMNST07,KitKMNST08}, for some higher spin chains \cite{Kit01,Deg12}, for the cyclic SOS model \cite{LevT13a,LevT13b,LevT14a}, and also for the quantum non-linear Schr\"odinger model \cite{KitKMST07}.
On one side, these correlation functions for the periodic case have confirmed the results of the $q$-vertex approach\footnote{It is also worth mentioning the interesting results on correlation functions obtained in \cite{BooJMST05,BooJMST06,BooJMST06a,BooJMST06b,BooJMST06c,BooJMST07,BooJMST09,JimMS09,JimMS11,MesP14,Poz17} through the study of hidden Grassmann structure.} and generalized them to the case of a non-zero global magnetic field.
On the other side, they have been at the basis of a breakthrough on the analytical study of long distances two-point and multi-point correlation functions \cite{KitKMST09a,KitKMST09b,KitKMST09c,KozMS11a,KozMS11b,KozT11,KitKMST11a,KitKMST11b,KitKMST12,DugGK13,KitKMT14} as well as of comparison with experimental settings like neutron scattering \cite{KenCTVHST02} thanks to numerical study of the dynamical structure factors in \cite{CauHM05,CauM05,PerSCHMWA06}. Generalisations of these results to non-zero-temperature correlation functions were also obtained by the Wuppertal group \cite{GohKS04,GohKS05,BooGKS07,GohS10,GohSS10,DugGKS15,GohKKKS17} in the framework of the so-called quantum transfer matrix approach.

Despite these progresses, the study of the correlation functions of the completely anisotropic XYZ spin chain is still a challenging problem. On the one hand, some results were obtained, directly in the infinite volume limit, in the $q$-vertex operator approach, in \cite{LasP98,LasP98b,Har00,Las02}: there, the XYZ correlation functions were expressed in terms of those of the SOS model by using the free boson description of the SOS models \cite{LukP96} and by means of Baxter's Vertex-IRF transformation; the latter results into the appearance of  non-local “tail operators”, which complexifies considerably the obtention of generic final expressions for the XYZ correlation functions.
On the other hand,  the study of the correlation functions of the finite size XYZ quantum spin 1/2 chain has so far remained out of reach of the Lyon group method based on ABA, even if the resolution of the inverse problem has been known for a long time \cite{MaiT00}, and despite the recent derivation of Slavnov’s type determinant representations for the scalar products of Bethe states  in the closed chain under periodic boundary conditions \cite{SlaZZ20,KulS23}\footnote{See nevertheless the study of \cite{KulS24} in the limit of free fermions.}. The central reason for it is the aforementioned complexity of the XYZ eigenstate representations due to the Vertex-IRF transformation.

In our previous papers  \cite{NicPT20,Nic21,NicT22,NicT23}, in the framework of the SoV approach \cite{Skl85,MaiN18}, we were able to develop a method that allows to compute correlation functions of the XXX and XXZ spin 1/2 chains under general integrable boundary conditions, which are not accessible within ABA. Here, we extend our method to the XYZ case, overcoming the  challenge to define a genuine finite lattice method for this type of XYZ spin chains, and we do this directly in the case of general integrable open boundary conditions.
Indeed, the study of spin chains (XXX, XXZ or XYZ) with general integrable boundary conditions is a longstanding problem which has attracted a large research activity \cite{AlcBBBQ87,Skl88,GhoZ94,JimKKKM95,JimKKMW95,FanHSY96,Nep02,Nep04,CaoLSW03,YanZ07,Bas06,BasK07,BasB13,BasB17,KitKMNST07,KitKMNST08,CraRS10,CraRS11,FilK11,FraSW08,FraGSW11,Nic12,CaoYSW13b,FalN14,FalKN14,KitMN14,BelC13,Bel15,BelP15,BelP15b,AvaBGP15,BelP16,GriDT19,QiaCYSW21,XinCYW24}. This problem is also related to the study of classical stochastic models, like asymmetric simple exclusion models \cite{deGE05}, and to out-of-equilibrium and transport properties in the spin chains, see e.g. \cite{Pro11} and related references. 

Let us recall that the quantum Separation of Variables is a non-Ansatz exact method, first developed by Sklyanin \cite{Skl85,Skl85a,Skl90,Skl92,Skl95,Skl96}, in the framework of Quantum Inverse Scattering Method \cite{FadS78,FadST79,FadT79,Skl79,Skl79a,FadT81,Skl82,Fad82,Fad96,BogIK93L}. It has recently  been reinvented relying only on the integrable structure of the models, i.e. its abelian algebra of conserved charges and fusion relations. The SoV approach has proven to have a large range of application to various integrable quantum models \cite{Skl85,Skl85a,Skl90,Skl92,Skl95,Skl96,BabBS96,Smi98a,Smi01,DerKM01,DerKM03,DerKM03b,BytT06,vonGIPS06,FraSW08,AmiFOW10,NicT10,Nic10a,Nic11,FraGSW11,GroMN12,GroN12,Nic12,Nic13,Nic13a,Nic13b,GroMN14,FalN14,FalKN14,KitMN14,NicT15,LevNT16,NicT16,KitMNT16,JiaKKS16,KitMNT17,MaiNP17,KitMNT18,RyaV19,MaiN19,MaiN19d,MaiN19b,MaiN19c,MaiNV20,RyaV20}, leading by construction to a complete characterization of the spectrum (eigenvalues and eigenstates). Moreover, it overcomes the restricted applicability of the ABA approach, then representing the natural framework in which to extend the original Lyon group approach beyond ABA.

From a more technical point of view,  the choices here made to study the open XYZ chain and to use the SoV approach are at the basis of our ability to compute the action of local operators on transfer matrix eigenstates and then to compute their matrix elements. In fact, the SoV approach has been already developed for the open XYZ spin chain \cite{FalN13,NicT24} and it has proven to lead to simple representations of its spectrum and of the scalar product of the class of the so-called separate states, which contains as special instances the transfer matrix eigenstates. 

More in details, the SoV representation of each eigenstate can intrinsically be rewritten in a boundary Bethe state (i.e. ABA-like) form, as a monomial of gauged $B$-operators acting on a defined “reference state”. While this construction will be detailed in the bulk of the paper, here, we want to stress that the “reference states” are simple tensor product states and the gauged $B$-operators are simple linear combinations of the elements of the boundary monodromy matrix.
This simplicity is central to derive the first fundamental result of this paper, i.e. the extension to the XYZ case of the so-called boundary-bulk decomposition of boundary Bethe states. By using it and the known solution of the quantum inverse problem, we are able to obtain our second fundamental result, i.e. to define a basis of local operators for which we can easily compute their action on boundary Bethe states, so also on transfer matrix eigenstates, as a multiple sum of boundary Bethe states. 
Then, the restrictions on the boundary parameters (one-constraint) and on the set of local operators (conservation of the number of gauged $B$-operators in the boundary Bethe states) are here imposed to use our rewriting \cite{NicT24} of the SoV scalar products of separate states as some generalized Slavnov’s type determinant, which is the third fundamental result needed to compute correlation functions. By using it, we derive multiple sum representations for the corresponding building blocks for correlation functions on the finite lattice, and then, taking the thermodynamic limit, we obtain multiple integral representations.

The paper is organized in the following way. In Section~\ref{sec-model}, we recall the general algebraic framework for the quantum XYZ open spin 1/2 chain. In Section~\ref{sec-gauge}, we  introduce the Vertex-IRF transformation, the gauge transformed version of the boundary monodromy matrix, their algebra and the representation of the transfer matrix in terms of them as well as their boundary-bulk decomposition. In Section~\ref{sec-SoV}, we define the SoV-basis as the pseudo-eigenbasis of the gauged $\mathcal{B}_-$ boundary operators and explain how to characterize the transfer matrix spectrum and eigenstates in this SoV basis. We in particular discuss the Bethe-like rewriting of the separate states and the functional representation of the spectrum by ordinary $TQ$-equation in the case with one constraint. In Section~\ref{sec-act}, we derive two fundamental tools for the computation of correlation functions, i.e. the boundary-bulk decomposition of gauged Bethe-like boundary states and the action on them of a basis of local operators. In Section~\ref{sec-corr}, we compute the matrix elements on the finite chain, and we present our main results on the infinite chain. In Section~\ref{sec-conclusion}, we put forward some conclusions and outlooks.
Some technical details (properties of the bulk gauged Yang-Baxter generators, trigonometric limit) are gathered in Appendices.

\section{The XYZ open spin chain: general algebraic framework}
\label{sec-model}

We consider here the XYZ open spin-1/2 chain with boundary fields. The Hamiltonian of this model is given by
\begin{equation}\label{Ham}
   H=
   \sum_{a\in\{x,y,z\}}\!\Bigg[\sum_{n=1}^N  J_a \sigma_n^a\sigma_{n+1}^a +h_+^a\sigma_1^a+h_-^a\sigma_N^a\Bigg],
\end{equation}
in which $\sigma_n^{x,y,z}\in\End\mathcal{H}_n$ denote the usual Pauli matrices, acting as local spin-1/2 operators on the local quantum space $\mathcal{H}_n=\mathbb{C}^2$ at site $n$ of the chain. We set $\mathcal H=\otimes_{n=1}^N\mathcal{H}_n$.
$J_x,J_y,J_z$ are three different coupling constants, parametrized as
\begin{equation}\label{coupling}
   J_x=\frac{\ths_4(\eta)}{\ths_4(0)}, \qquad
   J_y=\frac{\ths_3(\eta)}{\ths_3(0)},\qquad
   J_z=\frac{\ths_2(\eta)}{\ths_2(0)}, 
\end{equation}
in terms of two complex parameters, the so-called crossing parameter $\eta$ and an elliptic parameter $\omega$ such that $\Im(\omega) >0$. We shall moreover suppose that $\mathbb{Z}\eta\cap (\mathbb{Z}\pi+\mathbb{Z}\pi\omega)=\emptyset$.
Here and in the following, 
the functions $\ths_j(\lambda)\equiv\theta_j(\lambda|\omega)$, $j=1,2,3,4$, denote the usual theta functions \cite{GraR07L} with quasi-periods $\pi$ and $\pi\omega$, and $\ths(\lambda)\equiv\ths_1(\lambda)$. The spins at site 1 and at site $N$ are coupled to some boundary fields $h_+$ and $h_-$. We parametrize the components  $h_\pm^x,\ h_\pm^y,\ h_\pm^z$ of these boundary fields as
\begin{alignat}{2}
   &h_\pm^x=c^x_\pm \,\frac{\ths_1(\eta)}{\ths_4(0)}
   \qquad & &\text{with} \quad 
   c^x_\pm=\prod_{\ell=1}^3\frac{\ths_4(\alpha_\ell^\pm)}{\ths_1(\alpha_\ell^\pm)},
   \quad \label{h^x}\\
   &h_\pm^y=i c^y_\pm \,\frac{\ths_1(\eta)}{\ths_3(0)}
   \qquad & &\text{with}\quad 
   c^y_\pm=- \prod_{\ell=1}^3\frac{\ths_3(\alpha_\ell^\pm)}{\ths_1(\alpha_\ell^\pm)},
   \quad \\
   &h_\pm^z=c^z_\pm \,\frac{\ths_1(\eta)}{\ths_2(0)}
   \qquad & &\text{with}\quad  
   c^z_\pm=\prod_{\ell=1}^3\frac{\ths_2(\alpha_\ell^\pm)}{\ths_1(\alpha_\ell^\pm)}, \label{h^z}
\end{alignat}
in terms of six additional parameters $\alpha_\ell^\sigma$, $\ell=1,2,3$, $\sigma=\pm$, that we shall call boundary parameters.

The problem of the diagonalization of the Hamiltonian \eqref{Ham} can be formulated in the algebraic framework of the Quantum Inverse Scattering Method \cite{FadST79}, or more precisely in the boundary version proposed in \cite{Skl88}.
It relies on the R-matrix $R(\lambda)$ of the eight-vertex model ($\lambda$ being the so-called spectral parameter), and on the elliptic scalar boundary K-matrix $K(\lambda)$ found in \cite{InaK94,HouSFY95}, satisfying with $R(\lambda)$ the reflection equation formulated in \cite{Che84}.
The eight-vertex R-matrix is the elliptic solution of the Yang-Baxter equation, given as
\begin{align}\label{R-mat}
R(\lambda )
&=
\begin{pmatrix}
  \mathrm{a}(\lambda) & 0 & 0 & \mathrm{d}(\lambda ) \\ 
0 & \mathrm{b}(\lambda ) & \mathrm{c}(\lambda ) & 0 \\ 
0 & \mathrm{c}(\lambda ) & \mathrm{b}(\lambda ) & 0 \\ 
\mathrm{d}(\lambda ) & 0 & 0 & \mathrm{a}(\lambda )
\end{pmatrix}
\in\End(\mathbb{C}^2\otimes\mathbb{C}^2),
\end{align}
with
\begin{alignat}{2}
 &\mathrm{a}(\lambda )
  =\frac{2\thd_4(\eta)\, \thd_1(\lambda+\eta )\,\thd_4(\lambda  )}{\ths_2(0)\, \thd_4(0)}, 
  & \qquad
 &\mathrm{b}(\lambda )
  =\frac{2\thd_4(\eta )\, \thd_1(\lambda)\, \thd_4(\lambda +\eta  )}  {\ths_2(0)\, \thd_4(0)}, 
            \label{a-b-R} \\
 &\mathrm{c}(\lambda ) 
  =\frac{2\thd_1(\eta  )\, \thd_4(\lambda )\, \thd_4(\lambda +\eta  )} {\ths_2(0)\, \thd_4(0)}, 
  &
  &\mathrm{d}(\lambda )
   =\frac{2\thd_1(\eta )\,\thd_1(\lambda +\eta )\, \thd_1(\lambda  )} {\ths_2(0 )\, \thd_{4}(0)}.
              \label{c-d-R}
\end{alignat}
Here $\thd_j(\lambda)\equiv\theta_j(\lambda|2\omega)$, $j=1,2,3,4$, denote the theta functions with quasi-periods $\pi$ and $2\pi\omega$.
The corresponding scalar boundary K-matrix $K(\lambda)$ is \cite{InaK94,HouSFY95},
\begin{equation}\label{mat-K}
  K(\lambda)=\frac{\ths_1(2\lambda)}{2\ths_1(\lambda)}
   \left[ \mathbb{I}+c^x\frac{\ths_1(\lambda)}{\ths_4(\lambda)}\, \sigma^x+i c^y\frac{\ths_1(\lambda)}{\ths_3(\lambda)}\, \sigma^y+c^z\frac{\ths_1(\lambda)}{\ths_2(\lambda)}\, \sigma^z\right]
   \in\End(\mathbb{C}^2),
\end{equation}
where $c^x, c^y, c^z$ are some (a priori arbitrary) coefficients. In the following, we consider two such boundary K-matrices, encoding the components of the two boundary fields $h_-$ and $h_+$ in \eqref{Ham}-\eqref{h^z}:
\begin{equation}\label{K-}
   K_-(\lambda)=K(\lambda-\eta/2;\alpha_1^-,\alpha_2^-,\alpha_3^-)
\end{equation}
is the matrix \eqref{mat-K} in which $\lambda$ has been replaced by $\lambda-\eta/2$ and $c^x, c^y, c^z$ correspond to the coefficients $c^x_-,c^y_-,c^z_-$ given in terms of the boundary parameters $\alpha_\ell^-$, $1\le \ell \le 3$, as in \eqref{h^x}-\eqref{h^z};
\begin{equation}\label{K+}
   K_+(\lambda)=K(\lambda+\eta/2;\alpha_1^+,\alpha_2^+,\alpha_3^+)
\end{equation}
is the matrix \eqref{mat-K} in which $\lambda$ has been replaced by $\lambda+\eta/2$ and $c^x, c^y, c^z$ correspond to the coefficients $c^x_+,c^y_+,c^z_+$ given in terms of the boundary parameters $\alpha_\ell^+$, $1\le \ell \le 3$, as in \eqref{h^x}-\eqref{h^z}.

Let $\mathcal{H}_0$ be a 2-dimensional auxiliary space. As usual, we define the bulk monodromy matrix $T(\lambda)\in \End(\mathcal{H}_0\otimes\mathcal{H})$ as the following product of R-matrices\footnote{We have used here the same order as in our previous articles \cite{NicT22,NicT23,NicT24}.}, 
\begin{equation}\label{mat-T}
  T(\lambda)=R_{01}(\lambda-\xi_1-\eta/2)\ldots R_{0N}(\lambda-\xi_N-\eta/2) 
\end{equation}
in which $R_{0n}(\lambda)\in\End(\mathcal{H}_0\otimes\mathcal{H}_n)$. We have moreover introduced arbitrary inhomogeneity parameters $\xi_1,\ldots,\xi_N$. Following \cite{Skl88}, we can then construct the boundary monodromy matrix  $\mathcal{U}_-(\lambda)\in \End(\mathcal{H}_0\otimes\mathcal{H})$ as
\begin{equation}\label{def-U-}
  \mathcal{U}_-(\lambda)=T(\lambda)\, K_-(\lambda)\, \hat T(\lambda) ,
\end{equation}
and the boundary transfer matrix $\mathcal{T}(\lambda)\in \End \mathcal{H} $ as
\begin{equation}\label{transfer}
  \mathcal{T}(\lambda)
  =\tr_{\mathcal{H}_0}\left[ K_+(\lambda)\, T(\lambda)\, K_-(\lambda)\,\hat T(\lambda)\right]
  =\tr_{\mathcal{H}_0}\left[ K_+(\lambda)\,\mathcal{U}_-(\lambda)\right],
\end{equation}
in which we have considered that the boundary K-matrices $K_\pm(\lambda)$ in \eqref{def-U-} and in \eqref{transfer} act on $\mathcal{H}_0$.
In \eqref{def-U-}-\eqref{transfer}, we have also defined
\begin{equation}\label{mat-That}
   \hat T(\lambda)
   =(-1)^N \sigma_0^y\, T^{t_0}(-\lambda)\, \sigma_0^y
   =R_{0N}(\lambda+\xi_N-\eta/2)\ldots R_{01}(\lambda+\xi_1-\eta/2).
\end{equation}
The boundary monodromy matrix $\mathcal{U}_-(\lambda)$ can be conveniently written as a $2\times 2$ matrix on the auxiliary space $\mathcal{H}_0$, with entries being quantum operators acting on $\mathcal{H}$:
\begin{equation}\label{mat-U-}
  \mathcal{U}_-(\lambda)=\begin{pmatrix} \mathcal{A}_-(\lambda) & \mathcal{B}_-(\lambda) \\ \mathcal{C}_-(\lambda) & \mathcal{D}_-(\lambda) \end{pmatrix},
  \qquad
  \mathcal{A}_-(\lambda),\ \mathcal{B}_-(\lambda),\ \mathcal{C}_-(\lambda),\ \mathcal{D}_-(\lambda)\in\End \mathcal{H}.
\end{equation}

It follows from this construction that the boundary monodromy matrix \eqref{def-U-} satisfies the reflection equation,
\begin{equation}\label{refl-eq}
   R_{0'0}(\lambda-\mu)\, \mathcal{U}_{-,0}(\lambda)\, R_{00'}(\lambda+\mu-\eta)\, \mathcal{U}_{-,0'}(\mu) 
   = \mathcal{U}_{-,0'}(\mu)\, R_{0'0}(\lambda+\mu-\eta)\, \mathcal{U}_{-,0}(\lambda)\, R_{00'}(\lambda-\mu).
\end{equation}
This equation has to be understood on $\mathcal{H}_0\otimes\mathcal{H}_{0'}\otimes \mathcal{H}$, in which $\mathcal{H}_0$ and $\mathcal{H}_{0'}$ are two 2-dimensional auxiliary spaces: the indices in \eqref{refl-eq} indicate on which auxiliary space(s) the corresponding operator acts, and $R_{0'0}(\lambda)=P_{00'}R_{00'}(\lambda)P_{00'}$, $P_{00'}$ being the permutation operator on $\mathcal{H}_0\otimes\mathcal{H}_{0'}$. 
Hence, the operator entries $\mathcal{A}_-(\lambda),\ \mathcal{B}_-(\lambda),\ \mathcal{C}_-(\lambda),\ \mathcal{D}_-(\lambda)$ of $\mathcal{U}_-(\lambda)$ satisfy the commutation relations of the reflection algebra following from \eqref{refl-eq}.
Moreover, we have the following inversion relation
\begin{equation}\label{inv-U-}
   \mathcal{U}_-^{-1}(\lambda+\eta/2)
   =\frac{\ths(2\lambda-2\eta)}{\det_q\mathcal{U}_-(\lambda)}\,\mathcal{U}_-(-\lambda+\eta/2)
   =\frac{\widetilde{\mathcal U}_-(\lambda-\eta/2)}{ \det_q \mathcal{U}_-(\lambda)},
\end{equation}
in which $\widetilde{\mathcal U}_-$ is the 'algebraic adjunct' of $\mathcal{U}_-$,
\begin{align}\label{adjU-}
   \widetilde{\mathcal U}_-(\lambda)
    &=
    \begin{pmatrix}
     -\mathrm{c}(2\lambda)\,\mathcal{A}_-(\lambda)+\mathrm{b}(2\lambda)\,\mathcal{D}_-(\lambda) &
     -\mathrm{a}(2\lambda)\,\mathcal{B}_-(\lambda)+\mathrm{d}(2\lambda)\,\mathcal{C}_-(\lambda) \\
      -\mathrm{a}(2\lambda)\,\mathcal{C}_-(\lambda)+\mathrm{d}(2\lambda)\,\mathcal{B}_-(\lambda)   &
      -\mathrm{c}(2\lambda)\,\mathcal{D}_-(\lambda)+\mathrm{b}(2\lambda)\,\mathcal{A}_-(\lambda)
   \end{pmatrix},
\end{align}
and $\det_q \mathcal{U}_-$ its quantum determinant. The latter is given in terms of the quantum determinant of the bulk monodromy matrix $T(\lambda)$
\begin{align}\label{qdet-T}
  &{\det}_qT(\lambda) 
 =a(\lambda+\eta/2)\, d(\lambda-\eta/2), \\
 &a(\lambda)=\prod_{n=1}^N\theta(\lambda-\xi_n+\eta/2),
  \qquad
  d(\lambda)=a(\lambda-\eta)=\prod_{n=1}^N\theta(\lambda-\xi_n-\eta/2),
  \label{a-d}
\end{align}
and of the boundary matrix $K_-(\lambda)$,
\begin{equation}\label{qdet-K}
  {\det}_q K_-(\lambda)=\ths(2\lambda-2\eta)\, \prod_{\ell=1}^3 \frac{\ths(\alpha_\ell^--\lambda)\,\ths(\alpha_\ell^-+\lambda)}{\ths^2(\alpha_\ell^-)}.
\end{equation}
as
\begin{equation}
    {\det}_q\,\mathcal{U}_-(\lambda)={\det}_q T(\lambda)\,{\det}_q T(-\lambda)\,{\det}_qK_-(\lambda).
\end{equation}
Finally, the monodromy matrix \eqref{def-U-} satisfies the quasi-periodicity relations
\begin{align}
   &\mathcal{U}_-(\lambda+\pi)=-\sigma^z\, \mathcal{U}_-(\lambda)\, \sigma^z, \label{q-per-U1}\\
   &\mathcal{U}_-(\lambda+\pi\omega)=(-e^{-2i\lambda-i\pi\omega})^{2N+3}\, e^{3i\eta} \, \sigma^x\, \mathcal{U}_-(\lambda)\, \sigma^x, \label{q-per-U2}
\end{align}
which can easily be deduced from the quasi-periodicity properties of the R and K-matrices \eqref{R-mat} and \eqref{mat-K}:
\begin{alignat}{2}
   &R_{12}(\lambda +\pi )=-\sigma_{1}^{z}\, R_{12}(\lambda )\,\sigma _{1}^z,
   &\qquad
   &R_{12}(\lambda +\pi \omega )
   =-e^{-2i\lambda -i\pi\omega-i\eta }\, \sigma _{1}^{x}\, R_{12}(\lambda)\, \sigma _{1}^{x},
   \label{quasi-per-R}\\
   &K(\lambda +\pi )=-\sigma^{z}\, K(\lambda )\,\sigma^z,
   &\qquad
   &K(\lambda +\pi \omega )
   =(-e^{-2i\lambda -i\pi\omega })^3\, \sigma^{x}\, K(\lambda)\, \sigma^{x},
   \label{quasi-per-K}
\end{alignat}

It also follows from this construction that the transfer matrices \eqref{transfer} commute between themselves  for different values of the spectral parameter, and that the Hamiltonian \eqref{Ham} can be obtained in terms of a logarithmic derivative of the transfer matrix \eqref{transfer} in the homogeneous limit in which all inhomogeneity parameters vanish ($\xi_n=0$, $n=1,\ldots,N$):
\begin{equation}\label{H-transfer}
  H=\frac{\ths_1(\eta)}{\ths'_1(0)}\left[ \left. \frac{d}{d\lambda}\ln\mathcal{T}(\lambda) \right|_{\substack{\lambda=\eta/2, \xi_1=\ldots=\xi_N=0}}-2\frac{\ths_1'(2\eta)}{\ths_1(2\eta)}-(N-1)\frac{\ths_1'(\eta)}{\ths_1(\eta)}\right].
\end{equation}
Hence, the problem of the diagonalization of the Hamiltonian \eqref{Ham} is reduced in this framework to the problem of the diagonalization of the transfer matrix \eqref{transfer} independently of the spectral parameter $\lambda$.

The eigenstates of the transfer matrix can be constructed in the general framework of the new Separation of Variable Approach (SoV) introduced in \cite{MaiN19}, see \cite{NicT24}. In our prospect of computing the correlation functions, it is however more convenient to choose the particular SoV basis given by the Vertex-IRF transformation of the dynamical six-vertex model \cite{FalN14} : as shown in the following, this produces separate states which can be directly related to generalised Bethe states, and on which we are able to explicitly compute the action of local operators, similarly as in \cite{NicT22,NicT23}. Therefore, in the following section, we recall the generalized gauge transformation which enabled the authors of \cite{FalN14} to generalize Sklyanin's SoV approach \cite{Skl85,Skl90} to the open XYZ case.


\section{Gauge transformation of the model}
\label{sec-gauge}

In this section, we transform the reflection algebra by means of the Vertex-IRF transformation, a generalized gauge transformation first introduced  in \cite{Bax73a1} to relate the Boltzmann weights of the eight-vertex model to those of a particular face model solvable by Bethe Ansatz. This Vertex-IRF transformation hence opened the way to the solution by generalized Bethe ansatz of the eight-vertex model (or equivalently of the XYZ spin chain) with periodic boundary conditions, see \cite{Bax73a1,FadT79}. It was later used in \cite{FanHSY96} to construct an algebraic Bethe Ansatz solution of the open XYZ chain under some constraints on the boundary parameters. Here, we shall use it to formulate, following \cite{FalN14} (but in slightly different notations), the SoV solution of the model in a generalization of Sklyanin's SoV approach \cite{Skl85,Skl90}.

\subsection{The Vertex-IRF transformation}

Let us introduce the following R-matrix,
\begin{equation}\label{mat-R6VD}
   R^{\mathsf{(6VD)}}(\lambda |\beta )
   =\begin{pmatrix}
    \mathsf{a}(\lambda ) & 0 & 0 & 0 \\ 
    0 & \mathsf{b}(\lambda |\beta ) & \mathsf{c}(\lambda |\beta ) & 0 \\ 
    0 & \mathsf{c}(\lambda |-\beta ) & \mathsf{b}(\lambda |-\beta ) & 0 \\ 
    0 & 0 & 0 & \mathsf{a}(\lambda )
    \end{pmatrix} \in \End(\mathbb{C}^2\otimes\mathbb{C}^2),
\end{equation}
with
\begin{equation}
\mathsf{a}(\lambda )=\theta (\lambda +\eta ),\quad 
\mathsf{b}(\lambda |\beta)=\frac{\theta (\lambda )\theta ((\beta +1)\eta )}{\theta (\beta \eta )},\quad 
\mathsf{c}(\lambda |\beta )=\frac{\theta (\eta )\theta (\beta \eta+\lambda )}{\theta (\beta \eta )},
\end{equation}
which satisfies a dynamical version of the Yang-Baxter equation on $\mathbb{C}^2\otimes\mathbb{C}^2\otimes\mathbb{C}^2$ \cite{GerN84,Fel95}:
\begin{multline}\label{YBdyn}
   R_{12}^{\mathsf{(6VD)}}(\lambda_1-\lambda_2 |\beta +\sigma _3^{z})\,
   R_{13}^{\mathsf{(6VD)}}(\lambda_1-\lambda_3|\beta )\,
   R_{23}^{\mathsf{(6VD)}}(\lambda_2-\lambda_3 |\beta+\sigma _1^{z})
   \\
   =R_{23}^{\mathsf{(6VD)}}(\lambda_2-\lambda_3 |\beta )\,
   R_{13}^{\mathsf{(6VD)}}(\lambda_1-\lambda_3|\beta +\sigma_2^{z})\,
   R_{12}^{\mathsf{(6VD)}}(\lambda_1-\lambda_2|\beta ).
\end{multline}
In this context, the parameter $\beta$ is usually called the dynamical parameter. It corresponds to the height parameter of the face model considered in \cite{Bax73a1}.

The relation between  the eight-vertex R-matrix \eqref{R-mat} and the dynamical R-matrix \eqref{mat-R6VD} takes the following form on $\mathbb{C}^2\otimes\mathbb{C}^2$: 
\begin{equation}\label{Vertex-IRF}
R_{12}(\lambda_1-\lambda_2)\, S_{1}(\lambda _{1}|\alpha ,\beta)\, S_{2}(\lambda _{2}|\alpha ,\beta +\sigma _1^z)
=S_{2}(\lambda_{2}|\alpha ,\beta )\, S_{1}(\lambda _{1}|\alpha ,\beta +\sigma_{2}^{z})\, R_{12}^{\mathsf{(6VD)}}(\lambda_1-\lambda_2|\beta ).
\end{equation}
The matrix $S(\lambda | \alpha,\beta)$ involved in this transformation is called the Vertex-IRF matrix. It depends on the spectral parameter $\lambda$, on the dynamical parameter $\beta$, and on an additional parameter $\alpha$ which corresponds to an arbitrary shift of the spectral parameter. This matrix can be conveniently written in terms of two vectors $Y_{\alpha+\beta}(\lambda)$ and $Y_{\alpha-\beta}(\lambda)$ as
\begin{equation}\label{mat-S}
  S(\lambda | \alpha,\beta)
  =\begin{pmatrix} Y_{\alpha+\beta}(\lambda) & Y_{\alpha-\beta}(\lambda) \end{pmatrix},
\end{equation}
in which we have defined
\begin{equation}\label{vectY}
   Y_\gamma(\lambda)=\begin{pmatrix} \thd_2(\lambda-\gamma\eta) \\ \thd_3(\lambda-\gamma\eta) \end{pmatrix}.
\end{equation}
Its inverse reads explicitly:
\begin{align}\label{Sinv}
  S^{-1}(\lambda | \alpha,\beta)
  &=\frac{1}{\det S(\lambda|\alpha,\beta)}\, S^\mathrm{Adj}(\lambda|\alpha,\beta)
  \qquad \text{with} \quad
  S^\mathrm{Adj}(\lambda|\alpha,\beta)=
  \begin{pmatrix}
  \bar Y_{\alpha-\beta}(\lambda) \\ -\bar Y_{\alpha+\beta}(\lambda)
  \end{pmatrix},
\end{align}
in which
\begin{align}\label{vectYbar}
   &\bar Y_\gamma (\lambda) 
   =\begin{pmatrix} 
   \thd_3(\lambda-\gamma\eta) & -\thd_2(\lambda-\gamma\eta)
   \end{pmatrix},\\
   &\det S(\lambda|\alpha,\beta) = \theta(\lambda-\alpha\eta)\theta(\beta\eta). \label{detS}
\end{align}
Note that the vectors \eqref{vectY} and the covectors \eqref{vectYbar} satisfy the relations
\begin{equation}\label{YbarY}
  \bar Y_\gamma(u)\, Y_\delta(v) =\ths\big(\frac{u+v-(\gamma+\delta)\eta}2\big)\ths\big(\frac{u-v-(\gamma-\delta)\eta}2\big),
\end{equation}
so that in particular
\begin{align}
   &\bar Y_{\alpha-\beta}(\lambda)\, Y_{\alpha+\beta}(\lambda)=-\bar Y_{\alpha+\beta}(\lambda)\, Y_{\alpha-\beta}(\lambda)
   =\det S(\lambda|\alpha,\beta),\\
   &\bar Y_{\alpha-\beta}(\lambda)\, Y_{\alpha-\beta}(\lambda)=\bar Y_{\alpha+\beta}(\lambda)\, Y_{\alpha+\beta}(\lambda)=0.
\end{align}
Note also that
\begin{align}\label{q-per-S}
  S(\lambda+\pi|\alpha,\beta)=-\sigma^z\, S(\lambda|\alpha,\beta),
  \qquad
  S(\lambda+\pi\omega |\alpha,\beta)= e^{-i\lambda-i\frac{\pi\omega}2+i\alpha\eta}\, \sigma^x\, S(\lambda|\alpha,\beta)\, e^{i\beta\eta\sigma^z}.
\end{align}
%

\subsection{The gauged-transformed boundary monodromy matrix}

Using the Vertex-IRF matrix \eqref{mat-S}, let us define from $\mathcal{U}_{-}(\lambda)$ \eqref{def-U-} a gauge transformed boundary monodromy matrix $\mathcal{U}_{-}(\lambda |\alpha ,\beta )$ as
\begin{align}
  \label{gauged-U}
  \mathcal{U}_{-}(\lambda |\alpha ,\beta )
  & =S_{0}^{-1}(-\lambda+\eta /2|\alpha ,\beta )\ \mathcal{U}_{-}(\lambda )\ S_{0}(\lambda -\eta /2|\alpha,\beta )  
   \notag \\
     &
 = \begin{pmatrix}
  \mathcal{A}_{-}(\lambda |\alpha ,\beta ) & \mathcal{B}_{-}(\lambda |\alpha,\beta ) \\ 
  \mathcal{C}_{-}(\lambda |\alpha ,\beta ) & \mathcal{D}_{-}(\lambda |\alpha,\beta )%
  \end{pmatrix}.
\end{align}
It satisfies,  with the R-matrix \eqref{mat-R6VD}, the dynamical version of the reflection equation,
\begin{multline}
  R_{0'0}^{\mathsf{(6VD)}}(\lambda -\mu |\beta )\,\mathcal{U}_{-,0}(\lambda|\alpha ,\beta +\sigma _{0'}^{z})\,
  R_{00'}^{\mathsf{(6VD)}}(\lambda +\mu -\eta |\beta )\,\mathcal{U}_{-,0'}(\mu |\alpha ,\beta +\sigma _{0}^{z}) 
   \\
  =\mathcal{U}_{-,0'}(\mu |\alpha ,\beta +\sigma _{0}^{z})\,R_{0'0}^{\mathsf{(6VD)}}(\lambda +\mu -\eta |\beta )\,
  \mathcal{U}_{-,0}(\lambda |\alpha,\beta +\sigma _{0'}^{z})\,R_{00'}^{\mathsf{(6VD)}}(\lambda -\mu |\beta ),
\end{multline}
which can be obtained from \eqref{refl-eq}.
The matrix $\mathcal{U}_{-}(\lambda |\alpha,\beta )$ \eqref{gauged-U} also satisfies the inversion relation
\begin{equation}\label{inv-Ugauge}
   \mathcal{U}_{-}(\lambda +\eta /2|\alpha ,\beta )\,
   \mathcal{U}_{-}(-\lambda+\eta /2|\alpha ,\beta )
   =\frac{\det_{q}\mathcal{U}_{-}(\lambda )}{\ths(2\lambda -2\eta )},
\end{equation}
which is a direct consequence of the inversion relation \eqref{inv-U-} for the boundary monodromy matrix $\mathcal{U}_{-}(\lambda )$.
From \eqref{inv-U-}, one can also deduce the following useful parity relations:
\begin{align}
  &\mathcal{B}_-(-\lambda|\alpha,\beta)=\frac{\ths(\lambda+(\alpha-\frac 12)\eta)}{\ths(\lambda-(\alpha-\frac 12)\eta)} 
    \frac{\ths(2\lambda+\eta)}{\ths(2\lambda-\eta)}\, \mathcal{B}_-(\lambda|\alpha,\beta),
  \label{parity-Bbeta}\\
  &\mathcal{C}_-(-\lambda|\alpha,\beta)=\frac{\ths(\lambda+(\alpha-\frac 12)\eta)}{\ths(\lambda-(\alpha-\frac 12)\eta)}
    \frac{\ths(2\lambda+\eta)}{\ths(2\lambda-\eta)}\, \mathcal{C}_-(\lambda|\alpha,\beta),
  \label{parity-Cbeta}
\end{align}
\begin{multline}\label{parity-Abeta}
  \mathcal{A}_-(\lambda|\alpha,\beta-1)
  =\frac{\ths(\eta)\ths(\beta\eta-2\lambda)}{\ths(2\lambda)\ths((\beta-1)\eta)}\,\mathcal{D}_-(\lambda|\alpha,\beta+1)
  \\
   -\frac{\ths(2\lambda-\eta)\ths(\beta\eta)}{\ths(2\lambda)\ths((\beta-1)\eta)}\frac{\ths(\lambda-(\alpha-\frac 12)\eta)}{\ths(\lambda+(\alpha-\frac 12)\eta)}\, \mathcal{D}_-(-\lambda|\alpha,\beta+1) ,
\end{multline}
\begin{multline}\label{parity-Dbeta}
  \mathcal{D}_-(\lambda|\alpha,\beta+1)
  =\frac{\ths(\eta)\ths(2\lambda+\beta\eta)}{\ths(2\lambda)\ths((\beta+1)\eta)}\,\mathcal{A}_-(\lambda|\alpha,\beta-1) 
  \\
  -\frac{\ths(2\lambda-\eta)\ths(\beta\eta)}{\ths(2\lambda)\ths((\beta+1)\eta)}\frac{\ths(\lambda-(\alpha-\frac 12)\eta)}{\ths(\lambda+(\alpha-\frac 12)\eta)}\, \mathcal{A}_-(-\lambda|\alpha,\beta-1).
\end{multline}
From \eqref{q-per-U1}-\eqref{q-per-U2} and \eqref{q-per-S}, we get the quasi-periodicity relations:
\begin{align}
   &\mathcal{U}_{-}(\lambda+\pi |\alpha ,\beta )=- \mathcal{U}_{-}(\lambda |\alpha ,\beta ), \label{q-per-gaugeU1}\\
   & \mathcal{U}_{-}(\lambda+\pi\omega |\alpha ,\beta )= (-e^{-2i\lambda-i\pi\omega})^{2N+3}\, e^{3i\eta+2i\alpha\eta}\, e^{i\eta\beta\sigma^z}\, \mathcal{U}_{-}(\lambda |\alpha ,\beta )\, e^{i\eta\beta\sigma^z}.  \label{q-per-gaugeU2}
\end{align}
Let us finally note that, by definition, the matrix elements $\mathcal{B}_{-}(\lambda |\alpha ,\beta )$ and $\mathcal{C}_{-}(\lambda |\alpha ,-\beta )$ can be written as the following linear combinations of the matrix elements of the boundary monodromy matrix $\mathcal{U}_{-}(\lambda)$:
\begin{align}
\mathcal{B}_{-}(\lambda |\alpha ,\beta )
&=  \mathcal{C}_{-}(\lambda |\alpha ,-\beta )
=\frac{\bar Y_{\alpha-\beta}(\eta/2-\lambda)\,\mathcal{U}_-(\lambda)\, Y_{\alpha-\beta}(\lambda-\eta/2)}{\det S(\eta/2-\lambda|\alpha,\beta)},
\end{align}
which, up to the factor $\det S(\eta/2-\lambda|\alpha,\beta)$, depend only on the difference $\alpha -\beta $ of the gauge parameters $\alpha$ and $\beta$. Hence,  it will be convenient, in the following, to consider a rescaled version of these matrix elements,
\begin{align}\label{Bhat}
  &\widehat{\mathcal{B}}_{-}(\lambda |\alpha -\beta ) 
  =\det S(\eta/2-\lambda|\alpha,\beta)\, \mathcal{B}_{-}(\lambda |\alpha,\beta ), \\
 &\widehat{\mathcal{C}}_{-}(\lambda |\alpha +\beta ) 
  =\det S(\eta/2-\lambda|\alpha,\beta)\, \mathcal{C}_{-}(\lambda |\alpha,\beta ),
\end{align}
which depends only on the difference $\alpha -\beta $, respectively on the sum $\alpha +\beta $.

\subsection{The transfer matrix}

The transfer matrix \eqref{transfer} can be rewritten as
\begin{equation}\label{transfer-bis}
        \mathcal{T}(\lambda )=\tr_0 [ \check K_{+}(\lambda |\alpha,\beta )\, \check{\mathcal{U}}_{-}(\lambda|\alpha,\beta)],
\end{equation}
in which we have defined
\begin{align}
  \check{\mathcal{U}}_{-}(\lambda |\alpha ,\beta )
  & =S_{0}^{-1}(-\lambda+\eta /2|\alpha+1 ,\beta )\ \mathcal{U}_{-}(\lambda )\ S_{0}(\lambda -\eta /2|\alpha-1,\beta )  
   \notag \\
  & =%
  \begin{pmatrix}
  \check{\mathcal{A}}_-(\lambda |\alpha ,\beta ) & \check{\mathcal{B}}_-(\lambda |\alpha,\beta ) \\ 
  \check{\mathcal{C}}_-(\lambda |\alpha ,\beta ) & \check{\mathcal{D}}_-(\lambda |\alpha,\beta )%
  \end{pmatrix},
  \label{gauged-Utilde}
\end{align}
so that in particular
\begin{align}
   \check{\mathcal{A}}_-(\lambda |\alpha ,\beta )
   &=\frac{\bar Y_{\alpha-\beta+1}(\eta/2-\lambda)\,\mathcal{U}_-(\lambda)\, Y_{\alpha+\beta-1}(\lambda-\eta/2)}{\det S(\eta/2-\lambda|\alpha+1,\beta)}   ,\nonumber\\
   &=\frac{\det S(\eta/2-\lambda|\alpha,\beta-1)}{\det S(\eta/2-\lambda|\alpha+1,\beta)}\,
   {\mathcal{A}}_-(\lambda |\alpha ,\beta-1 ),
\\
   \check{\mathcal{D}}_-(\lambda |\alpha ,\beta )
   &=\frac{- \bar Y_{\alpha+\beta+1}(\eta/2-\lambda)\,\mathcal{U}_-(\lambda)\, Y_{\alpha-\beta-1}(\lambda-\eta/2)}{\det S(\eta/2-\lambda|\alpha+1,\beta)} 
   \nonumber\\
   &=\frac{\det S(\eta/2-\lambda|\alpha,\beta+1)}{\det S(\eta/2-\lambda|\alpha+1,\beta)}\,
   {\mathcal{D}}_-(\lambda |\alpha ,\beta+1 ).
\end{align}
We have also defined
\begin{align}\label{def-tildeK+}
    \check K_{+}(\lambda |\alpha,\beta )
        &=S^{-1}(\lambda -\eta /2|\alpha-1,\beta ) \,K_{+}(\lambda )\, S(\eta /2-\lambda |\alpha+1,\beta ) \nonumber\\
        &=S^{-1}(\lambda +\eta /2|\alpha,\beta ) \,K_{+}(\lambda )\, S(-\lambda-\eta /2 |\alpha,\beta ) \nonumber\\
    &=\begin{pmatrix}
  \check a_+(\lambda|\alpha,\beta) & \check b_+(\lambda|\alpha,\beta) \\
  \check c_+(\lambda|\alpha,\beta) & \check d_+(\lambda|\alpha,\beta)
  \end{pmatrix}.
\end{align}

The matrix elements of \eqref{def-tildeK+} can easily be computed using the explicit expression of the Vertex-IRF matrix \eqref{mat-S}. Setting $\lambda^+=\lambda+\eta/2$, we obtain
\begin{align}\label{b_+gauge}
 &\check b_+(\lambda|\alpha,\beta) =\check c_+(\lambda|\alpha,-\beta)\nonumber\\
 &\quad=\frac{1}{2\det S(\lambda^+|\alpha,\beta)} \, \frac{\ths_1(2\lambda^+)}{\prod_{\ell=1}^3\ths_1(\alpha_\ell^+)}
 \bigg[\sum_{k\not=3}\prod_{\ell=1}^3\ths_k(\alpha_\ell^+)\ths_k((\beta-\alpha)\eta)-\prod_{\ell=1}^3\ths_3(\alpha_\ell^+)\ths_3((\beta-\alpha)\eta)\bigg]
  \nonumber\\
 &\quad=\frac{\ths_1(2\lambda^+)\ths_1\!\big(\frac{(\alpha-\beta)\eta+\sum_{\ell=1}^3\alpha^+_\ell}2\big)}{\det S(\lambda^+|\alpha,\beta)\prod_{\ell=1}^3\ths_1(\alpha_\ell^+)}
 \prod_{k=1}^3\ths_1\!\!\left(\frac{-2\alpha^+_k+\sum_{\ell=1}^3\alpha^+_\ell+(\beta-\alpha)\eta}2\right).
\end{align}
Hence, we can impose $\check c_+(\lambda|\alpha,\beta)=0$ by choosing $\alpha$ and $\beta$ such that
\begin{equation}\label{cond-c+=0}
  (\alpha+\beta)\eta=-\sum_{\ell=1}^3\epsilon_{\alpha^+_\ell} \alpha^+_\ell\mod(2\pi,2\pi\omega),
  \qquad\text{with}\ \
  \epsilon_{\alpha^+_\ell} =\pm 1\ \ \text{such that}\ \ \epsilon_{\alpha^+_1} \epsilon_{\alpha^+_2} \epsilon_{\alpha^+_3} =1.
\end{equation}
Similarly, we can impose $\check b_+(\lambda|\alpha,\beta)=0$ by choosing $\alpha$ and $\beta$ such that
\begin{equation}\label{cond-b+=0}
  (\alpha-\beta)\eta=-\sum_{\ell=1}^3\tilde\epsilon_{\alpha^+_\ell}  \alpha^+_\ell\mod(2\pi,2\pi\omega),
  \qquad\text{with}\ \
  \tilde\epsilon_{\alpha^+_\ell} =\pm 1\ \ \text{such that}\ \ \tilde\epsilon_{\alpha^+_1} \tilde\epsilon_{\alpha^+_2} \tilde\epsilon_{\alpha^+_3}=1.
\end{equation}
We also have
\begin{align}
   \check a_+(\lambda|\alpha,\beta) 
   &=\check d_+(\lambda|\alpha,-\beta) \nonumber\\
   &=\frac{\ths(2\lambda^+)}{2\det S(\lambda^+|\alpha,\beta)} \, 
    \bigg[ \frac{\ths(\lambda^++\beta\eta)\ths(-\alpha\eta)}{\ths(\lambda^+)}
      + \prod_{\ell=1}^3\frac{\ths_4(\alpha^+_\ell)}{\ths_1(\alpha_\ell)} \frac{\ths_4(\lambda^++\beta\eta)\ths_4(-\alpha\eta) }{\ths_4(\lambda^+)}
      \nonumber\\
      &\hspace{0.5cm}
      -\prod_{\ell=1}^3\frac{\ths_3(\alpha^+_\ell)}{\ths_1(\alpha^+_\ell)}\frac{\ths_3(\lambda^++\beta\eta)\ths_3(-\alpha\eta) }{\ths_3(\lambda^+)}
      +\prod_{\ell=1}^3\frac{\ths_2(\alpha^+_\ell)}{\ths_1(\alpha^+_\ell)}\frac{\ths_2(\lambda^++\beta\eta)\ths_2(-\alpha\eta) }{\ths_2(\lambda^+)} \bigg].
\end{align}
Under the condition \eqref{cond-c+=0}, this expression can be factorized as
\begin{align}\label{a+constr}
    \check a_+(\lambda|\alpha,\beta) 
   &= \check d_+(\lambda|\alpha,-\beta) 
   =\prod_{i=1}^3\frac{\ths(\lambda^+-\epsilon_{\alpha^+_i}\alpha^+_i)}{\ths(-\epsilon_{\alpha^+_i}\alpha^+_i)} ,
\end{align}
whereas, under the condition \eqref{cond-b+=0}, we have
\begin{align}\label{d+constr}
    \check a_+(\lambda|\alpha,-\beta) 
   &= \check d_+(\lambda|\alpha,\beta) 
   =\prod_{i=1}^3\frac{\ths(\lambda^+-\tilde\epsilon_{\alpha^+_i}\alpha^+_i)}{\ths(-\tilde\epsilon_{\alpha^+_i}\alpha^+_i)} .
\end{align}

Note that both constraints  \eqref{cond-c+=0} and  \eqref{cond-b+=0} can be satisfied simultaneously only if\footnote{The invertibility of the Vertex-IRF matrix in \eqref{gauged-Utilde} imposes $\eta\beta\not=0\mod (\pi,\pi\omega)$. }
there exists $i_+\in\{1,2,3\}$ such that 
\begin{align}\label{signs+}
 &\epsilon_{\alpha^+_{i_+}}=\tilde\epsilon_{\alpha^+_{i_+}} \quad \text{and} \quad \epsilon_{\alpha^+_\ell} =-\tilde\epsilon_{\alpha^+_\ell}  \ \text{for }\ \ell\not=i_+,
\end{align}
which leads to the following solution of \eqref{cond-c+=0}-\eqref{cond-b+=0}:
\begin{align}\label{alpha-beta+}
 &\alpha\eta= -\epsilon_{\alpha^+_{i_+}}\alpha^+_{i_+}, \quad \beta\eta = -\sum_{\ell\not= i_+}\epsilon_{\alpha^+_\ell} \alpha_\ell^+=\sum_{\ell\not= i_+}\tilde\epsilon_{\alpha^+_\ell} \alpha_\ell^+.
\end{align}

Hence, under both constraints  \eqref{cond-c+=0} and  \eqref{cond-b+=0} on the gauge parameters $\alpha$ and $\beta$, the matrix \eqref{def-tildeK+} can be made diagonal, and the transfer matrix \eqref{transfer} can be written in terms of $ {\mathcal{A}}_-(\lambda |\alpha ,\beta-1 )$ and $ {\mathcal{D}}_-(\lambda |\alpha ,\beta+1 )$ only:
\begin{align}\label{expr-transfer1}
  \mathcal{T}(\lambda)
  &= \check a_+(\lambda|\alpha,\beta)\, \check{\mathcal{A}}_-(\lambda |\alpha ,\beta )
  + \check d_+(\lambda|\alpha,\beta)\, \check{\mathcal{D}}_-(\lambda |\alpha ,\beta ),
  \nonumber\\
  &=\check a_+(\lambda|\alpha,\beta)\,\frac{\det S(\eta/2-\lambda|\alpha,\beta-1)}{\det S(\eta/2-\lambda|\alpha+1,\beta)}\,
   {\mathcal{A}}_-(\lambda |\alpha ,\beta-1 )
   \nonumber\\
   &
   \hspace{4cm}
   +\check d_+(\lambda|\alpha,\beta)\,\frac{\det S(\eta/2-\lambda|\alpha,\beta+1)}{\det S(\eta/2-\lambda|\alpha+1,\beta)}\,
   {\mathcal{D}}_-(\lambda |\alpha ,\beta+1 ).
\end{align}
By using  the parity relations \eqref{parity-Abeta} and \eqref{parity-Dbeta} and the fact that $\mathcal{T}(-\lambda)=\mathcal{T}(\lambda)$, this expression can be rewritten in terms of ${\mathcal{A}}_-(\lambda |\alpha ,\beta-1 )$ and ${\mathcal{A}}_-(-\lambda |\alpha ,\beta-1 )$ only, or equivalently in terms of $\mathcal{D}_-(\lambda |\alpha ,\beta+1 )$ and $\mathcal{D}_-(-\lambda |\alpha ,\beta+1 )$ only.
We can therefore formulate the following result:

\begin{prop}\label{prop-expr-transfer}
Under the constraints  \eqref{cond-c+=0} and  \eqref{cond-b+=0} for the gauge parameters $\alpha$ and $\beta$, the transfer matrix \eqref{transfer} can be expressed in terms of the elements of the gauge boundary monodromy matrix \eqref{gauged-U} as
\begin{align}\label{expr-transfer}
  \mathcal{T}(\lambda)
   &=\frac{\ths(2\lambda+\eta)}{\ths(2\lambda)}\, \mathsf{a}_+(\lambda|\boldsymbol{\epsilon_{\alpha^+}})\, 
    {\mathcal{A}}_-(\lambda |\alpha ,\beta-1 )
      +\frac{\ths(2\lambda-\eta)}{\ths(2\lambda)}\, \mathsf{a}_+(-\lambda|\boldsymbol{\epsilon_{\alpha^+}})\, 
      {\mathcal{A}}_-(-\lambda |\alpha ,\beta-1 )
     \nonumber\\
   &=\frac{\ths(2\lambda+\eta)}{\ths(2\lambda)}\, \mathsf{a}_+(\lambda|\boldsymbol{\tilde\epsilon_{\alpha^+}})\,
   \mathcal{D}_-(\lambda |\alpha ,\beta+1 )+ \frac{\ths(2\lambda-\eta)}{\ths(2\lambda)}\,  \mathsf{a}_+(-\lambda|\boldsymbol{\tilde\epsilon_{\alpha^+}})\,
   \mathcal{D}_-(-\lambda |\alpha ,\beta+1 ),
\end{align}
in which
\begin{align}
   & 
\mathsf{a}_+(\lambda|\boldsymbol{\epsilon_{\alpha^+}}) =\prod_{i=1}^3\frac{\ths(\lambda-\frac\eta 2-\epsilon_{\alpha^+_i}\alpha^+_i)}{\ths(-\epsilon_{\alpha^+_i}\alpha^+_i)}, \qquad
\label{coeff-a+} 
 \mathsf{a}_+(\lambda|\boldsymbol{\tilde\epsilon_{\alpha^+}})=\prod_{i=1}^3\frac{\ths(\lambda-\frac\eta 2-\tilde\epsilon_{\alpha^+_i}\alpha^+_i)}{\ths(-\tilde\epsilon_{\alpha^+_i}\alpha^+_i)}.
\end{align}
\end{prop}

Note that, under the constraints  \eqref{cond-c+=0} and  \eqref{cond-b+=0}, i.e. with $\alpha$ and $\beta$ fixed in terms of the + boundary parameters by  \eqref{alpha-beta+}, the two coefficients \eqref{coeff-a+} are related by
\begin{equation}\label{rel-a+}
   \mathsf{a}_+(\lambda+\tfrac\eta 2|\boldsymbol{\tilde\epsilon_{\alpha^+}}) 
   =-\frac{\ths(\lambda+\alpha\eta)}{\ths(\lambda-\alpha\eta)}\, \mathsf{a}_+(-\lambda+\tfrac\eta 2|\boldsymbol{\epsilon_{\alpha^+}}).
\end{equation}
%

\subsection{Boundary-bulk decomposition of the boundary gauge operators}

As in our previous works \cite{NicT22,NicT23}, it will be helpful to decompose the gauge boundary operators in \eqref{gauged-U} into bulk ones.
Still as in \cite{NicT22,NicT23}, we introduce for convenience the following reformulations of the bulk monodromy matrices:
\begin{align}
   M(\lambda) & =\bar R_{0N}(\lambda-\xi_N+\eta/2)\ldots\bar R_{01}(\lambda-\xi_1+\eta/2)= (-1)^N \hat T(-\lambda)
   =\begin{pmatrix} A(\lambda) & B(\lambda) \\ C(\lambda) & D(\lambda) \end{pmatrix},
   \label{def-M}\\
   \hat M(\lambda) &= (-1)^N \sigma^y M^t(-\lambda)\sigma^y
   =\bar R_{01}(\lambda+\xi_1+\eta/2)\ldots\bar R_{0N}(\lambda+\xi_n+\eta/2)= (-1)^N T(-\lambda) ,
   \label{def-Mhat}
\end{align}
defined in terms of
\begin{equation}\label{barR}
  \bar R(\lambda) =-R(-\lambda),
\end{equation}
which also corresponds to the R-matrix \eqref{R-mat} in which $\eta$ has been changed into $-\eta$.
Then we have
\begin{equation}
  \mathcal{U}_-(\lambda)= \hat M(-\lambda)\, K_-(\lambda)\, M(-\lambda),
\end{equation}
and therefore, defining as in \cite{NicT22,NicT23},
\begin{align}
   &M(\lambda|(\alpha,\beta),(\gamma,\delta))=S^{-1}(-\lambda-\eta/2|\alpha,\beta)\, M(\lambda)\, S(-\lambda-\eta/2|\gamma,\delta),
   \label{def-Mgauge}\\
   &\hat M(\lambda|(\alpha,\beta),(\gamma,\delta))=S^{-1}(\lambda+\eta/2|\gamma,\delta)\, \hat M(\lambda)\, S(\lambda+\eta/2|\alpha,\beta),
   \label{def-Mhat-gauge}\\
   & K_-(\lambda|(\gamma,\delta),(\gamma',\delta'))=S^{-1}(-\lambda+\eta/2|\gamma,\delta)\, K_-(\lambda)\, S(\lambda-\eta/2|\gamma',\delta'),
   \label{def-K--gauge}
\end{align}
we can write
\begin{equation}\label{boundary-bulk1}
   \mathcal{U}_-(\lambda|\alpha,\beta)= \hat M(-\lambda|(\gamma,\delta),(\alpha,\beta))\, K_-(\lambda|(\gamma,\delta),(\gamma',\delta'))\, M(-\lambda|(\gamma',\delta'),(\alpha,\beta)).
\end{equation}
Note that, as in the six-vertex case, we have the identity
\begin{align}
   \hat M(\lambda|(\alpha,\beta),(\gamma,\delta))
   &=(-1)^N\frac{ \det S(\lambda+\eta/2|\alpha,\beta)}{\det S(\lambda+\eta/2|\gamma,\delta)}\,\sigma^y\,
   M^t(-\lambda|(\alpha-1,\beta),(\gamma-1,\delta))\,\sigma^y,
\end{align}
so that \eqref{boundary-bulk1} can be rewritten as
\begin{align}\label{boundary-bulk2}
   \mathcal{U}_-(\lambda|\alpha,\beta)
   &=(-1)^N\frac{ \det S(-\lambda+\eta/2|\gamma,\delta)}{\det S(-\lambda+\eta/2|\alpha,\beta)}\,\sigma^y\,
   M^t(\lambda|(\gamma-1,\delta),(\alpha-1,\beta))\,\sigma^y\, 
   \nonumber\\
   &\hspace{5cm}
   \times
   K_-(\lambda|(\gamma,\delta),(\gamma',\delta'))\, M(-\lambda|(\gamma',\delta'),(\alpha,\beta)).
\end{align}

In components, we shall write
\begin{align}\label{M-components}
   M(\lambda|(\alpha,\beta),(\gamma,\delta))
   &=\begin{pmatrix} A(\lambda|(\alpha,\beta),(\gamma,\delta)) & B(\lambda|(\alpha,\beta),(\gamma,\delta))\\
      C(\lambda|(\alpha,\beta),(\gamma,\delta)) & D(\lambda|(\alpha,\beta),(\gamma,\delta))
      \end{pmatrix}
      \nonumber\\
   &=\frac{1}{\det S(-\lambda-\eta/2|\alpha,\beta)}
       \begin{pmatrix}
       A(\lambda|\alpha-\beta,\gamma+\delta) & B(\lambda|\alpha-\beta,\gamma-\delta)\\
       C(\lambda|\alpha+\beta,\gamma+\delta) & D(\lambda|\alpha+\beta,\gamma-\delta)
       \end{pmatrix},
\end{align}
so that 
\begin{align}\label{Mhat-components}
   \hat M(\lambda|(\alpha,\beta),(\gamma,\delta))
   &=\frac{(-1)^N}{\det S(\lambda+\eta/2|\gamma,\delta)}\,
     \nonumber\\
   &\quad\times
   \begin{pmatrix}
    D(-\lambda|\alpha-1+\beta,\gamma-1-\delta) & -B(-\lambda|\alpha-1-\beta,\gamma-1-\delta) \\
    -C(-\lambda|\alpha-1+\beta,\gamma-1+\delta) & A(-\lambda|\alpha-1-\beta,\gamma-1+\delta)
   \end{pmatrix},
\end{align}
in which we have used that $\det S(\lambda-\eta/2|\alpha-1,\beta)=\det S(\lambda+\eta/2|\alpha,\beta)$.
If $\gamma=\gamma'$ and $\delta=\delta'$, we therefore have the following boundary-bulk decomposition:
\begin{multline}\label{Boundary-bulkU}
    \mathcal{U}_-(\lambda|\alpha,\beta)=\frac{(-1)^N}{\det S(-\lambda+\eta/2|\alpha,\beta)}\,
       \begin{pmatrix}
    D(\lambda|\gamma-1+\delta,\alpha-1-\beta) & -B(\lambda|\gamma-1-\delta,\alpha-1-\beta) \\
    -C(\lambda|\gamma-1+\delta,\alpha-1+\beta) & A(\lambda|\gamma-1-\delta,\alpha-1+\beta)
   \end{pmatrix}
   \\
   \times
   K_-(\lambda|(\gamma,\delta),(\gamma,\delta))
   \frac{1}{\det S(\lambda-\eta/2|\gamma,\delta)}
       \begin{pmatrix}
       A(-\lambda|\gamma-\delta,\alpha+\beta) & B(-\lambda|\gamma-\delta,\alpha-\beta)\\
       C(-\lambda|\gamma+\delta,\alpha+\beta) & D(-\lambda|\gamma+\delta,\alpha-\beta)
       \end{pmatrix}.
\end{multline}

Some useful properties of the bulk gauge Yang-Baxter generators, elements of \eqref{M-components}, are recalled in Appendix~\ref{app-bulk-gauge}.

The matrix elements of $K_-(\lambda|(\gamma,\delta),(\gamma,\delta))$ can be computed similarly as those of \eqref{def-tildeK+}. 
Setting
\begin{equation}
    K_-(\lambda|(\gamma,\delta),(\gamma,\delta))
    =\begin{pmatrix} a_-(\lambda|\gamma,\delta) & b_-(\lambda|\gamma,\delta) \\
       c_-(\lambda|\gamma,\delta) & d_-(\lambda|\gamma,\delta) \end{pmatrix},
\end{equation}
we obtain
\begin{align}\label{b_-gauge}
 b_-(\lambda|\gamma,\delta)
 &=c_-(\lambda|\gamma,-\delta) \nonumber\\
 &=\frac{\ths(2\lambda-\eta )\ths\!\big(\frac{(\delta-\gamma)\eta+\sum_{\ell=1}^3\alpha^-_\ell}2\big)}{\det S(\lambda-\frac\eta 2|-\gamma,-\delta)\prod_{\ell=1}^3\ths(\alpha_\ell^-)}
 \prod_{k=1}^3\ths\!\left(\frac{-2\alpha^-_k+\sum_{\ell=1}^3\alpha^-_\ell-(\delta-\gamma)\eta}2\right),
\end{align}
and
\begin{align}\label{a_-gauge}
   a_-(\lambda|\gamma,\delta) 
   &=d_-(\lambda|\gamma,-\delta) \nonumber\\
   &=\frac{\ths(2\lambda-\eta)}{2\det S(\lambda-\frac\eta 2|-\gamma,-\delta)} \, 
      \bigg[ \frac{\ths(\lambda-\frac\eta 2-\delta\eta)\ths(\gamma\eta)}{\ths(\lambda-\frac\eta 2)}
      + c^x_- \frac{\ths_4(\lambda-\frac\eta 2-\delta\eta)\ths_4(\gamma\eta) }{\ths_4(\lambda-\frac\eta 2)}
      \nonumber\\
      &\hspace{3cm}
      + c^y_-\frac{\ths_3(\lambda-\frac\eta 2-\delta\eta)\ths_3(\gamma\eta) }{\ths_3(\lambda-\frac\eta 2)}
      +c^z_-\frac{\ths_2(\lambda-\frac\eta 2-\delta\eta)\ths_2(\gamma\eta) }{\ths_2(\lambda-\frac\eta 2)} \bigg].
\end{align}
Hence we can impose $c_-(\lambda|\gamma,\delta)=0$ by choosing $\gamma$ and $\delta$ such that
\begin{equation}\label{cond-c-=0}
 (\gamma+\delta)\eta=\sum_{\ell=1}^3\epsilon_{\alpha^-_\ell} \alpha^-_\ell\mod(2\pi,2\pi\omega),
  \qquad\text{with}\ \
  \epsilon_{\alpha^-_\ell}=\pm 1\ \ \text{such that}\ \ \epsilon_{\alpha^-_1}\epsilon_{\alpha^-_2}\epsilon_{\alpha^-_3}=1.
\end{equation}
Similarly, we can impose $b_-(\lambda|\gamma,\delta)=0$ by choosing $\gamma$ and $\delta$ such that
\begin{equation}\label{cond-b-=0}
  (\gamma-\delta)\eta=\sum_{\ell=1}^3\tilde\epsilon_{\alpha^-_\ell} \alpha^-_\ell\mod(2\pi,2\pi\omega),
  \qquad\text{with}\ \
  \tilde\epsilon_{\alpha^-_\ell}=\pm 1\ \ \text{such that}\ \ \tilde\epsilon_{\alpha^-_1}\tilde\epsilon_{\alpha^-_2}\tilde\epsilon_{\alpha^-_3}=1.
\end{equation}
Thus, under the condition \eqref{cond-c-=0}:
\begin{align}\label{a-constr}
   a_-(\lambda|\gamma,\delta) 
   &= d_-(\lambda|\gamma,-\delta) 
   =\prod_{i=1}^3\frac{\ths(\lambda-\frac\eta 2-\epsilon_{\alpha^-_i} \alpha^-_i)}{\ths(-\epsilon_{\alpha^-_i} \alpha^-_i)} 
   \equiv \mathsf{a}_-(\lambda|\boldsymbol{\epsilon_{\alpha^-} })  ,
\end{align}
and, under the condition \eqref{cond-b-=0}, we have
\begin{align}\label{d-constr}
   a_-(\lambda|\gamma,-\delta) 
   &=d_-(\lambda|\gamma,\delta) 
   =\prod_{i=1}^3\frac{\ths(\lambda-\frac\eta 2-\tilde\epsilon_{\alpha^-_i} \alpha^-_i)}{\ths(-\tilde\epsilon^-_i\alpha_{\alpha^-_i} )}
    \equiv \mathsf{a}_-(\lambda|\boldsymbol{\tilde\epsilon_{\alpha^-} }) .
\end{align}
Note also that both constraints  \eqref{cond-c-=0} and  \eqref{cond-b-=0} can be satisfied simultaneously only if there exists $i_-\in\{1,2,3\}$ such that 
\begin{align}\label{signs-}
 &\epsilon_{\alpha^-_{i_-}}=\tilde\epsilon_{\alpha^-_{i_-}} \quad \text{and} \quad \epsilon_{\alpha^-_\ell}=-\tilde\epsilon_{\alpha^-_\ell} \ \text{for }\ \ell\not=i_-,
\end{align}
which leads to the following solution of \eqref{cond-c-=0}-\eqref{cond-b-=0}:
\begin{align}\label{gamma-delta-}
 &\gamma\eta= \epsilon_{\alpha^-_{i_-}}\alpha^-_{i_-}, \quad \delta\eta = \sum_{\ell\not= i_-}\epsilon_{\alpha^-_\ell}\alpha_\ell^-=-\sum_{\ell\not= i_-}\tilde\epsilon_{\alpha^-_\ell}\alpha_\ell^-.
\end{align}

In particular, from \eqref{Boundary-bulkU}, and under the conditions \eqref{cond-c-=0} and \eqref{cond-b-=0} which make $K_-(\lambda|(\gamma,\delta),(\gamma,\delta))$ diagonal, we have the following boundary-bulk decomposition for the operator $\widehat{\mathcal B}_-(\lambda|\alpha-\beta)$:
%
\begin{align}\label{boundary-bulk-B2}
    \widehat{\mathcal B}_-(\lambda|\alpha-\beta)
    &=\frac{(-1)^N}{\det S(\lambda-\eta/2|\gamma,\delta)}
    \Big[ D(\lambda|\gamma-1+\delta,\alpha-1-\beta) \, a_-(\lambda|\gamma,\delta) \, B(-\lambda|\gamma-\delta,\alpha-\beta)
    \nonumber\\
    &\hspace{3.1cm} -B(\lambda|\gamma-1-\delta,\alpha-1-\beta)\, d_-(\lambda|\gamma,\delta) \, D(-\lambda|\gamma+\delta,\alpha-\beta) \Big]
     \nonumber\\
  &=\frac{(-1)^N}{\ths(\eta(\delta+1))}
   \frac{\ths(2\lambda-\eta)}{\ths(2\lambda)}\sum_{\sigma=\pm} \frac{\sigma \mathsf{a}_-(-\sigma\lambda|{\boldsymbol{\epsilon_{\alpha^-} }})}{\ths(\sigma\lambda+\frac\eta 2+\eta \gamma)} \nonumber\\
  &\hspace{3cm}\times
   B(\sigma\lambda|\gamma-\delta-1,\alpha-\beta-1)\, D(-\sigma\lambda |\gamma+\delta,\alpha-\beta),
\end{align}    
in which we have used \eqref{Comm-DB} and the notation \eqref{a-constr}.
This decomposition will be used in section~\ref{sec-Bb-Bethe} for the boundary-bulk decomposition of the generalised Bethe states.

\section{Diagonalisation of the transfer matrix by SoV in the Vertex-IRF framework}
\label{sec-SoV}

In our previous work \cite{NicT24}, we explained how the transfer matrix $\mathcal{T}(\lambda)$ could be diagonalised within the new SoV approach developed for the open XXX and XXZ chains in \cite{MaiN19}.
We have shown that, under very general hypothesis (the hypothesis that the boundary matrices $K_+(\lambda)$ and $K_-(\lambda)$ are not both proportional to the identity and that the inhomogeneity parameters are generic enough) the transfer matrix is diagonalisable with simple spectrum. Under these hypothesis, the explicit diagonalisation of the transfer matrix was performed in \cite{NicT24} via the construction of some basis (the so-called {\em SoV basis})
\begin{equation}\label{bases-SoV}
  \big\{ \ket{\mathbf{h}}, \mathbf{h}\equiv(h_1,\ldots,h_N)\in\{0,1\}^N\big\},
  \quad \text{and} \quad
   \big\{ \bra{\mathbf{h}}, \mathbf{h}\equiv(h_1,\ldots,h_N)\in\{0,1\}^N\big\},
\end{equation}
of $\mathcal{H}$ and $\mathcal{H}^*$ respectively, such that
\begin{equation}\label{ortho-states}
   \moy{\mathbf{h}\,|\,\mathbf{k}}=\delta_{\mathbf{h},\mathbf{k}}\, \frac{\mathcal{N}_0}{V(\xi_1^{(h_1)},\ldots,\xi_N^{(h_N)})\,V(\xi_1,\ldots,\xi_N)},
   \qquad
   \forall \,\mathbf{h},\mathbf{k}\in\{0,1\}^N,
\end{equation}
for some normalisation coefficient $\mathcal{N}_0$,
and on which the right and left-eigenstates of $\mathcal{T}(\lambda)$ can be expressed in the form of {\em separate state} as
\begin{align}
  & 
  \sum_{\mathbf{h}\in\{0,1\}^N} \prod_{n=1}^N  Q(\xi_n^{(h_n)})\ 
  V(\xi_1^{(h_1)},\ldots,\xi_N^{(h_N)})\,\ket{\mathbf{h}},
  \label{eigen-r}\\
  & 
  \sum_{\mathbf{h}\in\{0,1\}^N} \prod_{n=1}^N \left[\left(\frac{\ths(2\xi_n-2\eta)\,\mathbf{A}(\frac\eta 2+\xi_n)}{\ths(2\xi_n+2\eta)\,\mathbf{A}(\frac \eta 2-\xi_n)}\right)^{\! h_n} Q(\xi_n^{(h_n)})\right]\, 
  V(\xi_1^{(h_1)},\ldots,\xi_N^{(h_N)})\,\bra{\mathbf{h}}.
  \label{eigen-l}
\end{align}
Here we have used the shortcut notations
\begin{align}
   &\xi_n^{(h)}=\xi_n+\frac\eta 2-h\eta, \qquad  n\in\{1,\ldots,N\}, \quad h\in\{0,1\},\label{def-xi-shift}\\
   &V(\zeta_1,\ldots,\zeta_N)=\prod_{1\le i< j\le N}\ths(\zeta_j-\zeta_i)\ths(\zeta_j+\zeta_i) ,  \label{VDM}
\end{align}
for any set of variables $\zeta_1,\ldots,\zeta_N$.
More precisely, the separate states of the form \eqref{eigen-r}, \eqref{eigen-l} are respectively the right and left eigenstate (which are uniquely defined up to an overall normalisation factor)  of the transfer matrix $\mathcal{T}(\lambda)$ with eigenvalue $\tau(\lambda)$ if their coefficients satisfy the following discrete version of the $TQ$-equation:
\begin{equation}\label{T-Qdis}
   \frac{Q(\xi_n^{(1)})}{Q(\xi_n^{(0)})}=\frac{\tau(\xi_n+\frac\eta 2)}{\mathbf{A}(\frac\eta 2+\xi_n)}=\frac{\mathbf{A}(\frac\eta 2-\xi_n)}{\tau(\xi_n-\frac\eta 2)},
   \qquad \forall n\in\{1,\ldots,n\},
\end{equation}
with $\mathbf{A}$ being a function such that
\begin{equation}\label{def-A}
   \mathbf{A}(\lambda+\eta/2)\,\mathbf{A}(-\lambda+\eta/2)=\frac{\det_q K_+(\lambda)\, \det_q\mathcal{U}_-(\lambda)}{\ths(\eta+2\lambda)\,\ths(\eta-2\lambda)}.
\end{equation}

In \cite{NicT24}, the SoV basis \eqref{bases-SoV} were directly constructed by multiple action of the transfer matrix itself on a sufficiently generic co-vector $\bra{S}$ and its vector counterpart $\ket{R}$ (adequately chosen with respect to $\bra{S}$ so as to satisfy the orthogonality condition \eqref{ortho-states}):
\begin{equation}\label{new-SoV-states}
  \bra{\mathbf{h}}\equiv \bra{S}\prod_{n=1}^N\left(\frac{\mathcal{T}(\xi_n-\frac \eta 2)}{\mathbf{A}(\frac\eta 2-\xi_n)}\right)^{\! 1-h_n},
  \qquad
  \ket{\mathbf{h}}=\prod_{n=1}^N\left(\frac{\ths(2\xi_n-2\eta)}{\ths(2\xi_n+2\eta)}
  \frac{\mathcal{T}(\xi_n+\frac \eta 2)}{\mathbf{A}(\frac \eta 2-\xi_n)}\right)^{\! h_n} \ket{R}.
\end{equation}

The construction \eqref{new-SoV-states} has the advantage of being very general, see also \cite{MaiN19}. However, 
for our purpose of computing correlation functions, we shall use here the more specific and explicit SoV construction initially proposed in \cite{FalN14}, which generalises Sklyanin's approach \cite{Skl85,Skl90} to the open XYZ chain in the Vertex-IRF framework. 
The corresponding SoV basis has the property to pseudo-diagonalise the generalised gauge-transformed operator $\mathcal{B}_-(\lambda|\alpha,\beta)$.
As a consequence, the separate states can be rewritten as generalised Bethe states, on which we can compute the action of local operators by usual techniques (see Section~\ref{sec-act}), similarly as what was done in our previous works \cite{NicT22,NicT23} on the open XXZ case. We recall here the details of this specific SoV construction within the notations introduced in the previous sections.

\subsection{The SoV basis}

As in \cite{NicT24}, we introduce the notations,
\begin{align}
   &\xi_{0}^{(0)}=\frac\eta 2, \quad \xi_{-1}^{(0)}=\frac \eta 2 +\frac\pi 2,\quad \xi_{-2}^{(0)}=\frac \eta 2 +\frac{\pi\omega} 2,\quad \xi_{-3}^{(0)}=\frac \eta 2 +\frac{\pi+\pi\omega} 2, \label{special_points}
\end{align}
and we suppose from now on that the inhomogeneity parameters are generic, or at least  that they satisfy the condition 
\begin{equation}\label{cond-inhom}
 \begin{aligned}
 &\epsilon_n \,\xi_n^{(h_n)},\ 1\le n\le N,\ h_n\in\{0,1\},\ \epsilon_n\in\{-1,1\}, \quad\\
 &\text{and} \quad \epsilon_j \,\xi_{-j}^{(0)},\ 1\le j\le 3,\ \epsilon_j\in\{-1,1\}, \quad 
 \end{aligned}
 \quad
 \text{are pairwise distinct modulo }(\pi,\pi\omega).
\end{equation}
%
%
%
We also consider the usual reference states
\begin{align}
   &\ket{0}=\Motimes_{n=1}^N \begin{pmatrix}\, 1\, \\ \, 0 \,\end{pmatrix}_{\! n},\qquad
   \ket{\underline{0}}=\Motimes_{n=1}^N\begin{pmatrix} \, 0\,  \\ 1 \end{pmatrix}_{\! n},
\end{align}
their respective dual states $\bra{0}$ and $\bra{\underline{0}}$, as well as  the following product of Vertex-IRF matrices all along the chain:
\begin{align}
    S_{1\ldots N}(\boldsymbol{\xi} |\alpha,\beta)
    &=S_1(-\xi_1|\alpha,\beta)\, S_2(-\xi_2|\alpha,\beta+\sigma_1^z)\ldots S_N(-\xi_N|\alpha,\beta+\sigma_1^z+\ldots+\sigma_{N-1}^z).
\end{align}

With these notations, let us define, for each $N$-tuple $\mathbf{h}=(h_1,\ldots,h_N)\in\{0,1\}^N$, and for arbitrary gauge coefficients $\alpha$ and $\beta$, the SoV states
\begin{align}
   &\ket{\mathbf{h},\alpha,\beta}
   =\prod_{n=1}^N\left(\frac{\mathcal{D}_-(\xi_n+\frac\eta 2|\alpha,\beta)}{k_n(\alpha)\,\mathsf{A}_-(\frac\eta 2-\xi_n)}\right)^{\! h_n} 
   S_{1\ldots N}(\boldsymbol{\xi} |\alpha,\beta-1)\,\ket{\underline{0}}, 
   \label{ket-state}\\
   &\bra{\alpha,\beta,\mathbf{h}}
   =\bra{0}\, S^{-1}_{1\ldots N}(\boldsymbol{\xi} |\alpha,\beta+1)
   \prod_{n=1}^N\left(\frac{\mathcal{A}_-(\frac\eta 2-\xi_n |\alpha,\beta)}{\mathsf{A}_-(\frac\eta 2-\xi_n)}\right)^{\! 1-h_n} ,
   \label{bra-state}
\end{align}
with coefficients given by
\begin{equation}\label{k_n-A-}
   k_n(\alpha)=\frac{\ths(2\xi_n+\eta)}{\ths(2\xi_n-\eta)}\,\frac{\ths(\alpha\eta-\xi_n)}{\ths(\alpha\eta+\xi_n)} ,
   \qquad
   \mathsf{A}_-(\lambda)=\mathsf{g}_-(\lambda)\, a(\lambda)\, d(-\lambda),
\end{equation}
where $\mathsf{g}_-$ is any function such that
\begin{equation}\label{rel-g-}
   \mathsf{g}_-(\lambda+\tfrac\eta 2)\, \mathsf{g}_-(-\lambda+\tfrac\eta 2)=\frac{\det_q K_-(\lambda)}{\ths(2\lambda-2\eta)}.
\end{equation}
%

\subsubsection{Action of the gauge boundary operators on the SoV states}

The states \eqref{ket-state} and \eqref{bra-state}  are right and left pseudo-eigenstates of $\mathcal{B}_-(\lambda|\alpha,\beta)$:
\begin{align}
  &\mathcal{B}_-(\lambda|\alpha,\beta)\, \ket{\mathbf{h},\alpha,\beta}=(-1)^N a_\mathbf{h}(\lambda)\,a_\mathbf{h}(-\lambda)\,
  \frac{\ths(\eta(\beta-N))}{\ths(\eta\beta)}\, b_-(\lambda|\alpha,\beta-N)\,\ket{\mathbf{h},\alpha,\beta+2},
  \label{eigen-r-B}\\
  &\bra{\alpha,\beta,\mathbf{h}}\, \mathcal{B}_-(\lambda|\alpha,\beta)=(-1)^N a_\mathbf{h}(\lambda)\,a_\mathbf{h}(-\lambda)\,
  \frac{\ths(\eta(\beta+N))}{\ths(\eta\beta)}\, b_-(\lambda|\alpha,\beta+N)\, \bra{\alpha,\beta-2,\mathbf{h}},
  \label{eigen-l-B}
\end{align}
where $b_-(\lambda|\alpha,\beta)$ is given by \eqref{b_-gauge} and
\begin{equation}\label{a_h}
   a_\mathbf{h}(\lambda)=\prod_{n=1}^N\ths(\lambda-\xi_n-\eta/ 2+h_n\eta).
\end{equation}
Equivalently, \eqref{eigen-r-B} 
can be rewritten as
\begin{align}
  &\widehat{\mathcal{B}}_-(\lambda|\alpha-\beta)\, \ket{\mathbf{h},\alpha,\beta}
  =(-1)^N a_\mathbf{h}(\lambda)\,a_\mathbf{h}(-\lambda)\, \ths(2\lambda-\eta)\, \mathsf{b}_-(\alpha-\beta+N)\,
  \ket{\mathbf{h},\alpha,\beta+2},
  \label{act-Bhat-ket} 
\end{align}
in which
\begin{equation}\label{coeff-b_-}
    \mathsf{b}_-(x)=\ths\bigg(\frac{\sum_{\ell=1}^3\alpha_\ell^--x\eta}2\bigg)
    \prod_{k=1}^3\frac{\ths\bigg(\frac{\sum_{\ell=1}^3\alpha_\ell^-+x\eta}2-\alpha_k^-\bigg)}{\ths(\alpha_k^-)}.
\end{equation}

Moreover, the action of the other elements of the gauge transformed boundary monodromy matrix \eqref{gauged-U} on the states \eqref{ket-state} and \eqref{bra-state} can also be explicitly computed.

In particular, the left action of $\mathcal{A}_-(\lambda|\alpha,\beta)$ on $ \bra{\alpha,\beta,\mathbf{h}}$ is quite simple. It can easily be computed at the points $\xi^{(0)}_{-k}$ \eqref{special_points} for $k=0,1,2,3$:
\begin{equation}\label{act0-A-bra}
  \bra{\alpha,\beta,\mathbf{h}}\, \mathcal{A}_-( \xi^{(0)}_{-k} |\alpha,\beta)
  = 
   (-1)^N\,  a_-( \xi^{(0)}_{-k} |\alpha,\beta+N)\, a( \xi^{(0)}_{-k} )\, d(- \xi^{(0)}_{-k} ) \,
   \bra{\alpha,\beta,\mathbf{h}},
\end{equation}
in which $a_-$ is given by \eqref{a_-gauge},
and at the points $\epsilon_n\xi_n^{(h_n)}$, $\epsilon_n=\pm$, $n=1,\ldots,N$:
\begin{equation}\label{act-A-xi-bra}
    \bra{\alpha,\beta,\mathbf{h}}\, \mathcal{A}_-(\epsilon_n\xi_n^{(h_n)}|\alpha,\beta)
    =\mathsf{A}_-(\epsilon_n\xi_n^{(h_n)})\, \bra{\alpha,\beta,\mathsf{T}_n^{\epsilon_n}\mathbf{h}},
\end{equation}
with the convention
\begin{align}
   &\mathsf{T}_n^{\epsilon_n} (h_1,\ldots,h_n,\ldots,h_N)=(h_1,\ldots,h_n+\epsilon_n,\ldots,h_N),
   \label{shift-h}\\
   &\bra{\alpha,\beta,\mathbf{h}}=0\qquad \text{if}\quad \mathbf{h}\notin\{0,1\}^N.
\end{align}
Using the fact that $\det S(\eta/2-\lambda|\alpha,\beta)\, \mathcal{A}_-(\lambda |\alpha,\beta)$ is an entire function of $\lambda$ which satisfies from \eqref{q-per-gaugeU1}-\eqref{q-per-gaugeU2} some specific quasi-periodicity properties\footnote{More precisely, if we set $\widehat{\mathcal{A}}_-(\lambda|\alpha,\beta)=\det S(\eta/2-\lambda|\alpha,\beta)\, \mathcal{A}_-(\lambda |\alpha,\beta)$, we have from \eqref{q-per-gaugeU1}-\eqref{q-per-gaugeU2}:
\begin{align}
 &\widehat{\mathcal{A}}_{-}(\lambda+\pi|\alpha,\beta)=\widehat{\mathcal{A}}_{-}(\lambda|\alpha,\beta), \label{q-per-hgaugeA1}\\
 &\widehat{\mathcal{A}}_{-}(\lambda+\pi\omega|\alpha,\beta)= (-e^{-2i\lambda-i\pi\omega})^{2N+4}\, e^{4i\eta+2i\eta\beta}\,\widehat{\mathcal{A}}_{-}(\lambda|\alpha,\beta). \label{q-per-hgaugeA2}
\end{align}
}
 with respect to $\pi$ and $\pi\omega$, we therefore get
\begin{multline}\label{act-A-lambda-bra}
    \bra{\alpha,\beta,\mathbf{h}}\,\mathcal{A}(\lambda|\alpha,\beta) = \mathfrak{A}_{\mathbf{h}}(\lambda|\alpha,\beta)\, \bra{\alpha,\beta,\mathbf{h}}
    \\
    +\sum_{n=1}^N\sum_{\epsilon_n=\pm}
    \frac{\ths(2\lambda-\eta) \ths(\lambda+\epsilon_n\xi_n^{(h_n)}) \ths(\lambda-\epsilon_n\xi_n^{(h_n)}\! -\eta\beta )}
    {\ths(2\xi_n^{(h_n)}-\epsilon_n\eta) \ths(2\xi_n^{(h_n)}) \ths(-\eta\beta )} \,
    \prod_{k\not= n} 
    \frac{\ths(\lambda-\xi_k^{(h_k)})\ths(\lambda+\xi_k^{(h_k)})}{\ths(\xi_n^{(h_n)}-\xi_k^{(h_k)})\ths(\xi_n^{(h_n)}+\xi_k^{(h_k)})}\,
    \\
    \times
    \frac{\det S(\tfrac\eta 2-\epsilon_n\xi_n^{(h_n)}|\alpha,\beta)}{\det S(\tfrac\eta 2-\lambda |\alpha,\beta)}\,
    \mathsf{A}_-(\epsilon_n\xi_n^{(h_n)})\, \bra{\alpha,\beta,\mathsf{T}_n^{\epsilon_n}\mathbf{h}},
\end{multline}
in which the coefficient of the diagonal action is given by
\begin{multline}\label{act-A-lambda-bra-0}
    \mathfrak{A}_{\mathbf{h}}(\lambda|\alpha,\beta)
    =\sum_{k=0}^3
    \frac{\ths(\lambda-\xi^{(0)}_{-k} +\pi+\pi\omega-\eta\beta)}{\ths(\pi +\pi\omega-\eta\beta)}\,
    \frac{a_\mathbf{h}(\lambda)\, a_\mathbf{h}(-\lambda)}{a_\mathbf{h}(\xi^{(0)}_{-k} )\, a_\mathbf{h}(-\xi^{(0)}_{-k} )} \,
    \prod_{\substack{ k'=0 \\ k'\not=k}}^3 \frac{\ths(\lambda-\xi^{(0)}_{-k'})}{\ths(\xi^{(0)}_{-k} -\xi^{(0)}_{-k'})}\,
     \\
     \times
  \frac{ \det S(\tfrac\eta 2-\xi^{(0)}_{-k}  |\alpha,\beta)}{\det S(\tfrac\eta 2-\lambda |\alpha,\beta)}\,
  (-1)^Na_-(\xi^{(0)}_{-k}  |\alpha,\beta+N)\, a(\xi^{(0)}_{-k} )\, d(-\xi^{(0)}_{-k} )  .
\end{multline}
The action of $\mathcal{D}_-(\lambda|\alpha,\beta+2)$ on $ \bra{\alpha,\beta,\mathbf{h}}$ can then be obtained by the parity relation \eqref{parity-Dbeta}.

A similar strategy can be used to compute the action of $\mathcal{D}_-(\lambda|\alpha,\beta)$ on $ \ket{\mathbf{h},\alpha,\beta}$:
\begin{multline}\label{act-D-lambda-ket}
    \mathcal{D}(\lambda|\alpha,\beta)\, \ket{\mathbf{h},\alpha,\beta}  = \mathfrak{A}_{\mathbf{h}}(\lambda|\alpha,- \beta)\, \ket{\mathbf{h},\alpha,\beta} 
    \\
    +\sum_{n=1}^N\sum_{\epsilon_n=\pm}
    \frac{\ths(2\lambda-\eta) \ths(\lambda+\epsilon_n\xi_n^{(h_n)}) \ths(\lambda-\epsilon_n\xi_n^{(h_n)}-\eta\beta )}
    {\ths(2\xi_n^{(h_n)}-\epsilon_n\eta) \ths(2\xi_n^{(h_n)})\, \ths(-\eta\beta )} \,
     \prod_{k\not= n} 
    \frac{\ths(\lambda-\xi_k^{(h_k)})\ths(\lambda+\xi_k^{(h_k)})}{\ths(\xi_n^{(h_n)}-\xi_k^{(h_k)})\ths(\xi_n^{(h_n)}+\xi_k^{(h_k)})}\,
    \\
    \times
    \frac{\det S(\tfrac\eta 2-\epsilon_n\xi_n^{(h_n)}|\alpha,\beta)}{\det S(\tfrac\eta 2-\lambda |\alpha,\beta)}\,
    k_n(\alpha)^{\epsilon_n}\, \mathsf{A}_-(-\epsilon_n\xi_n^{(1-h_n)})\, \ket{\mathsf{T}_n^{\epsilon_n}\mathbf{h},\alpha,\beta},
 \end{multline}
in which we have used the convention \eqref{shift-h} with $\ket{\mathbf{h},\alpha,\beta}=0$ if $\mathbf{h}\notin\{0,1\}^N$.
The action of $\mathcal{A}_-(\lambda|\alpha,\beta-2)$ on $\ket{\mathbf{h},\alpha,\beta} $ can then be obtained by the parity relation \eqref{parity-Abeta}.

\subsubsection{Orthogonality and normalisation}

It follows from these actions that the right and left SoV states \eqref{ket-state} and \eqref{bra-state} satisfy the orthogonality property\footnote{Note that \eqref{norm-cond} is of the form \eqref{ortho-states} with $\mathcal{N}_0=V(\xi_1,\ldots,\xi_N)\,  \mathsf{N}(\boldsymbol{\xi}|\alpha,\beta)$.}
\begin{equation}\label{norm-cond}
   \moy{ \alpha,\beta-1,\mathbf{h} \, | \, \mathbf{k},\alpha,\beta+1}
   = \delta_{\mathbf{h},\mathbf{k}}\, 
   \frac{\mathsf{N}(\boldsymbol{\xi}|\alpha,\beta)}{V(\xi_1^{(h_1)},\ldots,\xi_N^{(h_N)})},
   \qquad
   \forall \mathbf{h},\mathbf{k}\in\{0,1\}^N,
\end{equation}
with (see Appendix~C of \cite{KitMNT18})
\begin{align}\label{norm-coeff}
   \mathsf{N}(\boldsymbol{\xi}|\alpha,\beta)
   &=V(\xi_1^{(0)},\ldots,\xi_N^{(0)})\, \moy{\alpha,\beta-1,\mathbf{0}\, |\,\mathbf{0},\alpha,\beta+1}
   \nonumber\\
   &=(-1)^N\prod_{n=1}^{N}\left[ \frac{b_-(\tfrac\eta 2-\xi_n|\alpha,\beta+N-2n+1)}{\mathsf{g}_-(\tfrac\eta 2-\xi_n)}\,\frac{\ths((\beta+N-2n+1)\eta)}{\ths((\beta+N-n)\eta)}\right]
   \nonumber\\
   &\hspace{6cm}\times\frac{V(\xi_1^{(0)},\ldots,\xi_N^{(0)})}{V(\xi_1^{(1)},\ldots,\xi_N^{(1)})}\,  V(\xi_1,\ldots,\xi_N).
\end{align}
Hence, under the condition \eqref{cond-inhom}, and for any $\alpha$ and $\beta$ such that the normalisation coefficient \eqref{norm-coeff} is non-zero\footnote{This means in particular that the gauge parameters $\alpha$ and $\beta$ have to be chosen such that
\begin{equation}\label{cond-norm}
  \prod_{n=1}^{N} \mathsf{b}_-(\alpha-\beta-N+2n-1)\not=0,
\end{equation}
in which $\mathsf{b}_-(x)$ is given by \eqref{coeff-b_-}.
},
the states $\ket{\mathbf{h},\alpha,\beta+1}$, $\mathbf{h}\in\{0,1\}^N$, form a basis  of $\mathcal{H}$. Similarly, the states $\bra{\alpha,\beta-1,\mathbf{h}}$ $\mathbf{h}\in\{0,1\}^N$, form a basis of  $\mathcal{H}^*$, and we have
\begin{equation}\label{identity}
  \mathbb{I}=\frac{1}{\mathsf{N}(\boldsymbol{\xi}|\alpha,\beta)}\sum_{\mathbf{h}\in\{0,1\}^N}
  V(\xi_1^{(h_1)},\ldots,\xi_N^{(h_N)})\,
  \ket{\mathbf{h},\alpha,\beta+1} \bra{\alpha,\beta-1,\mathbf{h}}.
\end{equation}

\subsection{SoV characterization of the transfer matrix spectrum and eigenstates}

We recall that, for $\alpha$ and $\beta$ satisfying the constraints  \eqref{cond-c+=0} and  \eqref{cond-b+=0}, i.e. given in terms of the $+$ boundary parameters by \eqref{alpha-beta+}, the transfer matrix can be written in a simple form in terms of ${\mathcal{A}}_-(\pm\lambda |\alpha ,\beta-1 )$ or in terms of ${\mathcal{D}}_-(\pm\lambda |\alpha ,\beta+1 )$, see \eqref{expr-transfer}.
We moreover suppose from now on that, for these values of $\alpha$ and $\beta$, the coefficient \eqref{norm-coeff} is non-zero, so that the states $ \ket{\mathbf{h},\alpha,\beta+1}$ form a basis of $\mathcal{H}$, whereas the states $ \bra{\alpha,\beta-1,\mathbf{h}}$  form a basis of  $\mathcal{H}^*$ for $ \mathbf{h}\in\{0,1\}^N$.

It then follows from the form of the action \eqref{act-A-lambda-bra} and \eqref{act-D-lambda-ket} on the SoV states that, for these particular values of $\alpha$ and $\beta$,  the basis $\{ \ket{\mathbf{h},\alpha,\beta+1},\, \mathbf{h}\in\{0,1\}^N\}$  of $\mathcal{H}$, respectively the basis $\{ \bra{\alpha,\beta-1,\mathbf{h}},\, \mathbf{h}\in\{0,1\}^N\}$ of  $\mathcal{H}^*$, separates the variables for the spectral problem of the transfer matrix $\mathcal{T}(\lambda)$ \eqref{expr-transfer} at the points $\pm \xi_n^{(h_n)}$.
More precisely, we can formulate the following result.

\begin{theorem}\label{th-spectr-SoV}
Let us suppose that the boundary matrices $K_-(\lambda)$ and $K_+(\lambda)$ are not both proportional to the identity, and that the inhomogeneity parameters are generic enough, i.e. satisfy \eqref{cond-inhom}.

Then the transfer matrix $\mathcal{T}(\lambda)$ is diagonalisable with simple spectrum. 
Its set of eigenvalues is given by the set of functions $\tau(\lambda)$ such that
\begin{enumerate}
   \item[(i)] $\tau(\lambda)$ is an even entire function of $\lambda$ which satisfies the quasi-periodicity properties
\begin{equation}\label{qper-tau}
   \tau(\lambda +\pi )=\tau(\lambda ),
          \qquad
    \tau(\lambda+\pi \omega )=\left( - e^{-2i\lambda-i\pi\omega }\right) ^{2N+6}\tau(\lambda );
\end{equation}
   \item[(ii)]  its values at the points \eqref{special_points} are given by  $\tau(\xi_{-j}^{(0)}) = \tau_{-j}$, $j=0,1,2,3$, with
\begin{align}\label{tau0}
   &\tau_0=(-1)^N\frac{\ths(2\eta)}{\ths(\eta)}\, \det_q T(0), \qquad
   \tau_{-1}=c^z_-\,c^z_+\,\frac{\ths(2\eta)}{\ths(\eta)}\, \det_q T(\tfrac\pi 2), \\
   &\tau_{-2}= - e^{i\big(2\sum_{n=1}^N\xi_n-(N+3)\eta-\frac{3\pi\omega}2\big)}\,c^x_-\,  c_+^x\, \frac{\ths(2\eta)}{\ths(\eta)}\,  \det_q T(-\tfrac{\pi\omega}2)  , \\
   &\tau_{-3}= e^{i\big(2\sum_{n=1}^N\xi_n-(N+3)\eta-\frac{3\pi\omega}2\big)}\,c^y_-\,  c_+^y\, \frac{\ths(2\eta)}{\ths(\eta)}\,  \det_q T(-\tfrac{\pi+\pi\omega}2).
   \label{tau-3}
\end{align}
\item[(iii)] its values at the points $\xi_n+\tfrac\eta2$ and $\xi_n-\tfrac\eta 2$, $n\in\{1,\ldots,N\}$, are related by
\begin{equation}\label{qdet-xin}
   \tau(\xi_n+\eta/2)\,\tau(\xi_n-\eta/2)
   =-\frac{\det_q K_+(\xi_n)\, \det_q \mathcal{U}_-(\xi_n)}{\ths(2\xi_n+\eta)\, \ths(2\xi_n-\eta)}.
\end{equation}
\end{enumerate}

Let us moreover suppose that the gauge parameters $\alpha$ and $\beta$ satisfy the constraints  \eqref{cond-c+=0} and  \eqref{cond-b+=0}, i.e. are given in terms of the $+$ boundary parameters by \eqref{alpha-beta+}, and that the normalisation coefficient \eqref{norm-coeff} is non-zero.
Then, the one-dimensional right and left eigenstates of $\mathcal{T}(\lambda)$ with eigenvalue $\tau(\lambda)$ can respectively be expressed as
\begin{align}
   &\sum_{\mathbf{h}\in\{0,1\}^N}\prod_{n=1}^N Q(\xi_n^{(h_n)})\ V(\xi_1^{(h_1)},\ldots,\xi_N^{(h_N)})\ \ket{\mathbf{h},\alpha,\beta+1},
   \label{eigen-r-Skl}
\end{align}
and
\begin{align}
   &\sum_{\mathbf{h}\in\{0,1\}^N}\prod_{n=1}^N
   \left[ \left(\frac{\ths(2\xi_n-2\eta)}{\ths(2\xi_n+2\eta)}\, 
   \frac{\mathbf{A}(\xi_n+\frac{\eta}{2})}
         {\mathbf{A}(-\xi_n+\frac{\eta}{2})}
   \right)^{\! h_n} Q(\xi_n^{(h_n)})\right]
   V(\xi_1^{(h_1)},\ldots,\xi_N^{(h_N)})\    \bra{\alpha,\beta-1,\mathbf{h}},
\label{eigen-l-Skl}
\end{align}
in which $\mathbf{A}(\lambda)$ 
is defined as
\begin{align} \label{def-A-gauge}
   &\mathbf{A}(\lambda)=\frac{\ths(2\lambda+\eta)}{\ths(2\lambda)}\,\mathsf{a}_+( \lambda |\boldsymbol{\epsilon_{\alpha^+}})\,     \mathsf{A}_-(\lambda)
   = \frac{\ths(2\lambda+\eta)}{\ths(2\lambda)}\,\mathsf{a}_+( \lambda |\boldsymbol{\epsilon_{\alpha^+}})\,\mathsf{g}_-(\lambda)\, a(\lambda)\, d(-\lambda),
\end{align}
and with coefficients $Q(\xi_n^{(h_n)})$ satisfying \eqref{T-Qdis}.
\end{theorem}

This result nearly coincides with Theorem~3.1 of \cite{NicT24}.
In particular, the first part is just a reformulation of the first part of Theorem~3.1 of \cite{NicT24}, and the conditions \eqref{qper-tau}-\eqref{qdet-xin} directly follow from the transfer matrix properties enunciated in Proposition~3.1 of \cite{NicT24}. The only difference is that we use here the specific basis $\{ \ket{\mathbf{h},\alpha,\beta+1},\, \mathbf{h}\in\{0,1\}^N\}$  of $\mathcal{H}$ and $\{ \bra{\alpha,\beta-1,\mathbf{h}},\, \mathbf{h}\in\{0,1\}^N\}$ of  $\mathcal{H}^*$, for $\alpha,\beta$ fixed by \eqref{cond-c+=0} and  \eqref{cond-b+=0}, to express the eigenstates, instead of the general basis of the form \eqref{new-SoV-states}.
Note however that, up to this particular choice of basis, \eqref{eigen-r-Skl} and \eqref{eigen-l-Skl} have exactly the same form as \eqref{eigen-r} and \eqref{eigen-l}.

The function $\mathbf{A}(\lambda)$  is here related to the function $\mathsf{g}_-$ used in the definition of the SoV states \eqref{ket-state} and \eqref{bra-state} as in \eqref{def-A-gauge}. It is however important to underly that, to each function $\mathbf{A}(\lambda)$ solution of \eqref{def-A}, one can associate through \eqref{def-A-gauge} a function $\mathsf{g}_-(\lambda)$  satisfying \eqref{rel-g-}).

\subsection{Separate states as generalised Bethe vectors}

For some arbitrary function $Q$ on the discrete set of values $\{\xi_n^{(h_n)}, 1\le n\le N\}$, 
let us consider the separate state $\ket{Q}$ defined as
\begin{align}
   &\ket{Q}
   = 
   \sum_{\mathbf{h}\in\{0,1\}^N}\prod_{n=1}^N Q(\xi_n^{(h_n)})\ V(\xi_1^{(h_1)},\ldots,\xi_N^{(h_N)})\ \ket{\mathbf{h},\alpha,\beta+1} .
   \label{ket-Q-eps}
\end{align}
Let us moreover suppose that the function $Q$ can be extended to the complex plane into a function of the following form (that we still denote by $Q$):
\begin{equation}
  \label{Q-form}
   Q(\lambda)=\prod_{n=1}^M\ths(\lambda-\lambda_n)\ths(\lambda+\lambda_n), \qquad (\lambda_1,\ldots,\lambda_M)\in\mathbb{C}^M.
\end{equation}
Then, thanks to the form \eqref{act-Bhat-ket} of the diagonal action of the operator $\widehat{\mathcal{B}}_-(\lambda|\alpha-\beta)$ on the basis vectors \eqref{ket-state}, the state \eqref{ket-Q-eps} can be rewritten in the form of a generalised Bethe vector:
\begin{align}\label{separate-Bethe-r}
   \ket{Q} 
   = \mathsf{c}_{Q,\alpha,\beta}^{(R)} \, \underline{\widehat{\mathcal B}}_{-,M}(\{\lambda_j\}_{j=1}^M|\alpha-\beta+1)\, \ket{\Omega_{\alpha,\beta+1-2M}},
\end{align}
in which we have used the shortcut notation
\begin{align}
  \label{product-Bhat}
  \underline{\widehat{\mathcal{B}}}_{-,M}(\{\lambda _{i}\}_{i=1}^{M}|\alpha -\beta+1)
  &= 
  \widehat{\mathcal{B}}_{-}(\lambda _{1}|\alpha -\beta +1)\cdots 
  \widehat{\mathcal{B}}_{-}(\lambda _{M}|\alpha -\beta +2M-1)
  \nonumber\\
  &= \prod_{j=1\to M} \widehat{\mathcal{B}}_{-}(\lambda_j|\alpha -\beta +2j-1),
\end{align}
and defined the generalised SoV reference state $\ket{\Omega_{\alpha,\beta+1-2M}} $ as
\begin{align}
   &\ket{\Omega_{\alpha,\beta+1-2M}} =\frac{1}{\mathsf{N}(\boldsymbol{\xi}|\alpha,\beta-2M)}\sum_{\mathbf{h}\in\{0,1\}^N} V(\xi_1^{(h_1)},\ldots,\xi_N^{(h_N)})\ \ket{\mathbf{h},\alpha,\beta+1-2M} .
   \label{ref-SoV-r}
\end{align}
The coefficient $\mathsf{c}_{Q,\alpha,\beta}^{(R)}$ is explicitly given as
\begin{align}
   &\mathsf{c}_{Q,\alpha,\beta}^{(R)}=\frac{(-1)^{MN} \,\mathsf{N}(\boldsymbol{\xi}|\alpha,\beta-2M)} 
   {\prod_{j=1}^M \ths(2 \lambda_j-\eta)\, \mathsf{b}_-(\alpha-\beta+N+2j-1) }.
   \label{c-SoV-ABA-r}
\end{align}
%
Note that the form \eqref{separate-Bethe-r}-\eqref{c-SoV-ABA-r} of the separate state $\ket{Q}$ is completely general when expressed in the basis given by the states $\ket{\mathbf{h},\alpha,\beta+1}$, provided that $Q$ is a function of the form \eqref{Q-form}. Such a rewriting is valid for arbitrary $\alpha$ and $\beta$ such that $\mathsf{N}(\boldsymbol{\xi}|\alpha,\beta),\mathsf{N}(\boldsymbol{\xi}|\alpha,\beta-2M)\not= 0$.

It happens that, under some particular relation of $\alpha+\beta$ with the `$-$'  boundary parameters and some related choice of the function $\mathsf{g}_-$ in the definition of the SoV states \eqref{ket-state}, the special separate state $\ket{\Omega_{\alpha,\beta+1-2M}}$ \eqref{ref-SoV-r} can be identified with the usual gauge transformed reference state used notably in \cite{FanHSY96} (see also \cite{FadT79}) to construct  some of the transfer matrix eigenstates via generalised Algebraic Bethe Ansatz. More precisely, if we define the following gauge transformed reference state 
\begin{align}\label{ref-state-gauge}
  \ket{\eta,x} &=
  S_{1\ldots N}(\boldsymbol{\xi} |\alpha,x-\alpha+1-N)\ket{0}
  \nonumber\\
  &
  =\mathop{\otimes}_{n=1}^N \big[ Y_{x+n-N}(-\xi_n) \big]_n
    =\mathop{\otimes}_{n=1}^N \begin{pmatrix} \thd_2(\xi_n+(x+n-N)\eta) \\ \thd_3(\xi_n+(x+n-N)\eta) \end{pmatrix}_{\! n},
\end{align}
we can formulate the following result:

\begin{prop}\label{prop-ref-state}
   Let us suppose that the inhomogeneity parameters are generic, or at least satisfy \eqref{cond-inhom}.
   For a given integer $M$, let us moreover suppose that there exists $\boldsymbol{\epsilon_{\alpha^-}}=(\epsilon_{\alpha_1^-},\epsilon_{\alpha_2^-},\epsilon_{\alpha_3^-})\in\{1,-1\}^3$, with $\epsilon_{\alpha_1^-}\epsilon_{\alpha_2^-}\epsilon_{\alpha_3^-}=1$, such that
   \begin{equation}\label{cond-ref-states}
      (\alpha+\beta+N-2M-1)\eta=\sum_{\ell=1}^3\epsilon_{\alpha_\ell^-}\alpha_\ell^- \mod(2\pi,2\pi\omega),
   \end{equation}
   and let us consider the corresponding SoV basis $\{\ket{\mathbf{h},\alpha,\beta+1-2M}, \mathbf{h}\in\{0,1\}^N\}$, constructed as in \eqref{ket-state} by fixing 
   \begin{equation}\label{fix-g-}
      \mathsf{g}_-(\lambda)=(-1)^N\mathsf{a}_-(\lambda|\boldsymbol{\epsilon_{\alpha^-}}).
   \end{equation}
   Then, the corresponding SoV generalised reference state $\ket{\Omega_{\alpha,\beta-2M+1}}$ \eqref{ref-SoV-r}    coincides with the gauge transformed reference state $\ket{\eta, \alpha+\beta+N-2M-1}$ \eqref{ref-state-gauge}: 
   \begin{equation}\label{equality-ref-state}
      \ket{\Omega_{\alpha,\beta-2M+1}}=\ket{\eta, \alpha+\beta+N-2M-1}.
   \end{equation}
\end{prop}

\begin{proof}
It is easy to show, using the boundary-bulk decomposition of the operator $\mathcal{A}_-(\lambda|\alpha,\beta-2M-1)$ issued from \eqref{Boundary-bulkU} for $\gamma$ and $\delta$ such that $\gamma+\delta=\alpha+\beta+N-2M-1$, and the actions \eqref{actC-ref}-\eqref{actD-ref} of the gauge bulk operators on the gauge reference state \eqref{ref-state-gauge}, that, if  $\gamma+\delta=\alpha+\beta+N-2M-1$ satisfies the condition \eqref{cond-c-=0} (i.e. under \eqref{cond-ref-states}),  $\ket{\eta,\alpha+\beta+N-2M-1} $ is an eigenstate of $\mathcal{A}_-(\lambda|\alpha,\beta-2M-1)$:
\begin{multline}\label{act-A-ref}
  \mathcal{A}_-(\lambda|\alpha,\beta-2M-1)\,\ket{\eta,\alpha+\beta+N-2M-1}
  \\
  =(-1)^N \mathsf{a}_-(\lambda|\boldsymbol{\epsilon^-})\,  a(\lambda)\, d(-\lambda) \,\ket{\eta,\alpha+\beta+N-2M-1}.
\end{multline}
This enables us to compute the scalar products of the state $\ket{\eta, \alpha+\beta+N-2M-1}$ with any of the basis co-vectors $\bra{\alpha,\beta-2M-1,\mathbf{h}}$:
\begin{align}\label{SP-ref-state}
  &\moy{\alpha, \beta-2M-1,\mathbf{h} \, | \, \eta, \alpha+\beta+N-2M-1}
  \nonumber\\
  &\quad=\bra{0}\, S^{-1}_{1\ldots N}(\boldsymbol{\xi} |\alpha,\beta-2M)
   \prod_{n=1}^N\left(\frac{\mathcal{A}_-(\frac\eta 2-\xi_n |\alpha,\beta-2M-1)}{\mathsf{A}_-(\frac\eta 2-\xi_n)}\right)^{\! 1-h_n}\, \ket{\eta, \alpha+\beta+N-2M-1}
   \nonumber\\
   &\quad = \prod_{n=1}^N\left( (-1)^N\frac{\mathsf{a}_-(\frac\eta 2-\xi_n |\boldsymbol{\epsilon_{\alpha^-}})}{\mathsf{g}_-(\frac\eta 2-\xi_n)}\right)^{\! 1-h_n}
   =1
   \nonumber\\
   &\quad = \moy{\alpha, \beta-2M-1,\mathbf{h} \, | \, \Omega_{\alpha,\beta-2M+1}},
\end{align}
in which we have also used \eqref{fix-g-}. The identity \eqref{SP-ref-state} holds for any $\mathbf{h}\in\{0,1\}^N$, which therefore proves \eqref{equality-ref-state} under the conditions \eqref{cond-ref-states}-\eqref{fix-g-}.
\end{proof}

\begin{rem}
As  in the XXZ case, see \cite{NicT22}, some similar Bethe-type representation can also be written for left separate states $\bra{Q}$ of the form \eqref{eigen-l-Skl}, still for $Q$ of the form \eqref{Q-form}. We do not specify it here since we shall not use it.
\end{rem}



\subsection{The case with a constraint : spectrum and eigenstates by functional $TQ$-equation}

The SoV characterisation of the spectrum and eigenstates of Theorem~\ref{th-spectr-SoV} is very general, and is valid for arbitrary (not too specific) boundary fields. From now on, as in \cite{NicT24}, and following our previous works concerning the XXZ case \cite{NicT22,NicT23}, we shall particularise our study to the special case in which the boundary parameters are related by one constraint:
\begin{equation}\label{const-TQ}
   \sum_{\sigma=\pm}\sum_{i=1}^3 \epsilon_i^\sigma\alpha_i^\sigma=(N-2M-1)\eta,
\end{equation}
for some integer $M$ and some choice of $\boldsymbol{\varepsilon}\equiv(\epsilon_1^+,\epsilon_2^+,\epsilon_3^+,\epsilon_1^-,\epsilon_2^-,\epsilon_3^-)\in\{+1,-1\}^6$ such that $\prod_{\sigma=\pm}\prod_{i=1}^3\epsilon_i^\sigma=1$. 

We recall that, in that case \cite{YanZ06,NicT24}, similarly as what happens in the XXZ case \cite{Nep02,Nep04,KitMNT18}, part of the spectrum and eigenstates can be described in terms of usual $TQ$-equations, i.e. of usual Bethe equations:

\begin{prop}\label{prop-TQ}\cite{YanZ06,NicT24}
Let us suppose that the hypothesis of Theorem~\ref{th-spectr-SoV} are satisfied, and that the constraint \eqref{const-TQ} holds for some integer $M$ and some choice of $\boldsymbol{\varepsilon}\equiv(\epsilon_1^+,\epsilon_2^+,\epsilon_3^+,\epsilon_1^-,\epsilon_2^-,\epsilon_3^-)\in\{+1,-1\}^6$ such that $\prod_{\sigma=\pm}\prod_{i=1}^3\epsilon_i^\sigma=1$. 
Let $\tau$ be an entire function such that there exists $Q$ of the form \eqref{Q-form} satisfying with  $\tau$  the $TQ$-equation
\begin{equation}\label{hom-TQ}
   \tau(\lambda)\, Q(\lambda) = \mathbf{A}_{\boldsymbol{\varepsilon}}(\lambda)\, Q(\lambda-\eta)+\mathbf{A}_{\boldsymbol{\varepsilon}}(-\lambda)\, Q(\lambda+\eta),
\end{equation}
in which we have defined
\begin{equation}\label{def-A_eps}
   \mathbf{A}_{\boldsymbol{\varepsilon}}(\lambda)
   =(-1)^N\, \frac{\ths(2\lambda+\eta)}{\ths(2\lambda)}\,\mathsc{a}_{\boldsymbol{\varepsilon}}(\lambda)\, a(\lambda)\, d(-\lambda),
%
\quad \text{with}\quad
%
   \mathsc{a}_{\boldsymbol{\varepsilon}}(\lambda)
   =\prod_{\sigma=\pm}\prod_{i=1}^3\frac{\ths(\lambda-\frac\eta 2-\epsilon_i^\sigma\alpha_i^\sigma)}{\ths(-\epsilon_i^\sigma\alpha_i^\sigma)}.
\end{equation}
Then $\tau$ is an eigenvalue of the transfer matrix $\mathcal{T}$.  Moreover,  the associated $Q$-solution of \eqref{hom-TQ} of the form \eqref{Q-form} is unique, and the corresponding right and left eigenvectors of $\mathcal{T}$ are given in terms of this function $Q$ by \eqref{eigen-r-Skl} and \eqref{eigen-l-Skl} with $\alpha$ and $\beta$ given by \eqref{cond-c+=0}-\eqref{cond-b+=0}\footnote{They can of course also be given by \eqref{eigen-r} and \eqref{eigen-l}, still for $\mathbf{A}(\lambda)\equiv\mathbf{A}_{\boldsymbol{\varepsilon}}(\lambda)$.} for $\mathbf{A}(\lambda)\equiv\mathbf{A}_{\boldsymbol{\varepsilon}}(\lambda)$.
\end{prop}


Moreover, if the constraint \eqref{const-TQ} is satisfied for some boundary parameters with some choice of $\boldsymbol{\varepsilon}$ and $M$, it is also satisfied, for the same boundary parameters, with the choice $\boldsymbol{\varepsilon'}=-\boldsymbol{\varepsilon}$ and $M'=N-1-M$. It was conjectured in \cite{YanZ06}, following a similar conjecture based on numerics in the XXZ case \cite{NepR03,NepR04}, that one obtains the whole spectrum by combining these two sectors $(M,\boldsymbol{\varepsilon})$ and $(M',\boldsymbol{\varepsilon'})$. If this is true, this means that we can also obtain the whole set of eigenstates from \eqref{eigen-r-Skl} and \eqref{eigen-l-Skl}  with the corresponding functions $Q$ of degrees $2M$ and $2M'$ and for $\mathbf{A}$ given by $\mathbf{A}_{\boldsymbol{\varepsilon}}$ and $\mathbf{A}_{\boldsymbol{\varepsilon'}}$ respectively. It means in particular (which is enough for our present purpose) that the ground state itself can be described in this framework.

We also recall that,  as shown in our previous work \cite{NicT24}, the scalar products of the separate states admit, under  the constraint \eqref{const-TQ}, a Slavnov-type determinant representation (see Theorem~4.1 of \cite{NicT24}, which is also valid for separate states built from the present Vertex-IRF SoV basis).

Finally, as shown in the previous subsection, under the additional conditions \eqref{cond-ref-states}-\eqref{fix-g-}, the corresponding left separate states \eqref{ket-Q-eps}, for any $Q$ of the form \eqref{Q-form}, can be written as generalised Bethe states given by multiple action of the operator $\widehat{\mathcal B}$ on the gauge transformed reference state $\ket{\eta, \alpha+\beta+N-2M-1}$ \eqref{ref-state-gauge}:
\begin{align}\label{separate-Bethe}
   \ket{Q} 
   = \mathsf{c}_{Q,\alpha,\beta}^{(R)} \, \underline{\widehat{\mathcal B}}_{-,M}(\{\lambda_j\}_{j=1}^M|\alpha-\beta+1)\, \ket{\eta, \alpha+\beta+N-2M-1},
\end{align}
with coefficient $\mathsf{c}_{Q,\alpha,\beta}^{(R)}$ given by \eqref{c-SoV-ABA-r}.
Note that the condition \eqref{fix-g-}, together with \eqref{def-A-gauge}, means that we have fixed the signs in \eqref{const-TQ} as
\begin{align}\label{choice-eps+-}
  \boldsymbol{\varepsilon}=(\epsilon_{\alpha_1^+},\epsilon_{\alpha_2^+},\epsilon_{\alpha_3^+},\epsilon_{\alpha_1^-},\epsilon_{\alpha_2^-},\epsilon_{\alpha_3^-}),
\end{align}
i.e. that
\begin{equation}
   \mathsc{a}_{\boldsymbol{\varepsilon}}(\lambda)= \mathsf{a}_+( \lambda |\boldsymbol{\epsilon_{\alpha^+}})\, \mathsf{a}_-( \lambda |\boldsymbol{\epsilon_{\alpha^-}}).
\end{equation}
Note moreover that the condition \eqref{cond-ref-states} is indeed compatible with \eqref{cond-c+=0}, provided the constraint \eqref{const-TQ} is satisfied with signs given by \eqref{choice-eps+-}, which we suppose from now on.

All these ingredients will enable us to compute, similarly as in the XXZ case \cite{NicT22,NicT23}, some of the elementary building blocks for the zero-temperature correlation functions, in cases for which the ground state of the model is effectively among the states described by a $Q$-solution of \eqref{hom-TQ} of the form \eqref{Q-form}.

\section{Action of a basis of local operators on boundary Bethe states}
\label{sec-act}

We from now on suppose that we are in a situation in which the ground states itself can be expressed as a separate state of the form \eqref{ket-Q-eps}-\eqref{Q-form} (with $Q$ solution of the $TQ$-equation \eqref{hom-TQ}), i.e. as a generalised Bethe state of the form \eqref{separate-Bethe}.

Therefore, so as to compute correlation functions at zero temperature, the question is now to characterise the action of local operators  on the generalised Bethe states
\begin{align}\label{boundary-Bethe}
   \underline{\widehat{\mathcal B}}_{-,M}(\{\lambda_j\}_{j=1}^M|\alpha-\beta+1)\, \ket{\eta, \alpha+\beta+N-2M-1}.
\end{align}
The strategy is similar to  the one used in the XXZ diagonal case \cite{KitKMNST07} and in our previous works concerning the XXZ open spin chain with non-diagonal boundary conditions \cite{NicT22,NicT23}:
\begin{enumerate}[label=\alph*)]
  \item we decompose the generalised boundary Bethe states \eqref{boundary-Bethe} into sums of generalised bulk Bethe states;
  \item we use the reconstruction of local operators in terms of the elements of the bulk monodromy matrix so as to compute the action of these local operators on the generalised bulk Bethe states;
  \item we reconstruct the resulting sum of generalised bulk Bethe states in terms of generalised boundary Bethe states, i.e. in terms of Sklyanin's separate states.
\end{enumerate}
However,  the combinatorial complexity of the generalised boundary/bulk Bethe states, which involve the Vertex-IRF transform, makes it quite cumbersome to act on them with the standard basis of local operators. Therefore, as in \cite{NicT22,NicT23}, we identify a particular basis of the space $\End(\otimes _{n=1}^{m}\mathcal{H}_n)$ of local operators on the first $m$ sites of the chain, so that the elements of this basis act in a relatively simple combinatorial way on the generalised Bethe states. This basis is itself appropriately gauged transformed from the usual one.

More precisely, as in \cite{NicT22,NicT23}, we define the following gauge-transformed local elementary operators at a given site $n$ of the chain:
\begin{equation}\label{gauge-op-n}
  E_n^{i,j}(\lambda|(a,b),(\bar a,\bar b))=S_n(-\lambda|\bar a,\bar b)\, E_n^{i,j}\,S_n^{-1}(-\lambda|a,b),
  \qquad i,j\in\{1,2\}.
\end{equation}
Here $E^{i,j}\in\End(\mathbb{C}^2)$ stands for the usual elementary matrix with elements $(E^{i,j})_{k,\ell}=\delta_{i,k}\delta_{j,\ell}$, and $\lambda$, $a$, $b$, $\bar a$, $\bar b$ are parameters such that the Vertex-IRF matrices $S(-\lambda| a, b)$ and $S(-\lambda|\bar a,\bar b)$ are invertible.
Hence, for given such values of these parameters, the local operators \eqref{gauge-op-n} correspond to four different linear combinations of the usual local operators $E_n^{i,j}$ at site $n$\footnote{The latter are written explicitly as
\begin{align}
  &E^{1,1}_n(\lambda|(a,b),(\bar a,\bar b)) =\big[
  \thd_2(\lambda+(\bar a+\bar b)\eta)\thd_3(\lambda+(a-b)\eta)\, E^{1,1}_n
   -\thd_2(\lambda+(\bar a+\bar b)\eta)\thd_2(\lambda+(a-b)\eta)\, E^{1,2}_n \nonumber\\
  &\quad +\thd_3(\lambda+(\bar a+\bar b)\eta)\thd_3(\lambda+(a-b)\eta)\, E^{2,1}_n 
   -\thd_3(\lambda+(\bar a+\bar b)\eta)\thd_2(\lambda+(a-b)\eta)\, E^{2,2}_n
  \big]\frac{-1}{\ths(\lambda+a\eta)\ths(b\eta)},
  \\
  &E^{1,2}_n(\lambda|(a,b),(\bar a,\bar b)) =\big[
  -\thd_2(\lambda+(\bar a+\bar b)\eta)\thd_3(\lambda+(a+b)\eta)\, E^{1,1}_n
   +\thd_2(\lambda+(\bar a+\bar b)\eta)\thd_2(\lambda+(a+b)\eta)\, E^{1,2}_n \nonumber\\
  &\quad -\thd_3(\lambda+(\bar a+\bar b)\eta)\thd_3(\lambda+(a+b)\eta)\, E^{2,1}_n 
   +\thd_3(\lambda+(\bar a+\bar b)\eta)\thd_2(\lambda+(a+b)\eta)\, E^{2,2}_n
  \big]\frac{-1}{\ths(\lambda+a\eta)\ths(b\eta)},
  \\
  &E^{2,1}_n(\lambda|(a,b),(\bar a,\bar b)) =\big[
  \thd_2(\lambda+(\bar a-\bar b)\eta)\thd_3(\lambda+(a-b)\eta)\, E^{1,1}_n
   -\thd_2(\lambda+(\bar a-\bar b)\eta)\thd_2(\lambda+(a-b)\eta)\, E^{1,2}_n \nonumber\\
  &\quad +\thd_3(\lambda+(\bar a-\bar b)\eta)\thd_3(\lambda+(a-b)\eta)\, E^{2,1}_n 
   -\thd_3(\lambda+(\bar a-\bar b)\eta)\thd_2(\lambda+(a-b)\eta)\, E^{2,2}_n
  \big]\frac{-1}{\ths(\lambda+a\eta)\ths(b\eta)},
  \\
  &E^{2,2}_n(\lambda|(a,b),(\bar a,\bar b)) =\big[
  -\thd_2(\lambda+(\bar a-\bar b)\eta)\thd_3(\lambda+(a+b)\eta)\, E^{1,1}_n
   +\thd_2(\lambda+(\bar a-\bar b)\eta)\thd_2(\lambda+(a+b)\eta)\, E^{1,2}_n \nonumber\\
  &\quad -\thd_3(\lambda+(\bar a-\bar b)\eta)\thd_3(\lambda+(a+b)\eta)\, E^{2,1}_n 
   +\thd_3(\lambda+(\bar a-\bar b)\eta)\thd_2(\lambda+(a+b)\eta)\, E^{2,2}_n
  \big]\frac{-1}{\ths(\lambda+a\eta)\ths(b\eta)}.
\end{align}
}.

Still following \cite{NicT22,NicT23}, we define, for given values of gauge parameters $\alpha$ and $\beta$, and given $m$-tuples $\boldsymbol{\epsilon}\equiv (\epsilon _{1},\ldots ,\epsilon _{m})\in\{1,2\}^m$ and $\boldsymbol{\epsilon'}\equiv (\epsilon _{1}^{\prime },\ldots ,\epsilon _{m}^{\prime })\in\{1,2\}^m$, new parameters $a_{n},\bar{a}_{n},b_{n},\bar{b}_{n}$, $1\leq n\leq m$ as
\begin{alignat}{2}
& a_{n}=\alpha +1,\qquad  & & b_{n}=\beta -\sum_{r=1}^{n}(-1)^{\epsilon
_{r}},  \label{Gauge.Basis-1} \displaybreak[0]\\
& \bar{a}_{n}=\alpha -1,\qquad  & & \bar{b}_{n}=\beta
+\sum_{r=n+1}^{m}(-1)^{\epsilon _{r}^{\prime }}-\sum_{r=1}^{m}(-1)^{\epsilon
_{r}}=b_{n}+2\tilde{m}_{n+1},  \label{Gauge.Basis-2}
\end{alignat}
with 
\begin{equation}\label{tilde-m_n}
\tilde{m}_{n}=\sum_{r=n}^{m}(\epsilon _{r}^{\prime }-\epsilon
_{r})=\sum_{r=n}^{m}\frac{(-1)^{\epsilon _{r}^{\prime }}-(-1)^{\epsilon _{r}}}{2}.
\end{equation}
With the notations \eqref{Gauge.Basis-1}-\eqref{Gauge.Basis-2}, we introduce the following tensor product of gauge transformed local operators \eqref{gauge-op-n} along the first $m$ sites of the chain:
\begin{equation}\label{op-E_m}
  \barE_{m}^{\boldsymbol{\epsilon'},\boldsymbol{\epsilon}}(\alpha,\beta)
   \equiv \prod_{n=1}^{m}E_{n}^{\epsilon_n',\epsilon _n}(\xi_n |(a_n,b_n),(\bar{a}_n,\bar{b}_n))
   \in \End(\otimes _{n=1}^{m}\mathcal{H}_n).
\end{equation}
Then, we can show the following analog of Proposition~4.3 of \cite{NicT22}:

\begin{prop}\label{prop-basis-E_m}
For two given parameters $\alpha,\beta$, the set
\begin{equation}
\mathbb{E}_{m}(\alpha ,\beta )=\left\{  \barE_m^{\boldsymbol{\epsilon'},\boldsymbol{\epsilon}}(\alpha,\beta)\ \mid \ 
\boldsymbol{\epsilon},\boldsymbol{\epsilon'}\in \{1,2\}^{m}\right\} ,
\label{Local-Basis}
\end{equation}
of elements of the form \eqref{op-E_m} of $\End(\otimes _{n=1}^{m}\mathcal{H}_n)$, with  $a_{n},\bar{a}_{n},b_{n},\bar{b}_{n}$, $1\leq n\leq m$, fixed in terms $\alpha,\beta$, $\boldsymbol{\epsilon}$ and $\boldsymbol{\epsilon'}$  as in \eqref{Gauge.Basis-1}-\eqref{Gauge.Basis-2}, defines a basis of $\End(\otimes _{n=1}^{m}\mathcal{H}_n)$ except for a finite number of values of $\alpha$ and $\beta \mod (\tfrac{2\pi}\eta,\tfrac{2\pi\omega}\eta)$.
\end{prop}

In other words, it means that any operator acting only on the first $m$ sites of the chain can be expressed as a linear combination of operators of the form \eqref{op-E_m}. In the following, we compute the action of these operators on the generalised bulk and boundary Bethe states, and we shall in section~\ref{sec-corr} compute the matrix elements of these operators in the ground state.

\subsection{Boundary-bulk decomposition of boundary Bethe states}
\label{sec-Bb-Bethe}

The SoV separate states \eqref{ket-Q-eps}-\eqref{Q-form} can be alternatively expressed as generalised boundary Bethe states. Following the strategy mentioned above, we decompose these states as a sum over bulk Bethe states. This decomposition follows from the boundary-bulk decomposition  \eqref{boundary-bulk-B2} of the boundary operator $\widehat{\mathcal B}_-(\lambda|\alpha-\beta)$. 

Let us introduce the following shortcut notation for a product of gauge $B$-bulk operators:
\begin{align}
  \underline{B}_M(\{\mu_i\}_{i=1}^M|x-1,y+1) 
  &= B(\mu_1|x-1,y+1)\, B(\mu_2|x-2,y+2)\ldots B(\mu_M|x-M,y+M) \nonumber\\
  &=\prod_{j=1\to M} B(\mu_j|x-j,y+j) .
\end{align}

\begin{proposition}
\label{prop-new-BB}
Let $\{\lambda_1,\ldots,\lambda_M\}$ be an arbitrary set of spectral parameters  and  $\alpha,\beta$ be arbitrary gauge parameters.
Then, for $\gamma$ and $\delta$ satisfying \eqref{cond-c-=0} and \eqref{cond-b-=0} in terms of the boundary parameters $\alpha^-_\ell$, one can express  the gauged boundary Bethe state given by the action of \eqref{product-Bhat} on the gauge reference state $\ket{\eta,\gamma+\delta}$ \eqref{ref-state-gauge} as the following sum of generalised bulk Bethe states:
\begin{multline}\label{Bb-state}
   \underline{\widehat{\mathcal{B}}}_{-,M}(\{\lambda _{i}\}_{i=1}^{M}|\alpha -\beta+1)\,
   \ket{\eta, \gamma +\delta }
   = 
   \sum_{\boldsymbol{\sigma}\in\{-1,+1\}^M} 
   H_{\boldsymbol{\sigma}}(\{\lambda _{i}\}_{i=1}^{M}|\alpha-\beta,\gamma,\delta) 
    \\
 \times  \underline{B}_M(\{\lambda_i^{( \sigma ) }\}_{i=1}^M|\gamma -\delta -1,\alpha-\beta )\,
 \ket{\eta,\gamma +\delta +M} .
\end{multline}
The coefficient of this decomposition is
\begin{align}\label{H_sigma}
   H_{\boldsymbol{\sigma}}(\{\lambda _{i}\}_{i=1}^{M}|\alpha-\beta,\gamma,\delta)
   &=\prod_{n=1}^M
   \frac{ \bar Y_{\gamma +\delta+n-1-N}\big(\lambda _n^{(\sigma)}-\tfrac\eta 2\big) \, Y_{\alpha-\beta+2M-n}\big(\lambda _n^{(\sigma)}-\tfrac\eta 2\big)}
   {\ths(\eta (\delta +n)) \, \ths\!\big(\lambda_n^{( \sigma) }+\frac\eta 2+\eta \gamma\big)}\
      H^-_{\boldsymbol{\sigma}}(\{\lambda _{i}\}_{i=1}^{M}|{\boldsymbol{\epsilon_{\alpha^-}}})
      \nonumber\\
    &=\prod_{n=1}^M\frac{\ths\!\big(\tfrac{\alpha-\beta-\gamma -\delta+1+N+2(M-n)}2 \eta\big)\ths\!\big(\lambda _n^{(\sigma)}-\tfrac{\alpha-\beta+\gamma +\delta+2M-N}2 \eta\big)} {\ths(\eta (\delta +n)) \, \ths\!\big(\lambda_n^{( \sigma) }+\frac\eta 2+\eta \gamma\big)}\,  
    \nonumber\\
    &\hspace{6.7cm}\times
      H^-_{\boldsymbol{\sigma}}(\{\lambda _{i}\}_{i=1}^{M}|{\boldsymbol{\epsilon_{\alpha^-}}}),
 \end{align}
with
\begin{multline}
   H^-_{\boldsymbol{\sigma}}(\{\lambda _{i}\}_{i=1}^{M}|{\boldsymbol{\epsilon_{\alpha^-}}})
   =\prod_{n=1}^{M} \left[  -\sigma_n(-1)^N a(-\lambda_n^{(\sigma)})\,\frac{\ths(2\lambda_n-\eta)}{\ths(2\lambda_n)}\, \mathsf{a}_-(-\lambda _{n}^{( \sigma) }|{\boldsymbol{\epsilon_{\alpha^-}}})\right]
   \\
   \times
   \prod_{1\leq a<b\leq M}\frac{\ths (\lambda_{a}^{( \sigma) }+\lambda _{b}^{( \sigma ) }+\eta )}{\ths (\lambda _{a}^{( \sigma ) }+\lambda _{b}^{( \sigma) })}  .
     \label{H_sigma-bis}
\end{multline}
In all these expressions,  we have used the shortcut notation $\lambda_n^{(\sigma)}\equiv \sigma_n\lambda_n$ for $n\in\{1,\ldots,M\}$.
\end{proposition}

\begin{proof}
This proposition is the elliptic analog of Proposition~3.2 of \cite{NicT23} (see also the particular case considered in \cite{NicT22}) and can be proved along the same lines.
As in \cite{NicT22, NicT23}, we proceed by induction on $M$. 

For $M=1$, \eqref{Bb-state} follows straightforwardly from \eqref{boundary-bulk-B2} and \eqref{actD-ref}.
Supposing that it also holds for a given $M\ge 1$, we can write
\begin{multline}\label{inductionBb-M}
    \widehat{\mathcal{B}}_{-}(\lambda _{M+1}|\alpha -\beta +1)\,
    \underline{\widehat{\mathcal{B}}}_{-,M}(\{\lambda _{i}\}_{i=1}^{M}|\alpha -\beta+3)\,  \ket{\eta, \gamma +\delta }
    \\
    =\frac{(-1)^N}{\ths(\eta(\delta+1))}\frac{\ths(2\lambda_{M+1}-\eta)}{\ths(2\lambda_{M+1})}
       \sum_{\boldsymbol{\sigma}\in\{\pm 1\}^M}
   H_{\boldsymbol{\sigma}}(\{\lambda _{i}\}_{i=1}^{M}|\alpha-\beta+2,\gamma,\delta) 
 \\
 \times 
 \sum_{\sigma_{M+1}=\pm} \sigma_{M+1}\, \frac{\mathsf{a}_-(-\lambda_{M+1}^{(\sigma)}|{\boldsymbol{\epsilon_{\alpha^-}}})}{\ths(\lambda_{M+1}^{(\sigma)}+\frac\eta 2+\eta \gamma)}\, B(\lambda_{M+1}^{(\sigma)} |\gamma-\delta -1,\alpha -\beta) \, D(-\lambda_{M+1}^{(\sigma)} |\gamma +\delta ,\alpha -\beta+1 ) \\
 \times \underline{B}_M(\{\lambda_i^{( \sigma ) }\}_{i=1}^M|\gamma -\delta -1,\alpha-\beta+2 )\,
\ket{\eta,\gamma +\delta +M}.
\end{multline}
The action, given by \eqref{actD-ref}, of $D(-\lambda _{M+1}^{(\sigma) }|\gamma +\delta ,\alpha -\beta -1)$ on $\underline{B}_M(\{\lambda_i^{( \sigma ) }\}_{i=1}^M|\gamma -\delta -1,\alpha-\beta+2 )\,
\ket{\eta,\gamma +\delta +M}$ in \eqref{inductionBb-M} produces a direct term and a sum of indirect terms. 
We shall now show that the contribution of the sum of  indirect terms vanishes.

For a given $a\in\{1,\ldots,M\}$, let us consider the terms which result into the following vector:
\begin{equation}\label{vector-ind}
  \underline{B}_{M+1}( (\{\lambda_i^{( \sigma ) }\}_{i=1}^M\setminus\{\lambda_a^{(\sigma)}\})\cup\{\lambda_{M+1},-\lambda_{M+1}\}|\gamma -\delta -1,\alpha-\beta )\,
\ket{\eta,\gamma +\delta +M+1}.
\end{equation}
The  indirect action of $D(-\lambda _{M+1}^{(\sigma) }|\gamma +\delta ,\alpha -\beta -1)$ in \eqref{inductionBb-M} produces a term proportional to this vector:
\begin{multline}
  D(-\lambda _{M+1}^{(\sigma) }|\gamma +\delta ,\alpha -\beta -1)\, \underline{B}_M(\lambda _{M}^{( \sigma ) }|\gamma -\delta -1,\alpha-\beta+2 ) \,
  \ket{\eta,\gamma +\delta +M}\restrict{\substack{\text{a-th term of} \\ \text{indirect action}}}
  \\
=
  -a(\lambda _a^{(\sigma)})\,
  \ths\!\big(-\lambda _a^{(\sigma)}-\tfrac{\alpha-\beta+\gamma +\delta+2M+2-N}2 \eta\big)\ths\!\big(\tfrac{\alpha-\beta-\gamma -\delta+1+N}2 \eta\big)
  \\
  \qquad\quad \times
  \frac{\ths(-\lambda _{M+1}^{(\sigma) }- \lambda _a^{(\sigma)}+\eta (\delta +1))}{\ths(\eta (\delta +M+1))\ths(\lambda^{( \sigma) }_{M+1}+\lambda^{( \sigma) }_a) }
   \frac{\prod_{j=1}^M\ths(\lambda^{( \sigma) }_a-\lambda^{( \sigma) }_j-\eta)}{\prod_{\substack{j=1 \\ j\neq a}}^M\ths(\lambda^{( \sigma) }_a-\lambda^{( \sigma) }_j)} \
  \\
  \times
  \underline{B}_M((\{\lambda_i^{( \sigma ) }\}_{i=1}^M\setminus\{\lambda_a^{(\sigma)}\})\cup\{-\lambda _{M+1}^{(\sigma) }\}|\gamma-\delta -2,\alpha -\beta +1)  \,  \ket{\eta,\gamma+\delta+M+1},
\end{multline}
Hence, summing over $\sigma_a$ and $\sigma_{M+1}$, we end up with a contribution proportional to (factors in $\lambda _a^{( \sigma) }$ and $\lambda _{M+1}^{( \sigma) }$):
\begin{multline}
  \sum_{\sigma _{a}=\pm 1,\sigma _{M+1}=\pm 1} \sigma_{M+1}\,\sigma_a\,
  a(-\lambda _a^{( \sigma) })\, a(\lambda _a^{(\sigma)})\,
  \frac{\mathsf{a}_-(-\lambda_{M+1}^{(\sigma)}|{\boldsymbol{\epsilon_{\alpha^-}}})}{\ths(\lambda_{M+1}^{(\sigma)}+\frac\eta 2+\eta \gamma)}\,
  \frac{\mathsf{a}_-(-\lambda_{a}^{(\sigma)}|{\boldsymbol{\epsilon_{\alpha^-}}})}{\ths(\lambda_{a}^{(\sigma)}+\frac\eta 2+\eta \gamma)}
  \\
  \times
  \ths\!\big(\lambda _a^{(\sigma)}-\tfrac{\alpha-\beta+\gamma +\delta+2M+2-N}2 \eta\big)\,
  \ths\!\big(-\lambda _a^{(\sigma)}-\tfrac{\alpha-\beta+\gamma +\delta+2M+2-N}2 \eta\big)
  \\
  \times  
     \frac{\ths(\lambda _{M+1}^{(\sigma) }+\lambda _a^{(\sigma)}-\eta (\delta +1))}{\ths(\lambda^{( \sigma) }_{M+1}+\lambda^{( \sigma) }_a) }
   \prod_{j\not= a}\frac{\ths (\lambda_{a}^{( \sigma) }+\lambda _j^{( \sigma ) }+\eta )}{\ths (\lambda _{a}^{( \sigma ) }+\lambda _j^{( \sigma) })}
   \frac{\ths(\lambda^{( \sigma) }_a-\lambda^{( \sigma) }_j-\eta)}{\ths(\lambda^{( \sigma) }_a-\lambda^{( \sigma) }_j)},
\end{multline}
so that the coefficient in front of the vector \eqref{vector-ind} contains the following factor:
\begin{multline*}
  \sum_{\sigma _{a}=\pm 1,\sigma _{M+1}=\pm 1}\hspace{-4mm} \sigma_{M+1}\,\sigma_a\,
   \frac{\mathsf{a}_-(-\lambda_{M+1}^{(\sigma)}|{\boldsymbol{\epsilon_{\alpha^-}}})}{\ths(\lambda_{M+1}^{(\sigma)}+\frac\eta 2+\eta \gamma)}\,
  \frac{\mathsf{a}_-(-\lambda_{a}^{(\sigma)}|{\boldsymbol{\epsilon_{\alpha^-}}})}{\ths(\lambda_{a}^{(\sigma)}+\frac\eta 2+\eta \gamma)}\,
  \frac{\ths (\lambda _{M+1}^{( \sigma) }+\lambda _{a}^{(\sigma) }-\eta (\delta +1))}{\ths (\lambda _{M+1}^{( \sigma) }+\lambda _{a}^{( \sigma ) })} 
  \\
  \hspace{-1mm}\propto\hspace{-2mm}
   \sum_{\sigma _{a}=\pm 1,\sigma _{M+1}=\pm 1}\hspace{-4mm}
   \sigma_{M+1}\,\sigma_a\, \frac{\ths (\lambda _{M+1}^{( \sigma) }+\lambda _{a}^{(\sigma) }-\eta (\delta +1))}{\ths (\lambda _{M+1}^{( \sigma) }+\lambda _{a}^{( \sigma ) })}
   \prod_{\ell\not=i_-}\ths(\lambda _{M+1}^{( \sigma) }+\tfrac\eta 2+\epsilon_{\alpha^-_\ell}\alpha^-_\ell)
   \ths(\lambda _a^{( \sigma) }+\tfrac\eta 2+\epsilon_{\alpha^-_\ell} \alpha^-_\ell),
\end{multline*}
in which we have used \eqref{signs-}-\eqref{gamma-delta-}.
This expression is symmetric in $\lambda _{M+1}$ and $\lambda _a$. It is easy to see that it is an entire function of  $\lambda _{M+1}$ (the residues of the apparent poles at   $\lambda _{M+1}=\sigma\lambda _a$, $\sigma=\pm$, vanish). 
Moreover, since $\eta\delta=-\sum_{\ell\not=i_-}\epsilon_{\alpha_\ell^-}\alpha_\ell^-$, it is also an elliptic function of $\lambda _{M+1}$, 
which is therefore a constant. We moreover notice that it vanishes for $\lambda _{M+1}=0$, so that this constant is identically zero.

It means that the only non vanishing contribution comes from the direct action of  $D(-\lambda _{M+1}^{(\sigma) }|\gamma +\delta ,\alpha -\beta -1)$ in \eqref{inductionBb-M}, which leads to the result \eqref{Bb-state}-\eqref{H_sigma} for $M+1$.
\end{proof}

\subsection{Action on bulk Bethe states}

Let us now compute the action of the local operators \eqref{op-E_m}, i.e. of the elements of the basis \eqref{Local-Basis}, on the bulk Bethe states. As mentioned above, we follow the strategy proposed in \cite{KitKMNST07,NicT22}, using the solution of the bulk inverse problem \eqref{reconstr-1}-\eqref{reconstr-2} for the operators \eqref{op-E_m}, i.e. their expression in terms of the elements of the bulk monodromy matrix \eqref{def-Mgauge}.

It follows from \eqref{reconstr-1}-\eqref{reconstr-2} (see \cite{NicT22} for details, the arguments are the same as in the XXZ case) that the operators \eqref{op-E_m} can be expressed as
\begin{multline}\label{reconst-basis}
  \barE_{m}^{\boldsymbol{\epsilon'},\boldsymbol{\epsilon}}(\alpha,\beta)
    = \prod_{n=1}^m\frac{\det S(-\xi_n|c_n,d_n)}{\det S(-\xi_n| a_n, b_n)}
    \prod_{n=1\to m} \hspace{-2mm}
    M_{\epsilon_n,\epsilon'_n}(\xi_n-\eta /2|(c_n,d_n),(\bar{a}_n,\bar{b}_n)) 
    \\
    \times
    \prod_{n=m\to 1} 
    \frac{M_{3-\epsilon_n,3-\epsilon'_n}(\xi_n+\eta /2|(c_n-1,d_n),(a_n-1,b_n))}{\det_q M(\xi_n)}.
\end{multline}
Such an expression is a priori valid for any choice of the gauge parameters $c_n$ and $d_n$, $1\le n\le m$, provided that the matrices $S(-\xi_n|c_n,d_n)$ are invertible.
As in \cite{NicT22}, we choose these parameters to be on the form
\begin{align}
   &c_n=\frac{x+\alpha+\beta+N-1}2-M, \label{c_n}\\ 
   &d_n=\frac{-x+\alpha+\beta+N-1}2-M-\sum_{r=1}^n (-1)^{\epsilon_r}, \label{d_n}
\end{align}
for any arbitrary choice of $x$ so that the matrices  $S(-\xi_n|c_n,d_n)$ are invertible, and we compute the action of \eqref{reconst-basis} on generalised bulk Bethe states of the form
\begin{equation}\label{bulk-Bethe state}
  \underline{B}_M(\{\mu_i\}_{i=1}^M | x-1,\alpha-\beta)\, \ket{\eta,\alpha +\beta +N-M-1}.
\end{equation}

\begin{prop}\label{prop-act-bulk}
Let us consider, for given $m$-tuples $\boldsymbol{\epsilon}\equiv(\epsilon_1,\ldots,\epsilon_m), \boldsymbol{\epsilon'}\equiv(\epsilon'_1,\ldots,\epsilon'_m)\in \{1,2\}^{m}$ 
and gauge parameters $\alpha,\beta$, the operator $\barE_{m}^{\boldsymbol{\epsilon'},\boldsymbol{\epsilon}}(\alpha,\beta)$ defined as in \eqref{op-E_m}, with $a_n,b_n,\bar a_n,\bar b_n$ given in terms of $\alpha,\beta,\boldsymbol{\epsilon},\boldsymbol{\epsilon'}$ by \eqref{Gauge.Basis-1}-\eqref{Gauge.Basis-2}.
Then, for any gauge parameter $x$, its action on the generalised bulk states \eqref{bulk-Bethe state} takes the following form:
\begin{multline}\label{act-local-op-bulk}
  \barE_{m}^{\boldsymbol{\epsilon'},\boldsymbol{\epsilon}}(\alpha,\beta)\
   \underline{B}_M(\{\mu_i\}_{i=1}^M | x-1,\alpha-\beta)\, \ket{\eta,\alpha +\beta +N-M-1}
  \\
  =\sum_{\mathsf{B}_{\boldsymbol{\epsilon,\epsilon'}}} 
    \mathcal{F}_{\mathsf{B}_{\boldsymbol{\epsilon,\epsilon'}}}(\{\mu_j\}_{j=1}^M,\{\xi_j^{(1)}\}_{j=1}^m\mid\alpha,\beta,x)\
    \\
    \times
  \underline{B}_{M+\tilde m_{\boldsymbol{\epsilon,\epsilon'}}}(\{\mu_j\}_{j\in\mathsf{A}_{\boldsymbol{\epsilon,\epsilon'}}} | x-1,\alpha-\beta-2\tilde m_{\boldsymbol{\epsilon,\epsilon'}})\, 
  \ket{\eta,\alpha +\beta +N-M-1+\tilde m_{\boldsymbol{\epsilon,\epsilon'}}},
\end{multline}
where 
\begin{align}\label{Def-F_B}
     &\mathcal{F}_{\mathsf{B}_{\boldsymbol{\epsilon},\boldsymbol{\epsilon'}}}(\{\mu_j\}_{j=1}^{M},\{\xi_j^{(1)}\}_{j=1}^{m} \mid \alpha,\beta,x)
     =  f_{\boldsymbol{\epsilon},\boldsymbol{\epsilon'}}(\alpha,\beta,x)\, 
     \prod_{n=1}^m\frac{\ths(\xi_k+\eta\alpha)}{\ths(\eta b_n)\ths(\xi_k+\eta(\alpha+1))}\,    
     \nonumber\\
     &\qquad\times
     \frac{\prod\limits_{j=1}^{s+s'} \bigg[ d(\mu_{\text{\textsc{b}}_j}) \,   \ths( \mu_{\text{\textsc{b}}_j} +\eta(\tfrac{x+\alpha+\beta+N}2-M))
             \frac{\prod_{k=1}^M\ths(\mu_k-\mu_{\text{\textsc{b}}_j}-\eta)}
                    {\prod_{\substack{k=1 \\ k\neq \text{\textsc{b}}_j}}^M\ths(\mu_k-\mu_{\text{\textsc{b}}_j})}\bigg] }
           {\prod\limits_{j=1}^m \bigg[ d(\xi_j^{(1)}) \,   \ths( \xi_j^{(1)} +\eta(\tfrac{x+\alpha+\beta+N}2-M))\,
            \prod\limits_{k=1}^M\frac{\ths(\mu_k-\xi_j^{(1)}-\eta)}{\ths(\mu_k-\xi_j^{(1)})}
            \bigg]}  
           \,
     \prod_{1\leq i<j\leq s+s'}
     \frac{\ths (\mu _{\text{\textsc{b}}_i}-\mu_{\text{\textsc{b}}_j})}
            {\ths (\mu _{\text{\textsc{b}}_i}-\mu_{\text{\textsc{b}}_j}-\eta )}
     \nonumber\\
     &\qquad\times
     \prod_{p=1}^{s}\left[ 
     \ths(\xi_{i_p}^{(1)}-\mu_{\text{\textsc{b}}_p}+\eta(1+b_{i_p}))\,
     \frac{\prod_{k=i_p+1}^m \ths(\mu_{\text{\textsc{b}}_p}-\xi_k^{(1)}-\eta)}
            {\prod_{k=i_p}^m \ths(\mu_{\text{\textsc{b}}_p}-\xi_k^{(1)})} 
     \right]    
     \nonumber\\
     &\qquad\times
     \prod_{p=s+1}^{s+s'}\left[ 
     \ths(\xi_{i_p}^{(1)}-\mu_{\text{\textsc{b}}_p}-\eta(1-\bar b_{i_p}))\,
     \frac{\prod_{k=i_p+1}^m \ths(\xi_k^{(1)}-\mu_{\text{\textsc{b}}_p}-\eta)}
            {\prod_{\substack{k=i_p \\ k\not= M+m+1-{\text{\textsc{b}}_p}}}^m \ths(\xi_k^{(1)}-\mu_{\text{\textsc{b}}_p})} 
     \right]      ,     
\end{align}
with
\begin{align}\label{Def-f}
  f_{\boldsymbol{\epsilon},\boldsymbol{\epsilon'}}(\alpha,\beta,x)
   &=\prod_{k=1}^{M+\tilde m_{\boldsymbol{\epsilon,\epsilon'}}}
   \frac{\ths(\eta(-\beta-k+2M+1))}{\ths(\eta(\frac{-x+\alpha+\beta+N-2M-1}2+k))}
   \prod_{k=1}^{M}\frac{\ths(\eta(\frac{-x+\alpha+\beta+N-2M-1}2+k))}{\ths(\eta(-\beta-k+2M+1))}.
\end{align}
In these expressions, we have used the following notations:
\begin{alignat}{2}
   &\{ i_p\}_{p\in\{1,\ldots,s\}}=\{1,\ldots,m\}\cap \{j\mid\epsilon_j=2\}\qquad& &\text{with}\quad i_p<i_q\ \ \text{if}\ \ p<q,\label{i_p-s}\\
   &\{ i_p\}_{p\in\{s+1,\ldots,s+s'\}}=\{1,\ldots,m\}\cap \{j\mid\epsilon'_j=1\}\qquad& &\text{with}\quad i_p>i_q\ \ \text{if}\ \ p<q,\label{i_p-s'}
\end{alignat}
and 
\begin{equation}\label{m_epsilon}
   \tilde m_{\boldsymbol{\epsilon,\epsilon'}}=\sum_{r=1}^m(\epsilon'_r-\epsilon_r)=m-(s+s').
\end{equation}
We have also used the shortcut notation
\begin{equation}
      \mu_{M+j}=\xi_{m+1-j}^{(1)}, \qquad 1\le j \le m.
\end{equation}
With these notations, the sum in \eqref{act-local-op-bulk} runs over all possible sets of integers $\mathsf{B}_{\boldsymbol{\epsilon,\epsilon'}}=\{\text{\textsc{b}}_1,\ldots,\text{\textsc{b}}_{s+s'}\}$ such that
\begin{equation}\label{def-setB}
  \begin{cases}
  \text{\textsc{b}}_{p}\in \{1,\ldots ,M\}\setminus \{\text{\textsc{b}}_1,\ldots ,\text{\textsc{b}}_{p-1}\}\qquad 
  & \text{for}\quad 0<p\leq s, \\ 
  \text{\textsc{b}}_{p}\in \{1,\ldots ,M+m+1-i_{p}\}\setminus \{\text{\textsc{b}}_{1},\ldots ,\text{\textsc{b}}_{p-1}\}\qquad  
  & \text{for}\quad s<p\leq s+s',
\end{cases}
\end{equation}
and
\begin{equation}\label{def-setA}
  \mathsf{A}_{\boldsymbol{\epsilon,\epsilon'}}\equiv \{ \text{\textsc{a}}_a,\ldots,\text{\textsc{a}}_{M+ \tilde m_{\boldsymbol{\epsilon,\epsilon'}}}\}=\{1,\ldots,M+m\}\setminus  \mathsf{B}_{\boldsymbol{\epsilon,\epsilon'}}.
\end{equation}
\end{prop}

\begin{proof}
The computation of this action can be done similarly as for the XXZ case, see \cite{NicT22}. We use for that the reconstruction formula \eqref{reconst-basis}, with the particular choice \eqref{c_n} and \eqref{d_n} for the parameters $c_n$ and $d_n$ in terms of $x$.
More precisely, we proceed by induction, computing recursively, from $n=m$ to $n=1$, the action of the following monomials of elements of the bulk gauged monodromy matrix:
\begin{multline}
  \underline{E}^{(n,m)}_{\boldsymbol{\epsilon},\boldsymbol{\epsilon'}}(\alpha,\beta,x)
  =\prod_{k=n}^m\frac{\det S(-\xi_k|c_k,d_k)}{\det S(-\xi_k| a_k, b_k)}
  \prod_{k=n\to m} M_{\epsilon_k,\epsilon'_k}(\xi_k-\eta/2|(c_k,d_k),(\bar a_k,\bar b_k))\,
  \\
  \times
  \prod_{k=m\to n} \frac{M_{3-\epsilon_k,3-\epsilon_k}(\xi_k+\eta/2 |(c_k-1,d_k),( a_k-1,b_k))}{\det_q M(\xi_k)},
\end{multline}
on generalised bulk Bethe states of the form
\begin{equation}\label{interm-bBstate}
\underline{B}_M(\{\mu_i\}_{i=1}^M | x_n-1,z_n)\,\ket{\eta,y_n},
\end{equation}
with
\begin{align}
   &x_n=x+\sum_{r=1}^{n-1}( -1)^{\epsilon_r}=c_n-d_n-(-1)^{\epsilon_n},\\
   &z_n=\alpha -\beta+\sum_{r=1}^{n-1}( -1)^{\epsilon_r}=a_n-b_n-(-1)^{\epsilon_n}-1,\\
   &y_n=\alpha+\beta+N-M-1 -\sum_{r=1}^{n-1}( -1) ^{\epsilon _r} =a_n+b_n+(-1)^{\epsilon_n}+N-M-2
   \nonumber\\
   &\hphantom{y_n=\alpha+\beta+N-M-1 -\sum_{r=1}^{n-1}( -1) ^{\epsilon _r}}
 =c_n+d_n+M+(-1)^{\epsilon_n}.
\end{align}
The intermediate recursion formula is given by
\begin{multline}\label{action_Mn}
   \underline{E}^{(n,m)}_{\boldsymbol{\epsilon},\boldsymbol{\epsilon'}}(\alpha,\beta,x)\
   \underline{B}_M(\{\mu_i\}_{i=1}^M | x_n-1,z_n)\,\ket{\eta,y_n}
   = 
   \sum_{\mathsf{B}_{\boldsymbol{\epsilon},\boldsymbol{\epsilon'}}^{(n)}}
   \mathcal{F}_{\mathsf{B}_{\boldsymbol{\epsilon},\boldsymbol{\epsilon'}}^{(n)}}(\{\mu_j\}_{j=1}^{M},\{\xi_j^{(1)}\}_{j=n}^{m} \mid  \alpha,\beta,x)\\
   \times
   \underline{B}_{M+\tilde m_n}   (\{\mu_j\}_{j\in\mathsf{A}_{\boldsymbol{\epsilon},\boldsymbol{\epsilon'}}^{(n)}} | x_n-1,z_n-2\tilde m_n)\, \ket{\eta,y_n+\tilde{m}_n},
\end{multline}
where $\tilde m_n$ is given by \eqref{tilde-m_n}, and
\begin{align}\label{Def-F-rec}
     &\mathcal{F}_{\mathsf{B}_{\boldsymbol{\epsilon},\boldsymbol{\epsilon'}}^{(n)}}(\{\mu_j\}_{j=1}^{M},\{\xi_j^{(1)}\}_{j=n}^{m} \mid \alpha,\beta,x)
     =f^{(n,m)}_{\boldsymbol{\epsilon},\boldsymbol{\epsilon'}}(\alpha,\beta,x)\, 
      \prod_{k=n}^m 
     \frac{\ths(\xi_k+\alpha\eta) }{\ths(\eta b_k)\ths(\xi_k+\eta(\alpha+1))}  
     \nonumber\\
     &\qquad\times
     \frac{\prod\limits_{j=1}^{s_{(n)}+s'_{(n)}} \bigg[ d(\mu_{\text{\textsc{b}}_j^{(n)}}) \,    \ths( \mu_{\text{\textsc{b}}_j^{(n)}}+(c_n+\tfrac 12)\eta)\,
             \frac{\prod_{k=1}^M\ths(\mu_k-\mu_{\text{\textsc{b}}_j^{(n)}}-\eta)}
                    {\prod_{\substack{k=1 \\ k\neq \text{\textsc{b}}_j^{(n)}}}^M\ths(\mu_k-\mu_{\text{\textsc{b}}_j^{(n)}})}\bigg] }
           {\prod\limits_{j=n}^m \bigg[ d(\xi_j^{(1)}) \,\,    \ths( \xi_j^{(1)}+(c_n+\tfrac 12)\eta)
            \prod\limits_{k=1}^M\frac{\ths(\mu_k-\xi_j^{(1)}-\eta)}{\ths(\mu_k-\xi_j^{(1)})}
            \bigg]}  
           \,
     \nonumber\\
     &\qquad\times
     \prod_{1\leq i<j\leq s_{(n)}+s_{(n)}^{\prime }}
     \frac{\ths (\mu _{\text{\textsc{b}}_i^{(n)}}-\mu_{\text{\textsc{b}}_j^{(n)}})}
            {\ths (\mu _{\text{\textsc{b}}_i^{(n)}}-\mu_{\text{\textsc{b}}_j^{(n) }}-\eta )}
     \nonumber\\
     &\qquad\times
     \prod_{p=1}^{s_{(n)}}\left[ 
     \ths(\xi_{i_p^{(n)}}^{(1)}-\mu_{\text{\textsc{b}}_p^{(n)}}+\eta(1+b_{i_p^{(n)}}))\,
     \frac{\prod_{k=i_p^{(n)}+1}^m \ths(\mu_{\text{\textsc{b}}_p^{(n)}}-\xi_k^{(1)}-\eta)}
            {\prod_{k=i_p^{(n)}}^m \ths(\mu_{\text{\textsc{b}}_p^{(n)}}-\xi_k^{(1)})} 
     \right]    
     \nonumber\\
     &\qquad\times
     \prod_{p=s_{(n)}+1}^{s_{(n)}+s'_{(n)}}\left[ 
     \ths(\xi_{i_p^{(n)}}^{(1)}-\mu_{\text{\textsc{b}}_p^{(n)}}-\eta(1-\bar b_{i_p^{(n)}}))\,
     \frac{\prod_{k=i_p^{(n)}+1}^m \ths(\xi_k^{(1)}-\mu_{\text{\textsc{b}}_p^{(n)}}-\eta)}
            {\prod_{\substack{k=i_p^{(n)} \\ k\not= M+m+1-{\text{\textsc{b}}_p^{(n)}}}}^m \ths(\xi_k^{(1)}-\mu_{\text{\textsc{b}}_p^{(n)}})} 
     \right]      ,     
\end{align}
with
\begin{equation}\label{Def-f-rec}
  f^{(n,m)}_{\boldsymbol{\epsilon},\boldsymbol{\epsilon'}}(\alpha,\beta,x)
   = \begin{cases}
     \displaystyle \prod\limits_{k=1}^{\tilde m_n}\left[ -\frac{\ths(\eta(b_{n-1}+k-M-1))}{\ths(\eta(d_{n-1}+k+M))} \right]
         &\text{if}\quad \tilde m_n>0,\\
     1  &\text{if}\quad \tilde m_n=0,\\
    \displaystyle \prod\limits_{k=1}^{|\tilde m_n|} \left[ -\frac{\ths(\eta(d_{n-1}-k+M+1))}{\ths(\eta(b_{n-1}-k-M))}\right]
         &\text{if}\quad \tilde m_n<0.
   \end{cases}
\end{equation}
Here we have defined the intermediate sets of indices
\begin{alignat}{2}
   &\{i_{p}^{(n) }\}_{p\in \{1,\ldots ,s_{(n)}\}},\quad  &
   &\text{with}\quad   i_{k}^{(n)}<i_{h}^{(n)}\quad
   \text{for}\quad   0<k<h\leq s_{(n)},  
   \label{i_p-Def0} \\
  &\{i_{p}^{(n)}\}_{p\in \{s_{(n)}+1,\ldots ,s_{(n)}+s'_{(n)}\}},\quad &
  &\text{with}\quad  i_{k}^{(n)}>i_{h}^{(n)}\quad
  \text{for}\quad  s_{(n)}<k<h\leq s_{(n)}+s_{(n)}^{\prime },  
  \label{i_p-Def1}
\end{alignat}
so that
\begin{alignat}{2}
  &j\in \{i_{p}^{(n)}\}_{p\in \{1,\ldots ,s_{(n)}\}} \qquad & 
  &\text{iff}\qquad n\le j\le m\quad\text{and}\quad\epsilon _{j}=2,
  \\
  &j\in \{i_{p}^{(n)}\}_{p\in \{s_{(n) }+1,\ldots ,s_{(n) }+s_{(n) }^{\prime }\}} \qquad & 
  &\text{iff}\qquad n\le j\le m\quad\text{and}\quad \epsilon'_{j}=1. 
\end{alignat}
and the intermediate sets of integers 
$\mathsf{B}^{( n) }_{\boldsymbol{\epsilon},\boldsymbol{\epsilon'}}
=\{\text{\textsc{b}}_{1}^{( n) },\ldots ,\text{\textsc{b}}_{s_{(n)}+s'_{(n)}}^{( n) }\}$ such that
\begin{equation}
\begin{cases}
   \text{\textsc{b}}_{p}^{( n) }\in 
   \{1,\ldots ,M\}\setminus \{\text{\textsc{b}}_{1}^{( n) },\ldots ,\text{\textsc{b}}_{p-1}^{(n) }\}
   \qquad & \text{for}\quad 0<p\leq s_{( n) }, 
   \\ 
   \text{\textsc{b}}_{p}^{( n) }\in \{1,\ldots ,M+m+1-i_{p}^{(n) }\}\setminus \{\text{\textsc{b}}_{1}^{( n) },\ldots ,\text{\textsc{b}}_{p-1}^{( n) }\}
   \quad  & \text{for}\quad s<p\leq s_{( n) }+s_{( n) }^{\prime },
\end{cases}
\label{Def-Bss}
\end{equation}
and
\begin{equation}\label{Def-Ass}
   \mathsf{A}^{(n)}_{\boldsymbol{\epsilon},\boldsymbol{\epsilon'}}
   \equiv \big\{\text{\textsc{a}}_{1}^{( n) },\ldots,\text{\textsc{a}}_{M+\tilde m_n}^{( n)}\big\}
   =\{1,\ldots ,M+m+1-n\}\setminus \mathsf{B}^{( n) }_{\boldsymbol{\epsilon},\boldsymbol{\epsilon'}}.
\end{equation}
The proof of \eqref{action_Mn} is completely similar to the proof of Proposition~4.4 of \cite{NicT22}. The induction step from $n+1$ to $n$ can be done by distinguishing the four possible choices of $(\epsilon_n,\epsilon'_n)\in\{1,2\}^2$, and by using the form of the action \eqref{actA-Bethe} and \eqref{actD-Bethe} of the bulk gauge operators on the bulk Bethe states, as well as the identity \eqref{Alternative-q-det} and the cancellation properties of the product of operators \eqref{cancel-prod}.
The only noticeable difference with respect to the trigonometric case comes from the form of the coefficients \eqref{coeff-act-A} and \eqref{coeff-act-D}. We refer the reader to \cite{NicT22} for more details.
\end{proof}

\subsection{Action on boundary Bethe states}

From the action of the operator $\barE_{m}^{\boldsymbol{\epsilon'},\boldsymbol{\epsilon}}(\alpha,\beta)$ \eqref{op-E_m} on generalised bulk Bethe states that we have just computed, we can now deduce its action on generalised boundary Bethe states
by using the boundary-bulk decomposition for $\gamma$ and $\delta$ satisfying \eqref{cond-c-=0}, \eqref{cond-b-=0}.
These parameters $\gamma$ and $\delta$ have also to be chosen such that
\begin{equation}\label{relation-gamma+delta-alpha+beta}
    \gamma+\delta=\alpha+\beta+N-1-2M.
\end{equation}
This is possible as soon as $\alpha+\beta$ satisfies the following constraint:
\begin{equation}\label{cond-alpha+beta}
 (\alpha+\beta+N-1-2M)\eta=\sum_{\ell=1}^3\epsilon_{\alpha^-_\ell} \alpha^-_\ell
 \!\!\!\mod\!(2\pi,2\pi\omega),
\end{equation}
with $\epsilon_{\alpha^-_\ell}=\pm 1$ such that $ \epsilon_{\alpha^-_1}\epsilon_{\alpha^-_2}\epsilon_{\alpha^-_3}=1$
(which coincides with \eqref{cond-ref-states}).
Then, with the notation $ \lambda_{M+j}=\xi_{m+1-j}^{(1)}$ and the notations of Proposition~\ref{prop-act-bulk},
\begin{align}
  &\barE_{m}^{\boldsymbol{\epsilon'},\boldsymbol{\epsilon}}(\alpha,\beta)\
   \underline{\widehat{\mathcal B}}_{-,M}(\{\lambda_i\}_{i=1}^M|\alpha-\beta+1)\, 
   \ket{\eta,\alpha+\beta+N-1-2M}
   \nonumber\\
   &\quad
   = \sum_{\boldsymbol{\sigma}\in\{\pm 1\}^M} 
   \prod_{n=1}^M\frac{\ths\!\big((-\beta+2M-n+1) \eta\big)\ths\!\big(\lambda _n^{(\sigma)}+\tfrac\eta 2-\alpha \eta\big)} {\ths(\eta (\delta +n)) \, \ths\!\big(\lambda_n^{( \sigma) }+\frac\eta 2+\eta \gamma\big)}\,  
   H^-_{\boldsymbol{\sigma}}(\{\lambda _{i}\}_{i=1}^{M}|{\boldsymbol{\epsilon_{\alpha^-}}})
    \nonumber\\
    &\quad\times 
    \sum_{\mathsf{B}_{\boldsymbol{\epsilon,\epsilon'}}} 
    \mathcal{F}_{\mathsf{B}_{\boldsymbol{\epsilon,\epsilon'}}}(\{\lambda_i^{( \sigma ) }\}_{i=1}^M,\{\xi_j^{(1)}\}_{j=1}^m\mid\alpha,\beta,\gamma-\delta)\
    \nonumber\\
  &\quad    \times
  \underline{B}_{M+\tilde m_{\boldsymbol{\epsilon,\epsilon'}}}(\{\lambda_i^{( \sigma ) }\}_{i\in\mathsf{A}_{\boldsymbol{\epsilon,\epsilon'}}} | \gamma-\delta-1,\alpha-\beta-2\tilde m_{\boldsymbol{\epsilon,\epsilon'}})\, 
  \ket{\eta,\alpha +\beta +N-M-1+\tilde m_{\boldsymbol{\epsilon,\epsilon'}}},
\end{align}
in which we have used \eqref{relation-gamma+delta-alpha+beta} to simplify the coefficient of the boundary-bulk decomposition.

Recall that we have, from the definition of the sets \eqref{def-setB} and \eqref{def-setA},
\begin{equation}
   \{\lambda_i \}_{i\in\mathsf{A}_{\boldsymbol{\epsilon,\epsilon'}}}
   =\{\lambda_i \}_{i=1}^{M} \cup \{\xi_j^{(1)}\}_{j=1}^m\setminus\{\lambda_i\}_{i\in\mathsf{B}_{\boldsymbol{\epsilon,\epsilon'}}},
\end{equation}
Hence, defining two partitions $\mathsf{\Lambda}_+\cup\mathsf{\Lambda}_-$ of $\{1,\ldots,M\}$ and $\mathsf{\Gamma}_+\cup\mathsf{\Gamma}_-$ of $\{1,\ldots,m\}$ such that
\begin{alignat}{2}
& \mathsf{\Lambda}_+=\mathsf{B}_{\boldsymbol{\epsilon,\epsilon'}}\cap \{1,\ldots,M\}, & 
& \mathsf{\Lambda}_-=\mathsf{A}_{\boldsymbol{\epsilon,\epsilon'}}\cap\{1,\ldots ,M\}, 
\label{def-part-lambda}\\
& \mathsf{\Gamma}_-=\{M+m+1-j\}_{j\in \mathsf{B}_{\boldsymbol{\epsilon,\epsilon'}}\cap \{N+1,\ldots ,N+m\}},\quad & 
& \mathsf{\Gamma}_+=\{1,\ldots ,m\}\setminus\mathsf{\Gamma}_-,
\label{def-part-xi}
\end{alignat}
i.e.
\begin{align}
   &\{\lambda_i \}_{i\in\mathsf{A}_{\boldsymbol{\epsilon,\epsilon'}}}=\{\lambda_i \}_{i\in \mathsf{\Lambda}_-}\cup\{\xi_j^{(1)}\}_{j\in\mathsf{\Gamma}_+},\\
   &\{\lambda_i \}_{i\in\mathsf{B}_{\boldsymbol{\epsilon,\epsilon'}}}=\{\lambda_i \}_{i\in\mathsf{\Lambda}_+}\cup\{\xi_j^{(1)}\}_{j\in\mathsf{\Gamma}_-},
\end{align}
we have
\begin{align}\label{ratio-H}
 &H_{\boldsymbol{\sigma}}(\{\lambda _{i}\}_{i=1}^{M}|\alpha-\beta,\gamma,\delta) = H_{\boldsymbol{\sigma}}(\{\lambda _{i}\}_{i\in\mathsf{A}_{\boldsymbol{\epsilon,\epsilon'}}}|\alpha-\beta-2\tilde m_{\boldsymbol{\epsilon,\epsilon'}},\gamma,\delta)
 \nonumber\\
 &\qquad\quad\times
\frac{H^-_{\boldsymbol{\sigma}}(\{\lambda_i\}_{i\in\mathsf{\Lambda}_+}|{\boldsymbol{\epsilon_{\alpha^-}}}) }
  {H^-_{\boldsymbol{1}}(\{\xi_i^{(1)}\}_{i\in\mathsf{\Gamma}_+}|{\boldsymbol{\epsilon_{\alpha^-}}}) }\,
   \prod_{j\in\mathsf{\Lambda}_-}\Bigg\{\prod_{i\in\mathsf{\Lambda}_+}\frac{\ths(\lambda_i^{(\sigma)}+\lambda_j^{(\sigma)}+\eta)}{\ths(\lambda_i^{(\sigma)}+\lambda_j^{(\sigma)})}
  \prod_{i\in\mathsf{\Gamma}_+} \frac{\ths(\lambda_i^{(\sigma)}+\xi_j^{(1)})}{\ths(\lambda_i^{(\sigma)}+\xi_j^{(1)}+\eta)}\Bigg\}
  \nonumber\\
 &\qquad\quad \times
  \left[f_{\boldsymbol{\epsilon,\epsilon'}}(\alpha,\beta,\gamma-\delta)\right]^{-1}
  \frac{\prod_{j=1}^m\ths\big(\xi_j^{(1) }+\frac\eta 2+\eta \gamma\big)}{\prod_{n\in\mathsf{B}_{\boldsymbol{\epsilon,\epsilon'}}} \ths\big(\lambda _{n}^{( \sigma) }+\frac\eta 2+\eta \gamma\big)}
  \frac{\prod_{n\in\mathsf{B}_{\boldsymbol{\epsilon,\epsilon'}}} \ths\big(\lambda_n^{( \sigma) }+\frac\eta 2-\eta \alpha\big)}{\prod_{j=1}^m\ths\big(\xi_j^{(1) }+\frac\eta 2-\eta \alpha\big)}.
\end{align}
Here we have again used \eqref{relation-gamma+delta-alpha+beta}.
Finally, noticing that
\begin{align}
 \frac{\gamma-\delta+\alpha+\beta+N-1}2-M=\gamma,
\end{align}
we can formulate the following result, which is the analog of Theorem~2.2 of \cite{NicT23}:

\begin{theorem}\label{th-act-boundary}
For $\alpha$ and $\beta$ satisfying \eqref{cond-alpha+beta}, the action of the operator \eqref{op-E_m} on a generalised boundary Bethe state is given as
\begin{multline}\label{act-boundary}
   \barE_{m}^{\boldsymbol{\epsilon'},\boldsymbol{\epsilon}}(\alpha,\beta)\
   \underline{\widehat{\mathcal B}}_{-,M}(\{\lambda_i\}_{i=1}^M|\alpha-\beta+1)\, 
   \ket{\eta,\alpha+\beta+N-1-2M}
   \\
   = \sum_{\mathsf{B}_{\boldsymbol{\epsilon,\epsilon'}}} 
    \mathcal{\widehat F}_{\mathsf{B}_{\boldsymbol{\epsilon,\epsilon'}}}(\{\lambda_j\}_{j=1}^M,\{\xi_j^{(1)}\}_{j=1}^m | \beta)\
    \\
    \times
    \underline{\widehat{\mathcal B}}_{-,M+\tilde m_{\boldsymbol{\epsilon,\epsilon'}}}(\{\lambda_i\}_{\substack{i=1\\ i\notin\mathsf{B}_{\boldsymbol{\epsilon,\epsilon'}}}}^{M+m}|\alpha-\beta+1-2\tilde m_{\boldsymbol{\epsilon,\epsilon'}})\, 
   \ket{\eta,\alpha+\beta+N-1-2M}.
\end{multline}
Here, we have defined $\lambda_{M+j}:=\xi _{m+1-j}^{( 1) }$ for $j\in \{1,\ldots ,m\}$, and used the notations \eqref{i_p-s}-\eqref{i_p-s'} and \eqref{m_epsilon}. The sum in \eqref{act-boundary} runs over all possible sets of integers $\mathsf{B}_{\boldsymbol{\epsilon,\epsilon'}}=\{\text{\textsc{b}}_{1},\ldots ,\text{\textsc{b}}_{s+s^{\prime }}\}$ defined as in \eqref{def-setB}. 
The coefficient in \eqref{act-boundary} is
\begin{align}\label{coeff-act-boundary}
& \mathcal{\widehat{F}}_{\mathsf{B}_{\boldsymbol{\epsilon,\epsilon'}}}(\{\lambda_j\}_{j=1}^{M},\{\xi_j^{(1)}\}_{j=1}^m |\beta )
   =     \prod_{n=1}^m\frac{\ths(\xi_n+\eta\alpha)}{\ths(\eta b_n)\ths(\xi_n+\eta(\alpha+1))}\,    
      \sum_{\substack{\sigma_j=\pm \\ j\in\mathsf{\Lambda}_+}}
      \frac{\prod_{j=1}^{s+s^{\prime }}d(\lambda_{\text{\textsc{b}}_j}^{\sigma })\, \ths\!\big(\lambda_{\text{\textsc{b}}_j}^{\sigma }+\tfrac\eta 2-\alpha\eta\big)}
      {\prod_{j=1}^m d(\xi_j^{(1)})\, \ths\!\big(\xi_j^{(1)}+\tfrac\eta 2-\alpha\eta\big)}
      \nonumber \\
   & \qquad \times 
   \frac{H^-_{\boldsymbol{\sigma}}(\{\lambda_i\}_{i\in\mathsf{\Lambda}_+}|{\boldsymbol{\epsilon_{\alpha^-}}}) }
  {H^-_{\boldsymbol{1}}(\{\xi_i^{(1)}\}_{i\in\mathsf{\Gamma}_+}|{\boldsymbol{\epsilon_{\alpha^-}}}) }\,
  \prod_{i\in \mathsf{\Lambda}_- }\prod_{\epsilon =\pm }\left\{ \prod_{j\in \mathsf{\Lambda}_+}
  \frac{\ths (\lambda_j^{\sigma }+\epsilon \lambda_i+\eta )}{\ths (\lambda_j^{\sigma }+\epsilon\lambda_i)}
  \prod_{j\in \mathsf{\Gamma}_+}\frac{\ths (\xi_j^{(1)}+\epsilon \lambda_i)}{\ths(\xi_j^{(0)}+\epsilon \lambda_i)}\right\}  
      \nonumber \\
   & \qquad \times 
   \prod_{i\in \mathsf{\Lambda}_+}\left\{ \prod_{j\in \mathsf{\Gamma}_+}
   \frac{\ths (\xi_j^{(1)}-\lambda_i^{\sigma })}{\ths(\xi_j^{(0)}-\lambda_i^{\sigma })}\ 
   \frac{\prod_{j\in \mathsf{\Lambda}_+} \ths (\lambda_j^{\sigma}-\lambda_i^{\sigma }-\eta )}
          {\prod_{j\in \mathsf{\Lambda}_+\setminus \{i\}} \ths(\lambda_j^{\sigma }-\lambda_i^{\sigma })}\right\} 
   \prod_{1\leq i<j\leq s+s^{\prime }}
   \frac{\ths (\lambda_{\text{\textsc{b}}_i}^{\sigma }-\lambda_{\text{\textsc{b}}_j}^{\sigma })}
          {\ths (\lambda_{\text{\textsc{b}}_i}^{\sigma }-\lambda_{\text{\textsc{b}}_j}^{\sigma }-\eta )}  
       \nonumber \\
            &\qquad\times
     \prod_{p=1}^{s}\left[ 
     \ths(\xi_{i_p}^{(1)}-\lambda_{\text{\textsc{b}}_p}^{\sigma }+\eta(1+b_{i_p}))\,
     \frac{\prod_{k=i_p+1}^m \ths(\lambda_{\text{\textsc{b}}_p}^{\sigma }-\xi_k^{(1)}-\eta)}
            {\prod_{k=i_p}^m \ths(\lambda_{\text{\textsc{b}}_p}^{\sigma }-\xi_k^{(1)})} 
     \right]    
     \nonumber\\
     &\qquad\times
     \prod_{p=s+1}^{s+s'}\left[ 
     \ths(\xi_{i_p}^{(1)}-\lambda_{\text{\textsc{b}}_p}^{\sigma }-\eta(1-\bar b_{i_p}))\,
     \frac{\prod_{k=i_p+1}^m \ths(\xi_k^{(1)}-\lambda_{\text{\textsc{b}}_p}^{\sigma }-\eta)}
            {\prod_{\substack{k=i_p \\ k\not= M+m+1-{\text{\textsc{b}}_p}}}^m \ths(\xi_k^{(1)}-\lambda_{\text{\textsc{b}}_p}^{\sigma })} 
     \right]    ,
\end{align}
where we have used the notations \eqref{def-part-lambda}-\eqref{def-part-xi}. The sum in \eqref{coeff-act-boundary} runs over all $\sigma_j\in\{+,-\}$ for $j\in\mathsf{\Lambda}_+$, and we have used the shortcut notation $\lambda_{i}^{\sigma }=\sigma_{i}\lambda_{i}$ for $i\in \mathsf{B}_{\boldsymbol{\epsilon,\epsilon'}}$, with the convention $\sigma _{i}=1$ if $i>M$.
Finally, we recall that the parameters $b_n$ and $\bar b_n$ are given by \eqref{Gauge.Basis-1}-\eqref{Gauge.Basis-2}, and that the boundary-bulk coefficient $H^-_{\boldsymbol{\sigma}}(\{\lambda\}|{\boldsymbol{\epsilon_{\alpha^-}}}) $ is given by \eqref{H_sigma-bis}.
\end{theorem}

\section{On correlation functions}
\label{sec-corr}

As already mentioned, we suppose that the constraint \eqref{const-TQ} is satisfied for some integer $M$ and some choice of signs \eqref{choice-eps+-}, with $\epsilon_{\alpha_1^+}\epsilon_{\alpha_2^+}\epsilon_{\alpha_3^+}=\epsilon_{\alpha_1^-}\epsilon_{\alpha_2^-}\epsilon_{\alpha_3^-}=1$, and that the ground state of the model is among the states which can be described in terms of a solution $Q$ of the form \eqref{Q-form} of the corresponding functional $TQ$-equation \eqref{hom-TQ}. 

Let us denote by $\ket{Q}$, respectively by $\bra{Q}$,  the ground state written as a separate state of the form \eqref{eigen-r-Skl}, respectively of the form \eqref{eigen-l-Skl}, in the Vertex-IRF SoV basis in which $\alpha$ and $\beta$ satisfy \eqref{cond-c+=0}-\eqref{cond-b+=0} and \eqref{cond-ref-states}, and in which $\mathsf{g}_-$ is fixed as in \eqref{fix-g-}.
Our ultimate aim would be to compute the zero-temperature correlation function of some product of local operators $O_{1\to m}$ acting on the first $m$ sites of the chain, i.e. the mean value, in the ground state $\ket{Q}$, of the quasi-local operator $O_{1\to m}\in\End(\otimes_{n=1}^m\mathcal{H}_n)$:
\begin{equation}\label{moy-O}
   \moy{O_{1\to m}}=\frac{\bra{Q}\, O_{1\to m}\, \ket{Q}}{\moy{Q\,|\,Q}}.
\end{equation}
In general, this is a very complicated task, due to the combinatorial complexity induced by the use of the Vertex-IRF transformation for the explicit construction of the state $\ket{Q}$\footnote{and also due to the fact that we do not have a simple solution of the quantum inverse problem directly in terms of the boundary Yang-Baxter algebra.}. However, as in our previous works concerning the XXZ open spin chain \cite{NicT22,NicT23}, we will be able to compute some elementary building blocks for these correlation functions, i.e. matrix elements of the form
\begin{equation}\label{moy-Em}
   \moy{\barE_{m}^{\boldsymbol{\epsilon'},\boldsymbol{\epsilon}}(\alpha,\beta)}=\frac{\bra{Q}\, \barE_{m}^{\boldsymbol{\epsilon'},\boldsymbol{\epsilon}}(\alpha,\beta)\, \ket{Q}}{\moy{Q\,|\,Q}},
\end{equation}
in which $\barE_{m}^{\boldsymbol{\epsilon'},\boldsymbol{\epsilon}}(\alpha,\beta)$ stands for the tensor product of gauge transformed local operators as defined in \eqref{op-E_m}. We recall that the latter form a basis of $\End(\otimes_{n=1}^m\mathcal{H}_n)$ for $\boldsymbol{\epsilon},\boldsymbol{\epsilon'}\in \{1,2\}^{m}$, so that any quasi-local operator $O_{1\to m}\in\End(\otimes_{n=1}^m\mathcal{H}_n)$ can be expressed as a linear combination of such operators $\barE_{m}^{\boldsymbol{\epsilon'},\boldsymbol{\epsilon}}(\alpha,\beta)$.

\subsection{Finite size elementary blocks}

The strategy to compute \eqref{moy-Em} is completely similar to what we have done in \cite{NicT22,NicT23} in the XXZ case. We briefly recall it here.
%

We first express the right-separate state $\ket{Q}$ as a gauge boundary Bethe state, as in \eqref{separate-Bethe}:
\begin{equation}\label{block-1}
   \moy{\barE_{m}^{\boldsymbol{\epsilon'},\boldsymbol{\epsilon}}(\alpha,\beta)}=\mathsf{c}_{Q,\alpha,\beta}^{(R)} \, 
   \frac{\bra{Q}\, \barE_{m}^{\boldsymbol{\epsilon'},\boldsymbol{\epsilon}}(\alpha,\beta)\, 
   \underline{\widehat{\mathcal B}}_{-,M}(\{\lambda_j\}_{j=1}^M|\alpha-\beta+1)\, \ket{\eta, \alpha+\beta+N-2M-1}}{\moy{Q\,|\,Q}},
\end{equation}
with coefficient $\mathsf{c}_{Q,\alpha,\beta}^{(R)}$ given by \eqref{c-SoV-ABA-r}.
This enables us to use the result of the action of $\barE_{m}^{\boldsymbol{\epsilon'},\boldsymbol{\epsilon}}(\alpha,\beta)$ on the gauge boundary Bethe states 
obtained in Theorem~\ref{th-act-boundary}. We recall that the latter was computed by decomposing the gauge boundary Bethe states into bulk ones, by computing the action of $\barE_{m}^{\boldsymbol{\epsilon'},\boldsymbol{\epsilon}}(\alpha,\beta)$ on bulk gauge Bethe states, and by reconstructing the result in terms of gauge boundary Bethe states. This was possible in a not so intricate way (similar to what was done in the diagonal case \cite{KitKMNST07}) thanks to the adequate choice of the gauge parameters in \eqref{Gauge.Basis-1}-\eqref{Gauge.Basis-2}.
It gives
\begin{multline}\label{block-2}
   \moy{\barE_{m}^{\boldsymbol{\epsilon'},\boldsymbol{\epsilon}}(\alpha,\beta)}=\mathsf{c}_{Q,\alpha,\beta}^{(R)} \, 
   \sum_{\mathsf{B}_{\boldsymbol{\epsilon,\epsilon'}}} 
    \mathcal{\widehat F}_{\mathsf{B}_{\boldsymbol{\epsilon,\epsilon'}}}(\{ \lambda_j\}_{j=1}^M,\{\xi_j^{(1)}\}_{j=1}^m | \beta)\
    \\
    \times
   \frac{\bra{Q}\,  \underline{\widehat{\mathcal B}}_{-,M+\tilde m_{\boldsymbol{\epsilon,\epsilon'}}}(\{ \lambda_i\}_{\substack{i=1\\ i\notin\mathsf{B}_{\boldsymbol{\epsilon,\epsilon'}}}}^{M+m}|\alpha-\beta+1-2\tilde m_{\boldsymbol{\epsilon,\epsilon'}})\, \ket{\eta, \alpha+\beta+N-2M-1}}{\moy{Q\,|\,Q}},
\end{multline}
with coefficient given by \eqref{coeff-act-boundary}.

The next step is to reconstruct the resulting gauge boundary Bethe states into separate states, using \eqref{equality-ref-state} and \eqref{separate-Bethe-r}:
\begin{align}\label{separate-back}
   &\underline{\widehat{\mathcal B}}_{-,M+\tilde m_{\boldsymbol{\epsilon,\epsilon'}}}(\{ \lambda_i\}_{\substack{i=1\\ i\notin\mathsf{B}_{\boldsymbol{\epsilon,\epsilon'}}}}^{M+m}|\alpha-\beta+1-2\tilde m_{\boldsymbol{\epsilon,\epsilon'}})\, \ket{\eta, \alpha+\beta+N-2M-1}
   \nonumber\\
   &\quad
   =\underline{\widehat{\mathcal B}}_{-,M+\tilde m_{\boldsymbol{\epsilon,\epsilon'}}}(\{ \lambda_i\}_{\substack{i=1\\ i\notin\mathsf{B}_{\boldsymbol{\epsilon,\epsilon'}}}}^{M+m}|\alpha-\beta+1-2\tilde m_{\boldsymbol{\epsilon,\epsilon'}})\,
    \ket{\Omega_{\alpha,\beta-2M+1}}
    \nonumber\\
    &\quad
    =\frac{1}{\mathsf{c}_{\bar Q_{\mathsf{B}_{\boldsymbol{\epsilon,\epsilon'}}},\alpha,\beta+2\tilde m_{\boldsymbol{\epsilon,\epsilon'}}}^{(R)}}
     \sum_{\mathbf{h}\in\{0,1\}^N}\prod_{n=1}^N \bar Q_{\mathsf{B}_{\boldsymbol{\epsilon,\epsilon'}}}(\xi_n^{(h_n)})\ V(\xi_1^{(h_1)},\ldots,\xi_N^{(h_N)})\ \ket{\mathbf{h},\alpha,\beta+1+2\tilde m_{\boldsymbol{\epsilon,\epsilon'}}},
\end{align}
in which $\bar Q_{\mathsf{B}_{\boldsymbol{\epsilon,\epsilon'}}}$ is the function
\begin{equation}\label{Qbar}
   \bar Q_{\mathsf{B}_{\boldsymbol{\epsilon,\epsilon'}}}(\lambda)
   =\prod_{j=1}^{M+2\tilde m_{\boldsymbol{\epsilon,\epsilon'}}}
   \ths(\lambda-\bar \lambda_j)\ths(\lambda+\bar \lambda_j),
\end{equation}
with
\begin{equation}\label{lambda-bar}
   \{\bar \lambda_1,\ldots,\bar \lambda_{M+2\tilde m_{\boldsymbol{\epsilon,\epsilon'}}}\}
   =\{\lambda_1,\ldots, \lambda_{M+m}\}\setminus\{\lambda_j\}_{j\in\mathsf{B}_{\boldsymbol{\epsilon,\epsilon'}}}.
\end{equation}
Note however that, when $\tilde m_{\boldsymbol{\epsilon,\epsilon'}}\not= 0$, the separate state \eqref{separate-back} is expressed in a SoV basis in which the parameter $\beta$ is shifted with respect to the original one. This is a problem that we already encountered in our previous works  \cite{NicT22,NicT23}  concerning the XXZ case: at the moment, we do not know how to compute in a compact way the scalar product of the resulting separate state \eqref{separate-back} with the original one $\bra{Q}$ in cases in which $\tilde m_{\boldsymbol{\epsilon,\epsilon'}}\not= 0$. This means that, as in \cite{NicT22,NicT23}, we have to restrict our study to the elementary blocks \eqref{moy-Em} for which $\tilde m_{\boldsymbol{\epsilon,\epsilon'}}= 0$, i.e. for which
\begin{equation}\label{cond-block}
   \sum_{r=1}^m(\epsilon'_r-\epsilon_r)=0.
\end{equation}
Then, under the condition \eqref{cond-block}, the expression \eqref{block-2} for the elementary block becomes
\begin{align}\label{block-3}
   \moy{\barE_{m}^{\boldsymbol{\epsilon'},\boldsymbol{\epsilon}}(\alpha,\beta)}
   &=  \sum_{\mathsf{B}_{\boldsymbol{\epsilon,\epsilon'}}} 
    \mathcal{\widehat F}_{\mathsf{B}_{\boldsymbol{\epsilon,\epsilon'}}}(\{ \lambda_j\}_{j=1}^M,\{\xi_j^{(1)}\}_{j=1}^m | \beta)\
    \frac{\mathsf{c}_{Q,\alpha,\beta}^{(R)}}{\mathsf{c}_{\bar Q_{\mathsf{B}_{\boldsymbol{\epsilon,\epsilon'}}},\alpha,\beta}^{(R)}}\,
   \frac{\moy{Q\, |\, \bar Q_{\mathsf{B}_{\boldsymbol{\epsilon,\epsilon'}}}} }{\moy{Q\,|\,Q}}
   \nonumber\\
   &= \sum_{\mathsf{B}_{\boldsymbol{\epsilon,\epsilon'}}} 
    \mathcal{\widehat F}_{\mathsf{B}_{\boldsymbol{\epsilon,\epsilon'}}}(\{ \lambda_j\}_{j=1}^M,\{\xi_j^{(1)}\}_{j=1}^m | \beta)\
   \prod_{j=1}^M\frac{\ths(2\bar\lambda_j-\eta)}{\ths(2\lambda_j-\eta)}\,
     \frac{\moy{Q\, |\, \bar Q_{\mathsf{B}_{\boldsymbol{\epsilon,\epsilon'}}}} }{\moy{Q\,|\,Q}},
\end{align}
and we can now use the determinant representation obtained in \cite{NicT24} for the ratio of scalar products of separate states:
\begin{align}\label{ratio-det}
   \prod_{j=1}^M\frac{\ths(2\bar\lambda_j-\eta)}{\ths(2\lambda_j-\eta)}\,
   \frac{\moy{Q\, |\, \bar Q_{\mathsf{B}_{\boldsymbol{\epsilon,\epsilon'}}}} }{\moy{Q\,|\,Q}}
   =\frac{V(\lambda_1,\ldots,\lambda_M)}{V(\bar\lambda_1,\ldots,\bar\lambda_M)}\,\prod_{j=1}^M\frac{\ths(2\lambda_j+\eta)}{\ths(2\bar\lambda_j+\eta)}\,
   \frac{\det_M\big[\mathcal{S}_Q(\boldsymbol{\bar\lambda},\boldsymbol{\lambda})\big] }{\det_M\big[ \mathcal{S}_Q(\boldsymbol{\lambda},\boldsymbol{\lambda})\big] }.
\end{align}
Here $\mathcal{S}_Q(\boldsymbol{\bar\lambda},\boldsymbol{\lambda})$ is the $M\times M$ matrix with elements
\begin{align}\label{mat-SQ}
  \left[\mathcal{S}_Q(\boldsymbol{\bar\lambda},\boldsymbol{\lambda})\right]_{i,j}
  &= \mathbf{A}_{\boldsymbol{\varepsilon}}(\bar\lambda_i)\, Q(\bar\lambda_i-\eta) \big[ t(\bar\lambda_i+\lambda_j-\eta/2)-t(\bar\lambda_i-\lambda_j-\eta/2)\big]
  \nonumber\\
  &\qquad
  -\mathbf{A}_{\boldsymbol{\varepsilon}}(-\bar\lambda_i)\, Q(\bar\lambda_i+\eta)\big[ t(\bar\lambda_i+\lambda_j+\eta/2)-t(\bar\lambda_i-\lambda_j+\eta/2)\big],
\end{align}
with in particular
\begin{align}\label{mat-Gaudin}
 &\left. \left[\mathcal{S}_Q(\boldsymbol{\bar\lambda},\boldsymbol{\lambda})\right]_{i,j}  \right|_{\bar\lambda_i=\lambda_i}
  =  \mathbf{A}_{\boldsymbol{\varepsilon}}(-\lambda_i)\, Q(\lambda_i+\eta)
  \nonumber\\
  &\hspace{2cm}  
     \times
     \bigg[-\delta_{i,j}\,\frac\partial{\partial\mu}\left(\log\frac{\mathbf{A}_{\boldsymbol{\varepsilon}}(\mu)\, Q(\mu-\eta)}{\mathbf{A}_{\boldsymbol{\varepsilon}}(-\mu)\, Q(\mu+\eta)}\right)_{\mu=\lambda_i}
     + K(\lambda_i-\lambda_j)-K(\lambda_i+\lambda_j)\bigg].
\end{align}
In these expressions, the functions $t$ and $K$ are defined as
\begin{align}\label{def-t}
   &t(\lambda)=\frac{\ths'(\lambda-\frac\eta 2)}{\ths(\lambda-\frac\eta 2)}-\frac{\ths'(\lambda+\frac\eta 2)}{\ths(\lambda+\frac\eta 2)},
   \\
 \label{def-K}
   &K(\lambda)= t(\lambda+\eta/2)+t(\lambda-\eta/2)=\frac{\ths'(\lambda-\eta)}{\ths(\lambda-\eta)}-\frac{\ths'(\lambda+\eta)}{\ths(\lambda+\eta)}.
\end{align}
%
If we now decompose $\{\lambda_i \}_{1\le i\le M}$ and $\{\bar \lambda_i \}_{1\le i\le M}$ as
\begin{align}
   \{\lambda_i \}_{1\le i\le M}=\{\lambda_i \}_{i\in\mathsf{\Lambda}_-}\cup\{\lambda_i \}_{i\in\mathsf{\Lambda}_+},
   \qquad
   \{\bar \lambda_i \}_{1\le i\le M}=\{\lambda_i \}_{i\in \mathsf{\Lambda}_-}\cup\{\xi_{j_i}^{(1)}\}_{i\in\mathsf{\Lambda}_+},  
\end{align}
we can rewrite \eqref{ratio-det} as
\begin{multline}\label{ratio-det-2}
   \prod_{j=1}^M\frac{\ths(2\bar\lambda_j-\eta)}{\ths(2\lambda_j-\eta)}\,
   \frac{\moy{Q\, |\, \bar Q_{\mathsf{B}_{\boldsymbol{\epsilon,\epsilon'}}}} }{\moy{Q\,|\,Q}}
   =\prod_{\substack{a,b\in\mathsf{\Lambda}_+ \\ a<b}}\frac{\ths(\lambda_a-\lambda_b)\ths(\lambda_a+\lambda_b)}{\ths(\xi^{(1)}_{j_a}-\xi^{(1)}_{j_b})\ths(\xi^{(1)}_{j_a}+\xi^{(1)}_{j_b})}
   \prod_{\substack{a\in\mathsf{\Lambda}_- \\ b\in\mathsf{\Lambda}_+}}\frac{\ths(\lambda_a-\lambda_b)\ths(\lambda_a+\lambda_b)}{\ths(\lambda_a-\xi^{(1)}_{j_b})\ths(\lambda_a+\xi^{(1)}_{j_b})}
   \\
   \times \prod_{\substack{a\in\mathsf{\Lambda}_+}}\frac{\ths(2\lambda_a+\eta)\, \mathbf{A}_{\boldsymbol{\varepsilon}}(-\xi^{(1)}_{j_a})\, Q(\xi^{(1)}_{j_a}+\eta)}
   {\ths(2\xi^{(1)}_{j_a}+\eta)\, \mathbf{A}_{\boldsymbol{\varepsilon}}(-\lambda_a)\, Q(\lambda_a+\eta)}\,
   \frac{\det_M\big[\mathcal{M}(\boldsymbol{\bar\lambda},\boldsymbol{\lambda})\big] }{\det_M\big[ \mathcal{N}(\boldsymbol{\lambda})\big] },
\end{multline}
with
\begin{align}
   &\big[ \mathcal{N}(\boldsymbol{\lambda})\big]_{a,b}= \delta_{a,b}\, \frac\partial{\partial\mu}\left(\log\frac{\mathbf{A}_{\boldsymbol{\varepsilon}}(\mu)\, Q(\mu-\eta)}{\mathbf{A}_{\boldsymbol{\varepsilon}}(-\mu)\, Q(\mu+\eta)}\right)_{\mu=\lambda_a}
   +K(\lambda_a+\lambda_b)-K(\lambda_a-\lambda_b),
   \label{Mat-N}\\
   &\big[\mathcal{M}(\boldsymbol{\bar\lambda},\boldsymbol{\lambda})\big]_{a,b}
   =t(\xi_{j_a}^{(1)}+\lambda_b+\eta/2)-t(\xi_{j_a}^{(1)}-\lambda_b+\eta/2)
   =t(\xi_{j_a}+\lambda_b)-t(\xi_{j_a}-\lambda_b).
   \label{Mat-M}
\end{align}

Hence we can write for the finite chain:
%
\begin{theorem}\label{th-mat-el-finite}
Let us suppose that the constraint \eqref{const-TQ} is satisfied for some integer $M$ and some choice of signs \eqref{choice-eps+-}, with $\epsilon_{\alpha_1^+}\epsilon_{\alpha_2^+}\epsilon_{\alpha_3^+}=\epsilon_{\alpha_1^-}\epsilon_{\alpha_2^-}\epsilon_{\alpha_3^-}=1$. 
Let $Q$ be a solution of the form \eqref{Q-form} of the functional $TQ$-equation \eqref{hom-TQ}, and let $\ket{Q}$ and $\bra{Q}$ be the corresponding separate state of the form \eqref{eigen-r-Skl}, respectively of the form \eqref{eigen-l-Skl}, in the Vertex-IRF SoV basis in which $\alpha$ and $\beta$ satisfy \eqref{cond-c+=0}-\eqref{cond-b+=0} and \eqref{cond-ref-states}, and in which $\mathsf{g}_-$ is fixed as in \eqref{fix-g-}.

Then, for any $\boldsymbol{\epsilon},\boldsymbol{\epsilon'}\in \{1,2\}^{m}$ satisfying \eqref{cond-block}, the mean value \eqref{moy-Em} in the state $\ket{Q}$  of the tensor product of gauge transformed local operators $\barE_{m}^{\boldsymbol{\epsilon'},\boldsymbol{\epsilon}}(\alpha,\beta)$ \eqref{op-E_m}, where $a_n,b_n,\bar a_n,\bar b_n$ are expressed in terms of $\alpha,\beta,\boldsymbol{\epsilon},
\boldsymbol{\epsilon'}$ as in \eqref{Gauge.Basis-1}-\eqref{Gauge.Basis-2}, is given by 
\begin{equation}\label{corr-finite}
  \moy{\barE_{m}^{\boldsymbol{\epsilon'},\boldsymbol{\epsilon}}(\alpha,\beta)}
   =\sum_{\text{\textsc{b}}_{1}=1}^{M}\ldots
\sum_{\text{\textsc{b}}_{s}=1}^{M}\sum_{\text{\textsc{b}}_{s+1}=1}^{M+m}%
\ldots \sum_{\text{\textsc{b}}_{m}=1}^{M+m}
\frac{H_{\{\text{\textsc{b}}_{j}\}}(\{\lambda \}|\alpha,\beta )}
       {\prod\limits_{1\leq l<k\leq m}\!\!\!\!\theta(\xi _{k}-\xi _{l})\,\theta (\xi _{k}+\xi _{l})},
\end{equation}
with 
\begin{multline} \label{finite-H}
H_{\{\text{\textsc{b}}_{j}\}}(\{\lambda \}|\alpha,\beta )
=\prod_{n=1}^{m}\frac{\theta (\xi _{n}+\eta\alpha)}{\theta (\eta b_{n})\theta (\xi_{n}+\eta(\alpha+1))}
\sum\limits_{\sigma _{\text{\textsc{b}}_{j}}}
\frac{(-1)^{s}\prod\limits_{i=1}^{m}\sigma _{\text{\textsc{b}}_{i}}\prod\limits_{i=1}^{m}
         \prod\limits_{j=1}^{m}\theta (\lambda _{\text{\textsc{b}}_{i}}^{\sigma }+\xi _{j}+\eta /2)}
       {\prod\limits_{1\leq i<j\leq m}\theta (\lambda _{\text{\textsc{b}}_{i}}^{\sigma }-\lambda _{\text{\textsc{b}}_{j}}^{\sigma }-\eta )
            \theta (\lambda _{\text{\textsc{b}}_{i}}^{\sigma}+\lambda _{\text{\textsc{b}}_{j}}^{\sigma }+\eta )}  
            \\
\times \prod\limits_{p=1}^{s}\bigg\{\theta (\lambda _{\text{\textsc{b}}_{p}}^{\sigma }-\xi _{i_{p}}^{(1)}-\eta(1+b_{i_{p}}))
\prod\limits_{k=1}^{i_{p}-1}\theta (\lambda _{\text{\textsc{b}}_{p}}^{\sigma }-\xi _{k}^{(1)})\prod\limits_{k=i_{p}+1}^{m}\!\!\theta
(\lambda _{\text{\textsc{b}}_{p}}^{\sigma }-\xi _{k}^{\left( 0\right) })\bigg\} 
            \\
\times \!\!\prod\limits_{p=s+1}^{m}\bigg\{\theta (\lambda _{\text{\textsc{b}}%
_{p}\,}^{\sigma }-\xi _{i_{p}}^{(1)}+\eta (1-\bar{b}_{i_{p}}))\prod%
\limits_{k=1}^{i_{p}-1}\theta (\lambda _{\text{\textsc{b}}_{p}}^{\sigma
}-\xi _{k}^{(1)})\prod\limits_{k=i_{p}+1}^{m}\!\!\theta (\lambda _{\text{%
\textsc{b}}_{p}}^{\sigma }-\xi _{k}^{(1)}+\eta )\bigg\} \\
\times \prod\limits_{k=1}^{m}\prod\limits_{\substack{ l=1 \\ l\neq i_{+}}}^{3}\frac{\theta (\xi _{k}+\epsilon _{\alpha _{l}^{+}}\alpha _{l}^{+})}{\theta (\lambda _{\text{\textsc{b}}_{k}}^{\sigma }+\eta /2+\epsilon _{\alpha_{l}^{+}}\alpha _{l}^{+})}\ \det_{m}\Omega .
\end{multline}
The subscript $i_+$ in \eqref{finite-H} is fixed as in \eqref{signs+}-\eqref{alpha-beta+} (i.e. such that $\alpha\eta= -\epsilon_{\alpha^+_{i_+}}\alpha^+_{i_+}$),   we have moreover used the notations \eqref{i_p-s}-\eqref{i_p-s'}, and the sum is performed over all $\sigma _{\text{\textsc{b}}_{j}}\in \{+,-\}$ for \textsc{b}$_{j}\leq M$, and $\sigma _{\text{\textsc{b}}_{j}}=1$ for \textsc{b}$_{j}>M$.
Finally,  the $m\times m$ matrix $\Omega $ is given in terms of \eqref{Mat-N} and \eqref{Mat-M} by
\begin{alignat}{2}
& \Omega _{lk}=-\delta _{N+m+1-\textsc{b}_{l},k},\quad & & \text{for \textsc{b}}%
_{l}>M, \label{mat-reduced-1}\\
& \Omega _{lk}=\sum_{a=1}^{M}\left[ \mathcal{N}^{-1}\right] _{\text{\textsc{b}}_{l},a}%
\mathcal{M}_{a,k}, 
\quad & & \text{for \textsc{b}}_{l}\leq M.\label{mat-reduced}
\end{alignat}
\end{theorem}

\subsection{Elementary blocks in the half-infinite chain limit}

Let us now consider the half-infinite chain limit of the matrix elements \eqref{corr-finite} in which $\ket{Q}$ is the ground state of the chain. For simplicity, we restrict here our study to the domain in which
\begin{equation}\label{basic-domain}
   \omega,\eta\in i\mathbb{R},\qquad \text{with}\quad 0<i\eta <\pi\Im(\omega),
\end{equation}
for which we have in our notations
\begin{equation}\label{Jxyz-basicdomain}
    |J_x|<J_y<J_z.
\end{equation}
In this domain, real boundary fields can be obtained (modulo the quasi-periods $\pi$ and $\pi\omega$) from boundary parameters of the form 
\begin{equation}
   \alpha_1^\pm=i\beta_1^\pm,\qquad \alpha_2^\pm=i\beta_2^\pm+\tfrac\pi 2,\qquad \alpha_3^\pm=\beta_3^\pm\pm\tfrac{\pi\omega} 2,
   \qquad \text{with} \quad \beta_1^\pm,\beta_2^\pm,\beta_3^\pm\in \mathbb{R}.
\end{equation}
%
Note however that the other domains can be treated completely similarly.

As in our previous article concerning the XXZ chain \cite{NicT23}, we also restrict our study to configurations of the boundary fields for which the ground state is in the sector $M=\lfloor \tfrac N 2\rfloor -k$ (with $M$ satisfying the constraint \eqref{const-TQ} and $k$ remaining finite in the thermodynamic limit), and characterized in the homogeneous and thermodynamic limit by an infinite number of Bethe real roots which condensate on the interval $(0,\tfrac\pi 2)$ towards some density function $\rho(\lambda)$. The latter is solution of the integral equation
\begin{equation}\label{int-rho}
\rho (\lambda )+\frac{i}{2\pi}\int_{0}^{\tfrac\pi 2 }\big[K(\lambda -\mu )+K(\lambda+\mu )\big]\,\rho (\mu )\,d\mu 
=\frac{i t(\lambda )}{\pi},  
\end{equation}
in terms of the functions \eqref{def-t} and \eqref{def-K}. This integral equation can be extended by parity on the symmetric interval $(-\tfrac\pi 2,\tfrac\pi 2)$,
\begin{equation} \label{int-rho-ext}
\rho (\lambda )+\frac{i}{2\pi}\int_{\tfrac\pi 2}^{\tfrac\pi 2 } K(\lambda -\mu )\,\rho (\mu )\,d\mu 
=\frac{i t(\lambda )}{\pi}, 
\end{equation}
and can be solved by Fourier series, which gives
\begin{equation}\label{rho}
   \rho(\lambda) 
   = \frac{1}{ \pi} \sum_{k \in \mathbb{Z}} \frac{e^{2 i k \lambda}}{\cosh(k|\eta|)}
   =\frac{1}{\pi} 
   \frac{\theta'_1(0,q)}{\theta_2(0,q)}
   \frac{\theta_3(\lambda,q)}{\theta_4(\lambda,q)},
   \qquad q=e^{-i\eta},
\end{equation}
where $\theta_i(\lambda,q)$, $i\in\{1,2,3,4\}$, are the Theta functions of nome $q$. Note that, as in the XXZ case, this density function is just twice the density function for the periodic chain \cite{FadT79}.

As in the XXZ case \cite{KapS96,GriDT19}, and depending on the configuration of the boundary magnetic fields, the set of Bethe roots for the ground state may also contain some boundary roots $\check\lambda_{\ell}^\pm$, i.e. some isolated complex Bethe roots which tend, in the large $N$ limit, to one of the zeroes or poles of the boundary factor in the Bethe equations:
\begin{equation}\label{BR}
    \check\lambda_{\ell}^\pm 
    = -\tfrac\eta 2-\epsilon _{\alpha_\ell^\pm}\alpha _\ell^\pm + \check{\epsilon}_{\ell}^\pm, 
\end{equation}
with $\check{\epsilon}_{\ell}^\pm$ being exponentially small in $N$.
These boundary roots are of importance for the computation of the correlation functions since one single boundary root may provide contributions of order $1$ in the half-infinite chain limit, see \cite{KitKMNST07,GriDT19,NicT22,NicT23} and below.

Taking the half-infinite chain limit $N\rightarrow \infty $ of the expression \eqref{corr-finite} is then completely similar as in the XXZ case \cite{KitKMNST07,NicT22,NicT23}.
In this limit, the ratio of determinants can be computed as in \cite{KitKMNST07,NicT22,NicT23}. 
For a real root $\lambda _j$, the element $[\mathcal{N}(\boldsymbol{\lambda})]_{j,k}$ of the matrix \eqref{Mat-N} behaves as
\begin{equation}\label{mat-N-thermo}
  [\mathcal{N}(\boldsymbol{\lambda})]_{j,k} 
  =2i\pi N\,\delta _{j,k}\,\Big[ \rho (\lambda _{j})+O\Big(\frac{1}{N}\Big)\Big]
  +\big[K(\lambda_{j}+\lambda _{k})-K(\lambda _{j}-\lambda _{k})\big],
\end{equation}
so that
\begin{equation}
\sum_{\substack{ p=1 \\ \lambda _{p}\in (0,\frac\pi 2 )}}^{M}\mathcal{N}_{j,p}\,%
\frac{\rho (\lambda _{p}-\xi _{k})-\rho (\lambda _{p}+\xi _{k})}{2N\rho
(\lambda _{p})}\underset{N\rightarrow \infty }{\longrightarrow }[t(\xi
_{k}+\lambda _{j})-t(\xi _{k}-\lambda _{j})],
\end{equation}
in which we have used the integral equation \eqref{int-rho-ext}.
If instead $\lambda_j=\check\lambda_{\ell}^\sigma$ ($\sigma=\pm$) is a boundary root of the form \eqref{BR}\footnote{For simplicity, we consider here only the generic case in which the boundary parameters $\epsilon_ {\alpha_i^\pm}\alpha _i^\pm$ are pairwise distinct. In cases in which two (or more) of these boundary parameters coincide, \eqref{mat-N-thermo-BR} becomes
\begin{equation} \label{mat-N-thermo-BR-mult}
   [\mathcal{N}(\boldsymbol{\lambda})]_{j,k}=\frac{\check{n}_\ell^\sigma}{\check{\epsilon}_\ell^{\sigma }}\,\left[ \delta _{j,k}+O(N\check{\epsilon}_{\sigma })\right] ,
\end{equation}
in which $\check{n}_\ell^\sigma$ is the number of boundary roots coinciding with $\epsilon_ {\alpha_\ell^\sigma}\alpha_\ell^\sigma$, and 
\begin{equation}
\Omega_{a,b}\underset{N\rightarrow \infty }{\sim }
 i\pi \,\frac{\check{\epsilon}_\ell^\sigma }{\check{n}_\ell^\sigma}
       \big[\rho (\lambda_{\text{\textsc{b}}_a}-\xi _{j_{b}})-\rho (\lambda _{\text{\textsc{b}}_a}+\xi _{j_{b}})\big]
\end{equation}
if $\lambda _{\text{\textsc{b}}_a}=\check\lambda_{\ell}^\sigma$.
},
\begin{equation} \label{mat-N-thermo-BR}
   [\mathcal{N}(\boldsymbol{\lambda})]_{j,k}
   =-\frac{1}{\check{\epsilon}_\ell^{\sigma }}\,\left[ \delta _{j,k}+O(N\check{\epsilon}_{\sigma })\right] .
\end{equation}
%
Therefore, the elements $\Omega_{a,b}$ of the matrix \eqref{mat-reduced} for $\text{\textsc{b}}_a\leq M$ behaves when $N\to \infty$ as
\begin{equation}
\Omega_{a,b}\underset{N\rightarrow \infty }{\sim } \mathcal{R}(\lambda _{\text{\textsc{b}}_a},\xi _{j_{b}}),
\end{equation}
in which
\begin{equation}
\mathcal{R}(\lambda _{\text{\textsc{b}}_a},\xi _{j_{b}})=
\begin{cases}
\displaystyle\frac{\rho (\lambda _{\text{\textsc{b}}_a}-\xi _{j_{b}})-\rho (\lambda _{\text{\textsc{b}}_a}+\xi_{j_{b}})}{2N\rho (\lambda _{\text{\textsc{b}}_a})}
\quad  & \text{if }\lambda _{\text{\textsc{b}}_a}\in(0,\tfrac\pi 2 ),
\vspace{1mm} \\ 
\displaystyle -i\pi \,\check{\epsilon}_\ell^\sigma 
       \big[\rho (\lambda_{\text{\textsc{b}}_a}-\xi _{j_{b}})-\rho (\lambda _{\text{\textsc{b}}_a}+\xi _{j_{b}})\big]
       \quad  & \text{if }\lambda _{\text{\textsc{b}}_a}=\check\lambda_{\ell}^\sigma.
\end{cases}
\end{equation}
Note that we have
\begin{equation}\label{sym-det}
   \sigma_{\text{\textsc{b}}_a}\mathcal{R}(\lambda _{\text{\textsc{b}}_a},\xi _{j_{b}})=\mathcal{R}(\lambda^\sigma_{\text{\textsc{b}}_a},\xi _{j_{b}}).
\end{equation}

 In the half-infinite chain limit, the sum over real Bethe roots in \eqref{corr-finite} (i.e. the sum over the indices $\textsc{b}_{j}\le M$ with $\lambda_{\textsc{b}_{j}}$ real) becomes integrals over the density function according to the rule
%
\begin{equation}
    \frac{1}{N}\sum_{\substack{\text{\textsc{b}}_n=1 \\ \lambda_{\text{\textsc{b}}_n}\in (0,\frac\pi 2)}}^M \sum_{\sigma_{\text{\textsc{b}}_n}=\pm} 
    f(\lambda_{\text{\textsc{b}}_n}^{\sigma })
    \underset{N\rightarrow \infty }{\longrightarrow }
    \int_0^{\frac\pi 2 }\rho (\lambda)\sum_{\sigma ={\pm }} 
    f(\lambda ^{\sigma })\, d\lambda 
    =\int_{-\frac\pi 2}^{\frac\pi 2}f(\lambda )\,\rho (\lambda )\,d\lambda  .
\end{equation}
Moreover, noticing that
\begin{equation}
2i\pi \,\text{Res}\,\rho (\lambda -\xi )_{\,\vrule height13ptdepth1pt\>{\lambda =\xi -\eta /2}\!}=2,
\end{equation}
we can write the sum over the indices $\textsc{b}_{j}>M$ in \eqref{corr-finite} as contour integrals around the points $\xi_n^{(1)}$ (with index $-1$ due to the sign in \eqref{mat-reduced-1}).
Finally, if one of the roots $\lambda_{\text{\textsc{b}}_j}$ is a boundary root of the form $\check\lambda_{\ell}^\sigma$ \eqref{BR}, the corresponding row in the matrix $\Omega$ is proportional to $\check{\epsilon}_\ell^\sigma$ which is exponentially small in $N$. This means that the corresponding term in the sum in \eqref{corr-finite} do not contribute to the thermodynamic limit, except if this vanishing in $\check{\epsilon}_\ell^\sigma$ is compensated by a prefactor behaving in $1/\check\epsilon_{\ell}^\sigma$. This is the case if $(\ell,\sigma)\in\{(k,+)\}_{k\not= i_+}$, due to the pole in $-\epsilon_{\alpha_\ell^+}\alpha_\ell^+ -\eta /2$ ($\ell\not= i_+$) in the last line of the prefactor \eqref{finite-H}:
\begin{equation}
\frac{1}{\theta (\check{\lambda}_+ +\epsilon_{\alpha_\ell^+}\alpha_\ell^+ +\eta /2)}
\underset{N\rightarrow \infty }{\sim }
\frac{1}{\check\epsilon_{\ell}^+\ths'(0)}.
\end{equation}
In that case, the final contribution is of order 1, and can be written as a contour integral around the point $-\epsilon_{\alpha_\ell^+}\alpha_\ell^+ -\eta /2$  ($\ell\not= i_+$) with an index $-1$.

Hence, the ground state matrix elements of the quasi-local operators \eqref{op-E_m} satisfying the
condition \eqref{cond-block} can be written in the thermodynamic limit as 
\begin{multline}\label{result-thermo}
   \moy{\barE_{m}^{\boldsymbol{\epsilon'},\boldsymbol{\epsilon}}(\alpha,\beta)}
   =\prod_{n=1}^{m}\frac{\theta (\xi _{n}+\eta\alpha)}{\theta (\eta b_{n})\theta (\xi_{n}+\eta(1+\alpha))}\,
   \frac{(-1)^{s}}{\prod\limits_{j<i}\ths (\xi _{i}-\xi _{j})\prod\limits_{i\leq j}\ths(\xi _{i}+\xi _{j})}   
   \\
   \times 
   \int_{\mathcal{C}}\prod_{j=1}^{s}d\lambda _{j}\ \int_{\mathcal{C}_{\boldsymbol{\xi}}}\prod_{j=s+1}^{m}\!\!d\lambda _{j}\ 
   H_{m}(\{\lambda_n\}_{n=1}^{M};\{\xi _{k}\}_{k=1}^{m})\ \det_{1\leq j,k\leq m}\big[\Phi
(\lambda _{j},\xi _{k})\big],
\end{multline}
in which
\begin{equation}\label{mat-Phi}
\Phi (\lambda _{j},\xi _{k})=\frac 12 \big[\rho (\lambda _{j}-\xi_{k})-\rho (\lambda _{j}+\xi _{k})\big],  
\end{equation}
and 
\begin{multline}\label{H-thermo}
H_{m}(\{\lambda _n\}_{n=1}^{M};\{\xi _{k}\}_{k=1}^{m})
=\frac{\prod\limits_{j=1}^{m}\prod\limits_{k=1}^{m}\ths (\lambda _{j}+\xi_{k}+\eta /2)}
   {\!\!\!\!\prod\limits_{1\leq i<j\leq m}\!\!\!\!\ths (\lambda_{i}-\lambda _{j}-\eta )\,\ths (\lambda _{i}+\lambda _{j}+\eta )}
 \\
\times \prod\limits_{p=1}^{s}\bigg\{\theta (\lambda _{p}-\xi
_{i_{p}}^{(1)}-\eta (1+b_{i_{p}}))\prod\limits_{k=1}^{i_{p}-1}\theta
(\lambda _{p}-\xi _{k}^{(1)})\prod\limits_{k=i_{p}+1}^{m}\!\!\theta (\lambda
_{p}-\xi _{k}^{(1)}-\eta )\bigg\} \\
\times \!\!\prod\limits_{p=s+1}^{m}\bigg\{\theta (\lambda _{p}-\xi
_{i_{p}}^{(1)}+\eta (1-\bar{b}_{i_{p}}))\prod\limits_{k=1}^{i_{p}-1}\theta
(\lambda _{p}-\xi _{k}^{(1)})\prod\limits_{k=i_{p}+1}^{m}\!\!\theta (\lambda
_{p}-\xi _{k}^{(1)}+\eta )\bigg\} \\
\times 
 \prod\limits_{k=1}^{m}\prod\limits_{\substack{ \ell=1 \\ \ell\neq i_+}}^3
 \frac{\theta (\xi_k+\epsilon_{\alpha_\ell^+}\alpha_\ell^+)}{\theta (\lambda_k+\eta /2+\epsilon _{\alpha_\ell^+}\alpha _\ell^+)}.
\end{multline}
The integration contours are defined in the following way. Let $\Lambda_{\text{BR}}$ be the sub-set of complex boundary roots contained in the set of Bethe roots for the ground state, let $\bar\Lambda_{\text{BR}}$ the set of their corresponding limiting values in the thermodynamic limit, and let $\bar\Lambda_{\text{BR}^+,i_+}=\bar\Lambda_{\text{BR}}\cap\{-\eta /2-\epsilon _{\alpha_\ell^+}\alpha _\ell^+\}_{\ell\not= i_+}$. Then
\begin{equation}
   \mathcal{C}=[-\tfrac\pi 2,\tfrac\pi 2] \cup \Gamma^-(\bar\Lambda_{\text{BR}^+,i_+}),
\end{equation}
in which $\Gamma^-(\bar\Lambda_{\text{BR}^+,i_+})$ encircles, with index $-1$, all the elements of $\bar\Lambda_{\text{BR}^+,i_+}$ (and no other poles), whereas
\begin{equation}\label{C-xi}
   \mathcal{C}_{\boldsymbol{\xi}} =\mathcal{C}\cup \Gamma^- (\xi_1^{(1)},\ldots,\xi_m^{(1)}),
\end{equation}
where again $\Gamma^- (\xi_1^{(1)},\ldots,\xi_m^{(1)})$ denotes a contour surrounding the points of $\xi_1^{(1)},\ldots,\xi_m^{(1)}$  with index $-1$, all other poles of the integrand being outside.

\section{Conclusion}
\label{sec-conclusion}

In this paper, we have tackled the longstanding open problem of computing correlation functions for the XYZ quantum spin chain from first principle, i.e. starting from its integrable characterization on the finite lattice. More precisely, we have computed some of the elementary building blocks for these correlation  functions at zero-temperature.

This has been done for the open chain with boundary fields, imposing one single constraint on the six boundary parameters. We have thus extended to the XYZ case our approach and results previously obtained for the XXX/XXZ open spin chains under similar types of boundary conditions. The main motivations to impose this boundary constraint are also the same as in the XXZ chain. First, it allows for some of the eigenvalues and eigenstates of the transfer matrix to be described by solutions to usual Bethe equations, this description being conjectured to be complete through the combination of two related sectors  \cite{NepR03,NepR04,YanZ06}. For our purpose of computing zero-temperature correlation functions, we in fact only need the ground state to be described in such a way, and we suppose that it is in a sector close to half-filling. Second, the original scalar products of separate states in the SoV framework admit a reformulation in terms of generalized Slavnov’s determinants if one of the separate states is an eigenstate associated to a solution of these Bethe equations \cite{NicT24}. Finally, as for the XXZ case \cite{NicT22,NicT23}, we had to restrict ourselves to elementary building blocks involving local operators whose action on generalized Bethe states produce linear combinations of similar types of states, i.e. with the same number of gauged $B$-operators. This is due to the fact that only these types of states can easily be rewritten as separate states on the original SoV basis, which allows us to use our scalar product formula to compute matrix elements of local operators. 

An interesting feature of our multiple integral representation for the correlation functions is that, once the basis of local operators is appropriately chosen, the results for the XXX, XXZ and XYZ case have a universal structure. In particular, the integrand factorizes in a size $m$-determinant for a $m$-sites local operator, which has the same form when expressed in terms of the thermodynamic density of the ground state Bethe roots for these three models, times a function that has the same factorized form when expressed in terms of rational, hyperbolic and elliptic functions, respectively, for the XXX, XXZ and XYZ case. This is naturally expected if one considers that, starting from the XYZ model, one should obtain the XXZ and XXX models under special limits. The other interesting point is that these integrands only explicitly depend on two of the boundary parameters at the first site of the chain. Hence, all the dependence on the other boundary parameters is hidden in the solutions of the Bethe equations used to describe the ground state, which, in our results, is reflected in the type of contours we should choose for the multiple integrals in the thermodynamic limit.  

Several interesting questions remain to be studied. 

One of them is the exact description of the ground state, and in particular of the different types of complex Bethe roots it involves,  according to the different configurations of the boundary fields. This problem could in principle be addressed by enlarging the detailed analytical study of the ground state thermodynamic limit performed in \cite{GriDT19} in the diaonal XXZ case.
It would also be desirable to have an analytical proof of the completeness conjecture of \cite{NepR03,NepR04,YanZ06}.

 Another important question is the computation of the missing elementary building blocks, i.e. those involving local operators that change the number of $B$-operators when acting on generalized Bethe states. This is a limitation that we already encountered in the XXZ case, and that we should therefore try to overcome first in this case. From a technical point of view, this may require an extension of our current scalar product formula. Some insight on the results may be achieved by the consideration of the simpler XXX case, for which no such limitations on local operators appear \cite{Nic21}.


Finally, it would of course be very interesting to extend the computation of correlation functions to the case of completely general and unconstrained boundary fields. This remains a challenging open problem even for the XXX and XXZ spin chains. The question in this case is to have a description of the ground state which would be appropriate for the consideration of the scalar products and correlation functions  in the thermodynamic limit. Since a description of the spectrum in terms of usual $TQ$-equation is still missing in this case, one could possibly try to build on the solutions of the so-called inhomogeneous $TQ$-equations \cite{CaoYSW13a,CaoYSW13b,CaoYSW14}, which in the XYZ case requires in fact the use of two $Q$-functions. The fact that the SoV method applies for the most general boundary conditions allows to prove the completeness of such characterizations \cite{KitMN14}, but it does not seem to lead to a convenient description of the ground state in terms of well-organized Bethe roots in the thermodynamic limit. Alternatively, one could possibly try to rely on the new description of the spectrum proposed more recently in \cite{QiaCYSW21,XinCYW24}. The latter is however very different from the usual description in terms of Bethe roots, and it is still a completely open question on whether it could be used for the computation of correlation functions.

\section*{Acknowledgments}

G. N. is supported by CNRS and Laboratoire de Physique, ENS-Lyon. 
V. T. is supported by CNRS.

\appendix

\section{Useful properties of the bulk gauged Yang-Baxter generators}
\label{app-bulk-gauge}

We gather here some properties of the elements of the bulk gauge monodromy matrix \eqref{def-Mgauge},
\begin{align}\label{M-gauge}
  M(\lambda|(\alpha,\beta),(\gamma,\delta))
  &=S^{-1}(-\lambda-\eta/2|\alpha,\beta)\, M(\lambda)\, S(-\lambda-\eta/2|\gamma,\delta)
  \nonumber\\
  &=\frac{1}{\det S(-\lambda-\eta/2|\alpha,\beta)}
       \begin{pmatrix}
       A(\lambda|\alpha-\beta,\gamma+\delta) & B(\lambda|\alpha-\beta,\gamma-\delta)\\
       C(\lambda|\alpha+\beta,\gamma+\delta) & D(\lambda|\alpha+\beta,\gamma-\delta)
       \end{pmatrix}
       \nonumber\\
  &=\frac{1}{\det S(-\lambda-\eta/2|\alpha,\beta)}    
       \begin{pmatrix}
       \Theta (\lambda |\alpha -\beta ,\gamma +\delta ) & \Theta (\lambda |\alpha-\beta ,\gamma -\delta ) \\ 
      -\Theta (\lambda |\alpha +\beta ,\gamma +\delta ) & -\Theta (\lambda |\alpha+\beta ,\gamma -\delta )
\end{pmatrix},
\end{align}
in which
\begin{equation}\label{def-op-Theta}
    \Theta (\lambda |x,y)=\bar Y_x(-\lambda-\eta/2)\, M(\lambda)\, Y_y(-\lambda-\eta/2).
\end{equation}
%

\subsection{Commutation relations}

We have, for any values of the spectral parameters $\lambda$ and $\mu$ and of the gauge parameters $x,y,z$:
\begin{equation}\label{Comm-Theta1}
\Theta (\lambda |x+1,y-1)\, \Theta (\mu |x,y)
=\Theta (\mu |x+1,y-1)\, \Theta(\lambda |x,y),
\end{equation}
and 
\begin{multline}\label{Comm-Theta2}
\Theta (\lambda |x+1,y+1)\, \Theta (\mu |x,z)
 =\frac{\ths (\eta \,\frac{y-z}{2})\, \ths (\lambda -\mu +\eta )}{\ths (\eta (\frac{y-z}{2}-1))\,\ths (\lambda -\mu )}\,
\Theta (\mu |x+1,z+1)\, \Theta (\lambda |x,y)  
 \\
 -\frac{\ths (\eta) \, \ths (\lambda -\mu +\frac{y-z}{2}\eta )}{\ths (\eta (\frac{y-z}{2}-1))\, \ths (\lambda-\mu )}\,
\Theta (\lambda |x+1,z+1)\, \Theta (\mu|x,y),
\end{multline}
%
%
\begin{multline}\label{Comm-Theta3}
\Theta (\lambda |x-1,y-1)\, \Theta (\mu |z,y)
 =\frac{\ths (\eta \frac{x-z}{2})\, \ths (\lambda -\mu -\eta )}{\ths (\eta (\frac{x-z}{2}+1))\, \ths (\lambda-\mu )}\,
\Theta (\mu |z-1,y-1)\, \Theta (\lambda |x,y)   \\
 +\frac{\ths( \eta) \, \ths (\lambda -\mu+\frac{x-z}{2}\eta )}{ \ths (\eta (\frac{x-z}{2}+1))\, \ths (\lambda-\mu )}\, 
\Theta (\lambda |z-1,y-1)\,\Theta (\mu|x,y),
\end{multline}
which can be rewritten more explicitly in terms of the $A,B,C,D$ elements of \eqref{M-gauge} as
\begin{alignat}{2}
  &A(\lambda |x +1,y -1)\, A(\mu |x ,y )  &&=A(\mu | x +1,y -1)\, A(\lambda |x ,y ) ,  
  \label{Comm-AA}\\
  &B(\lambda |x +1,y -1)\, B(\mu |x ,y )  &&=B(\mu |x +1,y -1)\,  B(\lambda |x ,y ) ,
  \label{Comm-BB}\\
  &C(\lambda |x +1,y -1)\,  C(\mu |x , y )  &&=C(\mu |x +1,y -1)\,  C(\lambda |x ,y ) ,
  \label{Comm-CC}\\
  &D(\lambda |x +1,y -1)\,  D(\mu |x,y )   &&=D(\mu |x +1,y -1)\,  D(\lambda |x ,y ).
  \label{Comm-DD}
\end{alignat}
and
\begin{multline}\label{Comm-AB}
   A(\lambda |x+1,y+1)\, B(\mu |x,z)
   =\frac{\ths(\eta\frac{y-z}2)\,\ths(\lambda-\mu+\eta)}{\ths(\eta(\frac{y-z}2-1))\,\ths(\lambda-\mu)}\,
   B(\mu|x+1,z+1)\, A(\lambda |x,y)
   \\
   -\frac{\ths(\eta)\,\ths(\lambda-\mu+\eta\frac{y-z}2)}{\ths(\eta(\frac{y-z}2-1))\,\ths(\lambda-\mu)}\,
   B(\lambda |x+1,z+1)\, A(\mu |x,y),
\end{multline}
\begin{multline}\label{Comm-DB}
   D(\lambda|x-1,y-1)\, B(\mu|z,y)
   =\frac{\ths(\lambda-\mu-\eta)\,\ths(\eta\frac{x-z}2)}{\ths(\lambda-\mu)\,\ths(\eta(\frac{x-z}2+1))}\,
   B(\mu|z-1,y-1)\, D(\lambda|x,y)
   \\
   +\frac{\ths(\eta)\,\ths(\lambda-\mu+\eta\frac{x-z}2)}{\ths(\eta(\frac{x-z}2+1))\,\ths(\lambda-\mu)}\,
   B(\lambda|z-1,y-1)\, D(\mu|x,y).
\end{multline}

\subsection{Action on the reference state $\ket{\eta,x}$}

\begin{align}
   &\Theta (\lambda |x,x-N)\, \ket{\eta ,x} =0,
   \label{actTheta-ref1}\\
   &\Theta (\lambda |y,x-N)\, \ket{\eta ,x}
    = \bar Y_y(-\lambda-\eta/2)\, Y_x(-\lambda-\eta/2)\, d(\lambda)\,  \ket{\eta ,x-1}, 
    \label{actTheta-ref2}\\
    &\Theta (\lambda |x,y)\, \ket{\eta ,x}
     = -  \bar Y_{x-N}(-\lambda-\eta/2)\, Y_y(-\lambda-\eta/2)\, a(\lambda)\,    \ket{\eta ,x+1}.
     \label{actTheta-ref3}
\end{align}
i.e.
\begin{align}   
   &C(\lambda|x,x-N)\, \ket{\eta,x}=0, \label{actC-ref}\\
   &A(\lambda|y,x-N)\, \ket{\eta,x}
   =\bar Y_{y}(-\lambda-\eta/2)\, Y_{x}(-\lambda-\eta/2)\, d(\lambda)\,\ket{\eta,x-1},\label{actA-ref}\\
   &D(\lambda|x,y)\, \ket{\eta,x}
   =-  \bar Y_{x-N}(-\lambda-\eta/2)\, Y_{y}(-\lambda-\eta/2)\, a(\lambda)\,\ket{\eta,x+1}.\label{actD-ref}
\end{align}
%

\subsection{Action on a Bethe state}

The action of the gauged bulk monodromy operators on the generalised bulk Bethe states takes the following form:
\begin{multline}\label{actA-Bethe}
  A(\lambda_{M+1} |x ,y +M)\ \underline{B}_M(\{\lambda_j\}_{j=1}^M|x -1,z +1)\,  \ket{\eta ,y+N}
  \\
    =   \sum_{a=1}^{M+1} d(\lambda_a)\,
    \bar Y_{x-M}(-\lambda_a-\eta/2)\, Y_{y+N}(-\lambda_a-\eta/2)\,
  \frac{\ths(\lambda_{M+1}-\lambda_a+\eta\frac{y-z-2+M}{2})}{\ths (\eta\frac{y-z-2-M}{2}) }\,
  \\  \times
  \frac{\prod_{b=1}^M \ths (\lambda _{a}-\lambda _{b}+\eta)}{\prod_{\substack{b=1 \\ b\neq a}}^{M+1} \ths (\lambda _{a}-\lambda _{b})} \
  \underline{B}_M ( \{\lambda_j\}_{j=1}^{M+1}\setminus\{\lambda_a\}| x, z +2)\,
  \ket{\eta ,y +N-1} , 
\end{multline}
with
\begin{equation}\label{coeff-act-A}
   \bar Y_{x-M}(-\lambda_a-\eta/2)\, Y_{y+N}(-\lambda_a-\eta/2)
   =\ths(\lambda_a+\tfrac{x+y+N+1-M}2\eta)\ths(\tfrac{x-y-N-M}2\eta),
\end{equation}
and
\begin{multline}\label{actD-Bethe}
  D(\lambda_{M+1} | y-M,z)\
  \underline{B}_M(\{\lambda_j\}_{j=1}^M |x -1,z +1)\,
  \ket{\eta ,y}
  \\
  =- \sum_{a=1}^{M+1} a(\lambda_a)\,
    \bar Y_{y-N}(-\lambda_a-\eta/2)\, Y_{z+M}(-\lambda_a-\eta/2)\,
  \frac{\ths(\lambda_{M+1}-\lambda_a+\eta\frac{y-x-M+2}{2})}{\ths(\eta\frac{y-x+M+2}{2})}\,
  \\
  \times \frac{\prod_{j=1}^M\ths(\lambda_a-\lambda_j-\eta)}{\prod_{\substack{j=1 \\ j\neq a}}^{M+1}\ths(\lambda_a-\lambda_j)} \
  \underline{B}_M(\{\lambda_j\}_{j=1}^{M+1}\setminus\{\lambda_a\}| x -2,z)\,
  \ket{\eta ,y+1} ,
\end{multline}
with
\begin{equation}\label{coeff-act-D}
   \bar Y_{y-N}(-\lambda_a-\eta/2)\, Y_{z+M}(-\lambda_a-\eta/2)
   =\ths(\lambda_a+\tfrac{y+z-N+1+M}2\eta)\ths(\tfrac{y-z-N-M}2\eta).
\end{equation}
Note that the expression of the actions \eqref{actA-Bethe} and \eqref{actD-Bethe} nearly formally coincides with those used in the trigonometric case (see \cite{NicT22}) up to the substitution of the $\sinh$ functions by $\ths$ functions, except for the coefficients  \eqref{coeff-act-A} and \eqref{coeff-act-D} which are issued from  Vertex-IRF matrix.

\subsection{Reconstruction of local operators}

As in \cite{NicT22}, we have the following reconstruction formulas for the operators \eqref{gauge-op-n} in terms of the matrix elements of the bulk monodromy matrix \eqref{def-Mgauge}:
\begin{align}
  E_{n}^{i,j}(\xi _{a}|(\alpha ,\beta ),(\gamma ,\delta ))
  &=\prod_{k=1}^{n-1}t(\xi _k-\eta /2)\
  M_{j,i}(\xi _{n}-\eta /2|(\alpha ,\beta ),(\gamma ,\delta ))\,
  \prod_{k=1}^{n} \big[ t(\xi_k-\eta/2)\big]^{-1},
  \label{reconstr-1}
\end{align}
and
\begin{align}
  E_{n}^{i,j}(\xi _{a}|(\gamma ,\delta ),(\alpha ,\beta ))
    &=(-1)^N \prod_{k=1}^n t(\xi_k-\eta /2)\ 
      \frac{\hat{M}_{j,i}(-\eta /2-\xi _{n}|(\alpha ,\beta),(\gamma ,\delta ))}{\det_q M(\xi_n)}\,
      \prod_{k=1}^{n-1} \big[ t(\xi_k-\eta/2)\big]^{-1}
  \nonumber\\
  &=(-1)^{i-j}\,\frac{\det S(-\xi_n|\alpha,\beta)}{\det S(-\xi_n|\gamma,\delta)}\,  
       \prod_{k=1}^n t(\xi_k-\eta /2)\nonumber\\
   &\quad\times
       \frac{M_{3-i,3-j}(\xi_n+\eta/2|(\alpha-1,\beta),(\gamma-1,\delta))}{\det_q M(\xi_n)}\,
      \prod_{k=1}^{n-1} \big[ t(\xi_k-\eta/2)\big]^{-1},
  \label{reconstr-2}
\end{align}
where $t(\lambda )=\tr[M(\lambda )],$ is the bulk transfer matrix.
These expressions can be derived from the solution of the eight-vertex bulk inverse problem \cite{MaiT00} similarly as in Proposition~4.1 of \cite{NicT22}.

\subsection{Other useful identities}

These reconstructions formulas can be used to derive the following identities for the gauged monodromy matrix elements (see Corollary~4.1 of \cite{NicT22})
\begin{align}
  &M_{\epsilon ,\epsilon ^{\prime }}(\xi _{n}-\eta /2|(\alpha ,\beta ),(\gamma,\delta ))\,
  M_{\epsilon,\bar{\epsilon}^{\prime }}(\xi _{n}+\eta/2|(\alpha -1,\beta ),(\gamma' -1,\delta' )) =0, 
  \label{MM=0-1}\\
  &M_{\epsilon ^{\prime },\epsilon }(\xi _{n}+\eta /2|(\gamma -1,\delta),(\alpha -1,\beta ))\,
  M_{\bar{\epsilon}^{\prime },\epsilon}(\xi_{n}-\eta /2|(\gamma' ,\delta' ),(\alpha ,\beta )) =0,
  \label{MM=0-2}
\end{align}
and
\begin{align}
  &M_{\epsilon,\epsilon'}(\xi _{n}-\eta /2|(\alpha ,\beta),(\gamma ,\delta ))\,
  M_{3-\epsilon,\bar\epsilon'}(\xi_n+\eta/2|(\alpha-1,\beta),(\gamma'-1,\delta'))
  =(-1)^{\bar\epsilon-\epsilon}  \frac{\det S(-\xi_n|\alpha',\beta')}{\det S(-\xi_n|\alpha,\beta)}
  \nonumber\\
  &\hspace{2cm}\times
  M_{\bar\epsilon,\epsilon'}(\xi _{n}-\eta /2|(\alpha' ,\beta'),(\gamma ,\delta ))\,
  M_{3-\bar\epsilon,\bar\epsilon'}(\xi_n+\eta/2|(\alpha'-1,\beta'),(\gamma'-1,\delta')),
  \label{MM=MM-1}
  \\
  &
  M_{\epsilon',3-\epsilon}(\xi_n+\eta/2|(\gamma-1,\delta),(\alpha-1,\beta))\,
  M_{\bar\epsilon',\epsilon}(\xi_n-\eta/2|(\gamma',\delta'),(\alpha,\beta))
  =(-1)^{\bar\epsilon-\epsilon}\,\frac{\det S(-\xi_n|\alpha,\beta)}{\det S(-\xi_n|\alpha',\beta')}\,
  \nonumber\\
  &\hspace{2cm}\times
  M_{\epsilon',3-\bar\epsilon}(\xi_n+\eta/2|(\gamma-1,\delta),(\alpha'-1,\beta'))\,
  M_{\bar\epsilon',\bar\epsilon}(\xi_n-\eta/2|(\gamma',\delta'),(\alpha',\beta')),
  \label{MM=MM-2}
\end{align}
which are valid for any $\epsilon,\bar\epsilon,\epsilon',\bar\epsilon'\in\{1,2\}$, and any choice of parameters $\alpha,\beta,\alpha',\beta',\gamma,\delta,\gamma',\delta'$ such that the corresponding Vertex-IRF matrices are invertible.

In particular, we have the following useful identity: 
\begin{multline}\label{Alternative-q-det}
 C(\xi _{n}-\eta /2|(\alpha ,\beta),(\gamma ,\delta ))\,
  B(\xi_n+\eta/2|(\alpha-1,\beta),(\gamma'-1,\delta'))
  \\  
  = - \frac{\det S(-\xi_n|\alpha',\beta')}{\det S(-\xi_n|\alpha,\beta)}\,
  A(\xi _{n}-\eta /2|(\alpha' ,\beta'),(\gamma ,\delta ))\,
  D(\xi_n+\eta/2|(\alpha'-1,\beta'),(\gamma'-1,\delta')).
\end{multline}

The cancellation identities \eqref{MM=0-1}, \eqref{MM=0-2} can also be extended to the case in which the two operators belong to a larger product of operators: the product of operators
\begin{equation}\label{cancel-prod}
   \prod_{k=n\to m}\hspace{-2mm} M_{\epsilon_k,\epsilon'_k}(\xi_k-\eta/2|(\gamma_k,\delta_k),(\alpha_k,\beta_k))\,
   \prod_{k=m\to n}\hspace{-2mm} M_{\bar\epsilon_k,\bar\epsilon'_k}(\xi_k+\eta/2|(\gamma_k-1,\delta_k),(\alpha'_k-1,\beta'_k))
\end{equation}
vanishes as soon as there exists some $k\in\{n,\ldots,m\}$ such that $\epsilon_k=\bar\epsilon_k$.

\section{Trigonometric limit of correlation functions}
\label{app-trigo}

Here, we consider the trigonometric limit $\Im\omega\to +\infty$, and explain how, in this limit, one can recover the results of \cite{NicT23} concerning the elementary building block of correlation functions in the open XXZ spin chain.

This trigonometric limit has already been studied in Appendix A of \cite{NicT24}: there, we have shown that, in this limit,  the scalar products of the separate states in the XYZ open spin chain coincide with  those of the XXZ open spin chain \cite{KitMNT18}. We shall use here the same notations as in Appendix A of \cite{NicT24}.
There, we have implemented the following identifications of the boundary and crossing parameters 
\begin{equation}\label{bparam-trigo}
   \alpha_1^\pm=-i\varphi_\mp,\qquad 
   \alpha_2^\pm=i\psi_\mp -\epsilon\frac\pi 2,\qquad 
   \alpha_3^\pm=i\tau_\mp+\epsilon\frac\pi 2+\frac{\pi\omega}2, \qquad
\eta =i\tilde{\eta},
\end{equation}
for some arbitrary sign $\epsilon$,
so that, when $\Im\omega\to +\infty$ with all other parameters remaining finite, the XYZ Hamiltonian \eqref{Ham} tends to the XXZ one:
\begin{equation}\label{Ham-XXZ}
    H_\text{XXZ} =\sum_{n=1}^{N-1} \Big[ \sigma _n^{x} \sigma _{n+1}^{x}
    +\sigma_n^{y}\sigma _{n+1}^{y}
    +\cosh\tilde\eta \, \sigma _n^{z} \sigma _{n+1}^{z}\Big]
    +\sum_{a\in\{x,y,z\}}\Big[\tilde h_-^a\,\sigma_1^a+\tilde h_+^a\sigma_N^a\Big],
\end{equation}
with
\begin{align}
  &\tilde h_\pm^x 
                  =\frac{\sinh\tilde \eta\, \cosh\tau_\pm}{\sinh\varphi_\pm\,\cosh\psi_\pm},\qquad
  \tilde h_\pm^y 
                  =i\frac{i \sinh\tilde \eta\,\sinh\tau_\pm}{\sinh\varphi_\pm\,\cosh\psi_\pm},\qquad
  \tilde h_\pm^z 
                  =\sinh\tilde \eta\, \coth\varphi_\pm\,\tanh\psi_\pm.
\end{align}
In this trigonometric limit,  the 8-vertex R-matrix \eqref{R-mat}, the 8-vertex boundary K-matrices \eqref{mat-K} and so the 8-vertex transfer matrix all degenerate (up to some possible irrelevant prefactors) to the corresponding ones in the
6-vertex case, with XXZ parameters $\tilde{\eta},\varphi _{\pm },\psi_{\pm },\tau _{\pm }$ and spectral parameter $u=-i\lambda$ (we refer the reader to \cite{NicT24} for more details). 
Moreover, under the choice of signs $\epsilon_3^+=-\epsilon_3^-$, 
$\epsilon_1^\pm=\epsilon_{\varphi_\mp}$, $\epsilon_2^\pm=\pm\epsilon_{\psi_\mp}$, 
the constraint \eqref{const-TQ} can be written in terms of \eqref{bparam-trigo} as
\begin{equation}\label{const-XXZ}
   \epsilon_3^+(\tau_+-\tau_-)+\sum_{\sigma=\pm}\left[ \epsilon_{\varphi_\sigma}\varphi_\sigma+\sigma\epsilon_{\psi_\sigma}(\psi_\sigma+i\epsilon\frac\pi 2)\right]+(N-2M-1)\tilde\eta=0,
\end{equation}
with $\epsilon_{\varphi_+}\epsilon_{\varphi_-}\epsilon_{\psi_+}\epsilon_{\psi_-}=1$, and one recovers in the trigonometric limit $\Im\omega\to +\infty$ the formulation of the $TQ$-equation in the XXZ case as written  in \cite{NicT23}.

To prove that our expressions \eqref{corr-finite} and \eqref{result-thermo} for correlation functions degenerate into their XXZ analog (5.7) and (6.8) of \cite{NicT23} in the trigonometric limit, we then have to study the trigonometric limit of both sides of \eqref{corr-finite} and \eqref{result-thermo}.
It means that we have to show, under the adequate choice of the gauge parameters $\alpha$ and $\beta$, that on the one hand the elements of the local basis of operators defined in \eqref{op-E_m} for the XYZ case  converge to the corresponding elements of the XXZ basis as defined in (4.2) of \cite{NicT23}, and that on the other hand the expression of their corresponding matrix elements (i.e. the right hand sides of \eqref{corr-finite} and \eqref{result-thermo}) converge to those obtained in the XXZ case in \cite{NicT23}.

In order to show the first assertion, let us recall the definition of the Vertex-IRF matrix in the 8-vertex and 6-vertex cases respectively:
\begin{align}
&S_{(\eta )}^{(8V)}(\lambda |\alpha ,\beta )=%
\begin{pmatrix}
\theta _{2}(\lambda -(\alpha +\beta )\eta |2\omega ) & \theta _{2}(\lambda
-(\alpha -\beta )\eta |2\omega ) \\ 
\theta _{3}(\lambda -(\alpha +\beta )\eta |2\omega ) & \theta _{3}(\lambda
-(\alpha -\beta )\eta |2\omega )%
\end{pmatrix},
\\
&S_{(\tilde{\eta})}^{(6V)}(u|\alpha ,\beta )=%
\begin{pmatrix}
e^{u-\tilde{\eta}(\alpha +\beta )} & e^{u-\tilde{\eta}(\alpha -\beta )} \\ 
1 & 1%
\end{pmatrix}.
\end{align}
Then, the following identity holds (with $u=-i\lambda$, $\tilde\eta=-i\eta$, $\alpha,\beta$ remaining finite):
\begin{equation}
\lim_{\Im \omega \rightarrow +\infty }S_{(\eta )}^{(8V)}(\lambda |\alpha -\frac{\pi \omega }{2\eta},\beta )
=S_{(\tilde{\eta})}^{(6V)}(u|\alpha ,\beta ),
\qquad
\end{equation}
 as a trivial consequence the following limits:
\begin{equation}
\lim_{\Im \omega \rightarrow +\infty }\thd _{2}(x+\pi \omega /2)=e^{-ix},\qquad
\lim_{\Im \omega \rightarrow +\infty }\thd_{3}(x+\pi \omega /2)=1.
\end{equation}
Let us now fix the 8-vertex gauge parameters, particularizing the condition \eqref{alpha-beta+} to the following choice (with $\epsilon_3^+=1$):
\begin{equation}
\eta \alpha ^{(8V)}=-\alpha _{3}^{+}=i\tilde{\eta}\alpha ^{(6V)}-\frac{\pi\omega }{2},\qquad 
\eta \beta ^{(8V)}=-\epsilon_1^+(\alpha _{1}^{+}+\alpha _{2}^{+})=i\tilde{\eta}\beta ^{(6V)}.
\end{equation}
Then, the value for the 6-vertex gauge parameters read 
\begin{equation}\label{gauge-6V}
\tilde{\eta}\alpha ^{(6V)}=-\tau _{-}+i\epsilon\frac{\pi }{2},\qquad 
\tilde{\eta}\beta ^{(6V)}=\epsilon_{\varphi_-}(\varphi _{-}-\psi _{-}-i\epsilon\frac{\pi }{2}),
\end{equation}
once we impose the boundary rewriting \eqref{bparam-trigo} with $\epsilon_1^+=\epsilon_{\varphi_-}$.
If we choose moreover here $\epsilon=-\epsilon_{\varphi_-}$,  then \eqref{gauge-6V} exactly coincides with (2.43)-(2.44) of \cite{NicT23} when fixing there\footnote{which is a necessary condition to impose there for the Vertex-IRF matrix to be invertible.} $\epsilon _{\varphi_-}=\epsilon _{-}$.
Hence, under this choice of signs,
\begin{equation}
\lim_{\Im \omega \rightarrow +\infty }S_{(\eta )}^{(8V)}(\lambda |\alpha
^{(8V)},\beta ^{(8V)})=S_{(\tilde{\eta})}^{(6V)}(u|\alpha ^{(6V)},\beta
^{(8V)}),
\end{equation}
and then we obtain the following identity for the trigonometric limit of the local
operators:
\begin{equation}
\lim_{\Im \omega \rightarrow +\infty }\left[ E_{n}^{\epsilon _{n}^{\prime
},\epsilon _{n}}(\xi _{n}|(a_{n},b_{n}),(\bar{a}_{n},\bar{b}_{n}))\right]
^{\left( 8V\right) }=\left[ E_{n}^{\epsilon _{n}^{\prime },\epsilon
_{n}}(\xi _{n}|(a_{n},b_{n}),(\bar{a}_{n},\bar{b}_{n}))\right] ^{\left(
6V\right) },  \label{Trigo-Op-gauged}
\end{equation}
where the local operator on the left hand side is the 8-vertex one, defined
by \eqref{gauge-op-n} in terms of the 8-vertex Vertex-IRF matrix with the
parameters $(a_{n},b_{n})$ and $(\bar{a}_{n},\bar{b}_{n})$ defined by %
\eqref{Gauge.Basis-1} and \eqref{Gauge.Basis-2}, under the choice:%
\begin{equation}
(a,b)=(\alpha ^{(8V)},\beta ^{(8V)}),
\end{equation}
and the local operator on the right hand side is the 6-vertex one, also
defined by \eqref{gauge-op-n} but now in terms of the 6-vertex $S$-matrix
and with the parameters $(a_{n},b_{n})$ and $(\bar{a}_{n},\bar{b}_{n})$
defined by \eqref{Gauge.Basis-1} and \eqref{Gauge.Basis-2}, but under the
choice:
\begin{equation}
(a,b)=(\alpha ^{(6V)},\beta ^{(6V)}).
\end{equation}
This means that the local operators
generated by this trigonometric limit $\Im \omega \rightarrow +\infty $
exactly coincides with those defined in (4.1), (4.3) and (4.4) of \cite%
{NicT23}.

It is now easy to show that, with these conventions, the trigonometric limit of the expressions \eqref{corr-finite} and \eqref{result-thermo} coincide with  the results (5.7)
and (6.8) of \cite{NicT23}, respectively\footnote{Note that, from our choice of domain \eqref{basic-domain}-\eqref{Jxyz-basicdomain},  the trigonometric limit of \eqref{result-thermo} naturally gives the correlation functions of the XXZ chain in the antiferromagnetic regime.}.

\bibliographystyle{SciPost_bibstyle}
\bibliography{/Users/vterras/Documents/Dropbox/Bib_files/biblio.bib}

\end{document}